%% file: quantum-covert-communication.tex
\documentclass[%
reprint,
 amsmath,amssymb,
 aps,
pra,
]{revtex4-1}
\usepackage[english]{babel}

\usepackage{stmaryrd,acronym,amsthm,dsfont,graphicx,enumitem}
\setlist[itemize]{leftmargin=*}
\setlist[enumerate]{leftmargin=*}

\newcommand{\pr}[1]{\left({#1}\right)}
\newcommand{\ket}[1]{|#1\rangle}
\newcommand{\bra}[1]{\langle #1 |}
\newcommand{\braket}[2]{\langle #1 | #2 \rangle}

\newtheorem{theorem}{Theorem}
\newtheorem{lemma}{Lemma}
\newtheorem{definition}{Definition}

\input{arcom_notation.tex}
\input{arcom_functions.tex}

\begin{document}

\title{A framework for covert and secret key expansion over quantum channels}
\author{Mehrdad Tahmasbi}
\thanks{Author to whom correspondence should be addressed. Email: mtahmasbi3@gatech.edu} 
\affiliation{Georgia Institute of Technology}

\author{Matthieu R. Bloch}
\thanks{Email: matthieu.bloch@ece.gatech.edu} 
\affiliation{Georgia Institute of Technology}


\begin{abstract}
Covert and secret quantum key distribution aims at generating information-theoretically secret bits between distant legitimate parties in a manner that remains provably undetectable by an adversary. We propose a framework in which to precisely define and analyze such an operation, and we show that covert and secret key expansion is possible. For fixed and known classical-quantum channels, we develop and analyze protocols based on forward and reverse reconciliation. When the adversary applies the same quantum channel independently on each transmitted quantum state, akin to a collective attack in the quantum key distribution literature, we propose a protocol that achieves covert and secret key expansion under mild restrictions. The crux of our approach is the use of information reconciliation and privacy amplification techniques that are able to process the sparse signals required for covert operation and  whose Shannon entropy scales as the square root of their length. 
In particular, our results show that the coordination  required between legitimate parties to achieve covert communication can be achieved with a negligible number of secret key bits.
\end{abstract}

\pacs{Valid PACS appear here}

\maketitle

\section{Introduction}
\label{sec:introduction}


Securing communications has become an essential requirement in modern communication systems. Secrecy, i.e, the ability to prevent unauthorized parties from extracting the information content of a signal, is typically enforced using conventional computationally-secure encryption although Quantum Key Distribution (QKD) remains to date the only approach to unconditional secrecy~\cite{Gisin2002,Scarani2009}. Another desirable feature of secure communications is covertness, i.e., the ability to hide the presence of communication signals from an unauthorized party and provably avoid detection~\cite{Bash2015}. While secrecy has been largely explored for quantum communications both theoretically and experimentally, the mechanisms required to achieve covertness are still much less understood.

Covertness, also referred to as low probability of detection, is conceptually related to classical and quantum steganography~\cite{Cachin2004,Ker2007,Shaw2011,Sanguinetti2016}, by which legitimate parties embed a message into a covertext then disclosed to an adversary \footnote{In~\cite{Shaw2011}, secrecy and covertness are referred to as security and secrecy, respectively. We adopt here a different terminology more in line with common usage in information-theoretic security.}. In many quantum steganography protocols, an innocent quantum state, in the form a codeword from a quantum error-control code, is used as the cover to embed another quantum state. The embedding is performed to simulate the transmission of an innocent state through a noisy channel and relies on shared secret keys with well characterized rates. A crucial assumption in these quantum steganography protocols is that the true physical channel is better than what the adversary expects. In covert communications, however, the role of the covertext is played by the communication channel, which introduces noise and imperfections that are outside the control of and only statistically known to the transmitter. 
There has been a recent surge of interest for covert communications, which has led to the discovery of a ``square-root law'' similar to that in steganography~\cite{Ker2007} in both classical~\cite{Bash2013,Wang2016b,Bloch2016a} and quantum settings~\cite{Bash2015a,Wang2016c,Sheikholeslami2016, bradler2016absolutely}. The square-root law, according to which the number of covert bits can only scale with the square-root of the number of channel uses, has also been experimentally validated in an optical test-bed~\cite{Bash2015a}. The authors of \cite{Bash2015a} also showed that, for a bosonic channel, covert communication is impossible without sources of imperfection in the adversary's observations as the detection of a single photon would indicate with certainty the existence of the communication. The possibility of quantum covert and secret key generation was recently  explored~\cite{Arrazola2016,Arrazola2017,Liu2017} but has led to the rather pessimistic conclusion that ``\emph{covert QKD consumes more secret bits than it can generate}''~\cite{Arrazola2016}.   

Our main contribution is to offer a more nuanced and optimistic perspective and show that covert and secret key expansion is actually possible over quantum channels. The intuition behind our approach is the following. In layman's terms, the covertness constraint requires the number of qubit transmissions to scale as $O(\sqrt{T})$ for $T$ channel uses~\cite{Bash2015a}. A crucial characteristic of earlier works~\cite{Bash2015a,Arrazola2016} is that the scaling is ensured by having the legitimate parties \emph{coordinate} the sparse transmission of $\sqrt{T}$ qubits in channel uses chosen secretly and uniformly at random out of $T$. Unfortunately, the secret key size required to select these secret channel uses scales as $\Omega(\sqrt{T}\log T)$ and necessarily exceeds the number of covert bits that one can hope to obtain, which scales as $\Omega(\sqrt{T})$. In contrast, we introduce more sophisticated coding schemes for information reconciliation and privacy amplification that do not require such coordination and are able to directly process the sparse and diffuse statistical information content of covert signals. The protocols that we present do not yet offer the secrecy levels of state-of-the art~QKD against coherent attacks but already achieve covert and secret key expansion and might pave the way to more broadly applicable protocols. 
 

Our results are developed in three steps as follows. We first lay out a precise model for quantum covert and secret key generation that captures a wide range of attacks by the adversary and protocols for legitimate parties, along with quantifiable metrics to assess the performance of a covert and secret key generation protocol over quantum channels. The main distinction with previous models~\cite{Arrazola2016,Arrazola2017,Liu2017} is the inclusion of the public communication required for information reconciliation in the analysis; specifically, since an adversary may devise a hypothesis test for detection based on all its observations, the probability distribution of the public communication has to be considered jointly with the quantum measurements in evaluating covertness. We then proceed to analyze an instance of quantum covert and secret key generation in which the classical-quantum channels are fixed and known,  for which we can define and analyze the covert and secret key capacity. We lower-bound the covert and secret key capacity by developing coding schemes using both forward and reverse reconciliation. The forward reconciliation scheme can be constructed by a suitable modification of established protocols for quantum covert communication~\cite{Sheikholeslami2016} to guarantee secrecy. In contrast, the reverse reconciliation scheme requires a new approach because of technical challenges precluding the direct use of well-known results on information reconciliation and privacy amplification for the sparse distribution needed for covert communication. Finally, we consider an instance of quantum covert and secret key generation in which the classical-quantum channel is fixed but under the control of the adversary and unknown to the legitimate users. Under some conditions to limit the power of the adversary, which we precisely characterize, we prove the existence of covert and secret key generation protocols consisting of a channel estimation phase followed by a key-generation phase. The estimation phase is based on a covert quantum tomography protocol that estimates the required parameters of the channel and the key-generation phase is based on universal results for covert quantum communication. While covertness cannot be unconditionally guaranteed, our protocol offers the legitimate parties with the ability to successfully abort before engaging in key generation. We do not instantiate explicit codes but recent progress in designing codes for covert communications~\cite{Kadampot2018} suggests that the protocols described here can be implemented with low-complexity.

\section{Notation}
\label{sec:notation}

We briefly introduce the notation used throughout the paper. For a finite-dimensional Hilbert space  $\calH$, $\dim \calH$ denotes the dimension of $\calH$, and $\calL(\calH)$ denotes the space of all linear operators from $\calH$ to $\calH$. We denote the adjoint of an operator $X \in \calL(\calH)$ by $X^\dagger$, and call $X$ Hermitian if $X=X^\dagger$. $X\in \calL(\calH)$ is positive (non-negative) semi-definite, if it is Hermitian and all of its eigenvalues are positive (non-negative). $\calD(\calH)$ denotes the set of all density operators on $\calH$, i.e., all non-negative operators with unit trace. For $X, Y\in \calL(\calH)$, we write $X\succ Y$ ($X\succeq Y$), if $X-Y$ is positive (non-negative) semi-definite.  For $X\in\calH$, let $\sigma_{\min}(X)$ and $\sigma_{\max}(X)$ denote the  minimum and the maximum singular value of $X$, respectively, and if $X$ is Hermitian, let $\lambda_{\min}(X)$ and $\lambda_{\max}(X)$ denote the  minimum and maximum eigenvalue of $X$. Furthermore, we define norms of $X\in\calL(\calH)$ as $\|X\|_1 \eqdef \tr{\sqrt{X^\dagger X}}$ and $\|X\|_2 \eqdef \sqrt{\tr{X^\dagger X}}$. For a Hermitian operator $X\in\calL(\calH)$ with eigen-decomposition $X=\sum_x x\ket{x}\bra{x}$, we define the projection $\{X\succeq 0\} \eqdef \sum_{x\geq 0} \ket{x}\bra{x}$. A quantum channel $\calE_{A\to B}$ is a completely positive and trace preserving linear map from $\calL(\calH^A)$ to $\calL(\calH^B)$. An isomorphic extension of $\calE_{A\to B}$, $U_{A\to BE}$, satisfies $\calE_{A\to B}(\rho^A) = \textnormal{tr}_E(U_{A\to BE}\rho^AU_{A\to BE}^\dagger)$ for all $\rho^A\in\calD(\calH^A)$. We denote the complementary channel of $\calE_{A\to B}$ by $\calE^\dagger_{A\to B}(\rho^A)\eqdef \calE_{A\to E}(\rho^A)\eqdef \textnormal{tr}_B(U_{A\to BE}\rho^AU_{A\to BE}^\dagger)$, which is well-defined and unique up to a unitary transformation \cite{wilde2013quantum}. A classical-quantum (cq)-channel is a map from an abstract set $\calX$ to $\calD(\calH)$, denoted by $x\mapsto \rho_x$.

 For $\rho^A \in \calD(\calH^A)$ we define von Neumann entropy $H(\rho^A) \eqdef \avgH{A}_\rho \eqdef -\tr{\rho^A\log \rho^A}$. For $\rho^{AB}\in\calD(\calH^A\otimes \calH^B)$, we define conditional von Neumann entropy $\avgH{A|B}_\rho \eqdef H(\rho^{AB}) - H(\rho^B)$ where $\rho^B\eqdef \text{tr}_A(\rho^{AB})$, and quantum mutual information $\avgI{A;B}_\rho\eqdef H(\rho^A)+H(\rho^B) - H(\rho^{AB})$. Similarly, we define conditional quantum  mutual information$\avgI{A;B|C} \eqdef H(\rho^{AC}) + H(\rho^{BC}) - H(\rho^{ABC}) - H(\rho^C)$ for any $\rho^{ABC} \in \calD(\calH^A\otimes \calH^B\otimes \calH^C)$. If $P_X$ is a distribution on $\calX$ and $x\mapsto \rho_x$ is a cq-channel, we denote the Holevo information by
\begin{align}
I(P_X, \rho_x) \eqdef H\pr{\sum_x P_X(x) \rho_x} - \sum_x P_X(x) H(\rho_x).
\end{align} 
 For $\rho, \sigma\in \calD(\calH)$, the quantum relative entropy is
 \begin{align}
 \D{\rho}{\sigma} \eqdef \begin{cases} \tr{\rho\pr{\log \rho - \log \sigma}} &\text{if }\, \text{supp} (\rho) \subset \text{supp} (\sigma),\\
 \infty&\text{otherwise},\end{cases}
\end{align}  
and the $\chi_2$ distance is
\begin{align}
 \chi_2\pr{\rho\|\sigma} \eqdef \begin{cases} \tr{\rho^2\sigma^{-1}}-1 &\text{if }\, \text{supp} (\rho) \subset \text{supp} (\sigma),\\
 \infty&\text{otherwise}.\end{cases}
\end{align}

\section{Framework for covert and secret key generation over classical-quantum channels}
\label{sec:prob-form}

\begin{figure}
  \centering
  \includegraphics[width=\linewidth]{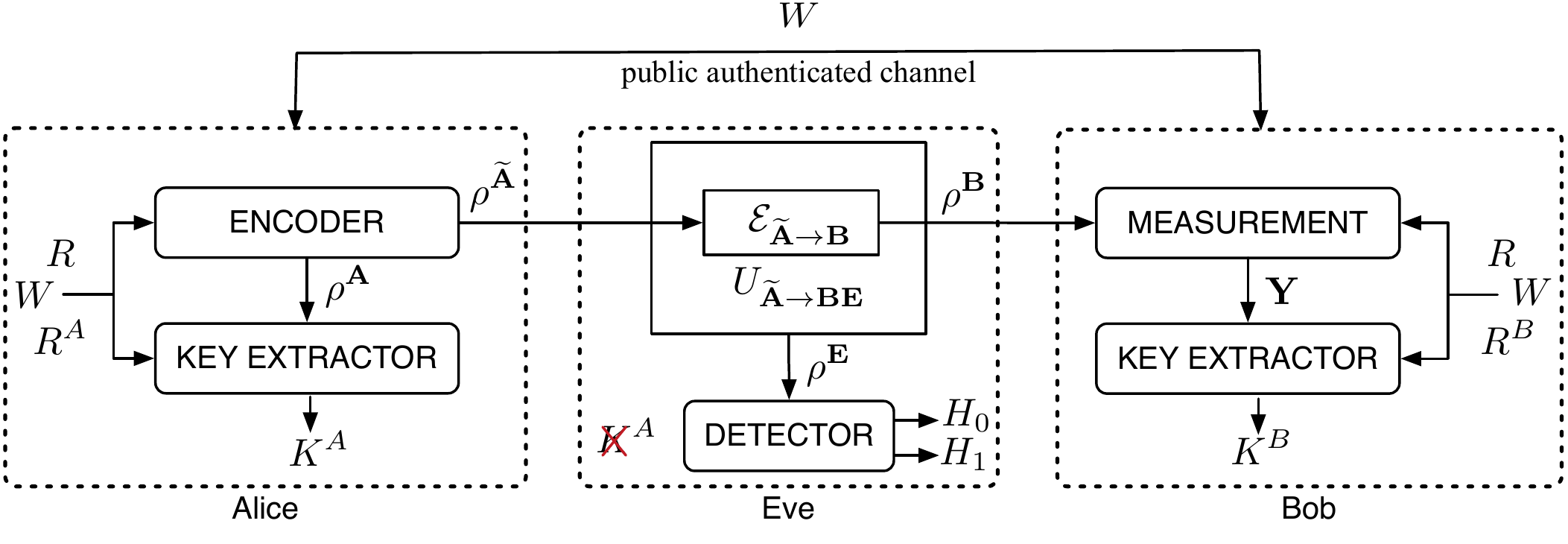}
  \caption{Model of covert  and secret key expansion}
  \label{fig:model}
\end{figure}
As illustrated in Fig.~\ref{fig:model}, we consider a setting in which two legitimate parties, Alice and Bob, desire to share a secret key while avoiding detection from an adversary, Eve, by exploiting one-way quantum channel and a two-way classical authenticated public channel of unlimited capacity. Specifically, in an entanglement-based representation, over  $T$ time steps, Alice prepares a classical-quantum state $\rho^{A\widetilde{A}}$, possibly depending on public communications, on a bipartite system described by a Hilbert space $\calH^{A} \otimes \calH^{\widetilde{A}}$ and sends the sub-system $\widetilde{A}$ to Bob. We assume that for $\calX\subset \mathbb{R}$, $\{\ket{x}^A\}_{x\in\calX}$ is an orthonormal basis for $\calH^A$,  all eigenvectors of $\rho^A$ are always in $\{\ket{x}\}_{x\in\calX}$, and for any $x\in\calX$, the conditional state $\rho_x^{\widetilde{A}}$ is fixed. For simplicity, we restrict our attention to a two-dimensional $\calH^A$, i.e., $\calX = \{0, 1\}$, in which $0$ represents an ``innocent'' symbol, corresponding to the absence of communication, while $1$ represents an `non-innocent'' symbol. We further assume that the ``start'' ($t=1$) and ``stop'' ($t=T$) times of the protocol are known to all parties and obtained through other modalities, e.g., GPS signals. Eve  expects the product state $(\rho_0^{\widetilde{A}})^{\proddist T}$ when there is no communication and may modify the states according to a quantum channel.  We denote the entire state received by Bob and acting on the product Hilbert space $(\calH^{B})^{\proddist T}$ by $\rho^{\mathbf{B}}$.

 For the purpose of covert communications, we need to distinguish protocols based on the type of Eve's attacks. In the most general case, Eve implements a \emph{coherent attack} described by a  quantum channel 
 \begin{align}\calE_{\widetilde{\mathbf{A}}\to \mathbf{B}}:\calL\pr{\pr{\calH^{\widetilde{A}}}^{\proddist T}} \to \calL\pr{\pr{\calH^B}^{\proddist T}},
 \end{align}
with isomorphic extension $U_{\mathbf{\widetilde{A}}\to \mathbf{B}\mathbf{E}}$, in which Bob receives $\rho^{\mathbf{B}} = \calE_{\widetilde{\mathbf{A}}\to \mathbf{B}}(\rho^{\widetilde{\mathbf{A}}})$ at the end of the transmission, and therefore, no useful public communication can happen during the transmission. Note that this has no impact on QKD since no useful information is shared until the end of the protocol. However, aborting the protocol in the middle could be crucial to be undetectable.
 A less powerful Eve can only implement \emph{collective attacks} described by quantum channels of the form $\calE_{\widetilde{\mathbf{A}}\to \mathbf{B}}= \calE_{\widetilde{A}\to B}^{\proddist T}$, i.e., Eve applies the same channel independently to each state transmitted by Alice. In this case, we can assume that Bob receives each state before Alice transmits the next state, which allows meaningful public communication during the transmission between Alice and Bob. Throughout the paper, we consider two scenarios  for collective attacks based on Alice's and Bob's knowledge about Eve's attack. First, when Alice and Bob have exact knowledge of the attack, we define an effective cq-channel $x \mapsto \rho_x^{BE}$, with marginal cq-channels  $x\to\rho_x^B$ and $x\to\rho_x^E$ from Alice to Bob and Eve, respectively. Second, when Eve's channel $\calE_{\widetilde{A}\to B}$ is unknown, we still consider effective cq-channels $x \mapsto \rho_x^B$ and $x\mapsto\rho_x^E$ where $\rho_x^B$ and $\rho_x^E$ are defined as $\calE_{A\to B}(\rho_x^{\widetilde{A}})$ and $\calE_{A\to B}^\dagger(\rho^{\widetilde{A}})$, respectively, and are unknown to both Alice and Bob. Our choice of $\rho_x^E$ accounts for the maximum amount of information that Eve can possibly gain, i.e., the state corresponding to a reference system for an isomorphic extension of the channel from Alice to Bob. Finally, Alice and Bob have access to independent local sources of randomness, denoted by $R^A \in \calR^A$ and $R^B\in\calR^B$, respectively, as well as a source of secret key $R\in\calR$.


For simplicity, we describe the protocols with only reverse public communication, but extension  to the general case in which forward public communication is also allowed would be possible. A protocol for key generation operates in $T$ time steps as follows. Alice and Bob draw realizations $r^A$, $r^B$, and $r$ of their local and common randomness. Subsequently, in every state $t\in\intseq{1}{T}$:
\begin{itemize}
\item Alice prepares a classical-quantum state $\rho^{A\widetilde{A}}$ as explained earlier using her local randomness $r^A$, the common randomness $r$, as well as past public messages from Bob denoted $(w_1,\cdots,w_{t-1})$ and sends $\rho^{\widetilde{A}}$ to Bob through the channel controlled by Eve;
\item Bob performs a quantum measurement on his available quantum state to obtain a classical measurement $y_t\in\calY\subset\mathbb{R}$;
\item Bob sends a message $W_t\in\calW_t$ over the public channel using his local randomness $r_B$, the common randomness $r$, as well as past measurements $y^{t-1}$. The choice of alphabet $\calW_t$ is part of the protocol design.
\end{itemize}
At the end of time step $T$, \emph{when no further public communication happens}, Eve performs a measurement on her state $\rho^{\mathbf{E}}$, as an attempt to detect the communication and obtain information about the secret key, while Alice and Bob use all their available information and randomness to compute two long binary strings $s^X$ and $s^Y$, respectively, as well as the number of bits $\ell^X$ and $\ell^Y$, respectively, to use as a secret key. The length of $s^X$ and $s^Y$ is public and fixed at the beginning of the protocol. Alice finally sets her key $k^X$ to be the first $\ell^X$ bits of $s^X$ while Bob sets his key $k^Y$ to be the first $\ell^Y$ bits of $s^Y$.

A protocol is called an $(\epsilon,\delta,\mu)$-protocol if the following properties hold. Let $W$, $S^X$, $S^Y$, $K^X$, $K^Y$, be the random variables representing the total public communication, Alice's random string, Bob's random string, Alice's key, and Bob's key, respectively. We require:
\begin{itemize}
\item $\epsilon$-reliability: $P_e\eqdef \P{K^X\neq K^Y}\leq \epsilon$, which implicitly includes the condition $\ell^X=\ell^Y$;
\item $\delta$-secrecy: $S \eqdef \D{\rho^{\mathbf{E}WS^X}}{\rho^{\mathbf{E}W} \otimes \rho^{S^X}_{\text{unif}}} \leq \delta$, where $\rho^{\mathbf{E}WS^X}$ is the joint density matrix of the eavesdropper's observations, public messages and Alice's random string, and $\rho^{S^X}_{\text{unif}} $ is a mixed state for $S^X$ corresponding to a uniform distribution;
  \item $\mu$-covertness: $C \eqdef  \D{\rho^{\mathbf{E}W}}{\pr{\rho^{\mathbf{E}}_0}\otimes \rho^{W}_{\text{unif}}}\leq\mu$, where $\rho^{\mathbf{E}}_0$ is the density matrix of the eavesdropper's observations when no communication takes place and $\rho^{W}_{\text{unif}}$ is a mixed state for $W$ corresponding to a uniform distribution on $\times _t \calW_t$.
\end{itemize}
A protocol is \emph{efficient} if it allows key expansion so that the number of key bits created exceeds the number of common randomness bits consumed. Our goal is to analyze under what conditions  efficient $(\epsilon,\mu,\delta)$-protocols might exist.

A couple of remarks are in order regarding our protocol definition. Note that the choice of
  the key length is a part of the protocol. However, $\delta$-secrecy requires the string $S^X$ to be secret and not just
  $K^X$. This is merely enforced for technical reasons, so that the
  relative entropy is a deterministic quantity irrespective of the
  length of the key. Since $\epsilon$-reliability only applies to the
  bits of $K^X$, Alice can always generate the remaining bits of $S^X$
  independently and uniformly at random using her local randomness, so
  that our definition does not incur any loss of generality.
 By convention, we assume that the public communication is not by itself a proof of communication. Instead, $\mu$-covertness only requires that the public bits look uniformly distributed and do not reveal communication on the quantum channel.
 We point out that $\delta$-secrecy and $\mu$-covertness are  ``one-shot'' guarantees, in the sense that they only ensure a low probability of detection for a single execution of the protocol. In fact, by repeating the protocol $k$ consecutive and independent times, a $(\epsilon,\delta,\mu)$-protocol gives rise to a $(k\epsilon,k\delta,k\mu)$-protocol. Additional post-processing can reduce the constant $k\epsilon$ and $k\delta$ but cannot affect the constant $k\mu$. This suggests that the protocol should be designed for small values of $\mu$ and large values of $T$. Finally, the particular choice of the quantum state $\rho_{\text{unif}}^W$ in the definition of covertness plays no role in our proofs. As long as there exists a specific state corresponding to no communication for the public communication, our proof holds and leads to  a covert and secret key generation scheme.

\section{Covert and secret key generation over known cq-channel}
\label{sec:quantum_passive_results}

We first address the situation in which the cq-channels are \emph{fixed} and known ahead of time, and in which the adversary is \emph{passive}. In this special case, the length of the key can be computed ahead of time, and there is no need to distinguish between the random strings $S^X$ and $S^Y$ and the keys $K^X$ and $K^Y$. Furthermore, it becomes possible to define a notion of covert and secret key capacity as follows. 
A throughput $\Theta$ is achievable if there exists a sequence of $(\epsilon_T, \delta_T, \mu_T)$-protocols generating  $\ell_T$ bits of secret key while consuming  $r_T$ bits of secret key over $T$ stages and such that
\begin{align}
  \lim_{T\to\infty} \epsilon_T = \lim_{T\to\infty} \delta_T &= \lim_{T\to\infty} \mu_T = 0,\\
  \ell_T &= \omega(\log T),\\
  \text{ and } \lim_{T\to \infty} \frac{ \ell_T-r_T}{\sqrt{T\mu_T}} &\geq \Theta.\label{eq:def-rate}
\end{align}
The supremum of all achievable throughputs is called the \emph{covert and secret key capacity} and denoted $C_\textnormal{qck}$. Note that the definition of the throughput already captures the scaling of the throughput with the square root of the number of channel uses, $\sqrt{T}$. The scaling is justified a posteriori by our analysis that shows that $C_{\textnormal{qck}}$ is lower bounded by a constant that only depends on the channel parameters. The unit of $C_{\textnormal{qck}}$ is therefore in nats per square root of channel use. Our main results are lower-bounds on the covert capacity obtained by showing the existence of sequences of covert secret key generation protocols using reverse or forward reconciliation. 

To analyze the performance of protocols with forward reconciliation, we build upon existing results for covert communication over cq-channels \cite{Sheikholeslami2016, Wang2016c} with appropriate extensions to guarantee secrecy. The innovative principle of our approach is best highlighted for protocols with reverse reconciliation as follows. In a first phase, Alice transmits a sequence of independent and identically distributed (iid) symbols $\mathbf{X}$ distributed according to a Bernoulli$(\alpha_T)$ distribution over the cq-channel, where
$\alpha_T\in\omega((\frac{\log T}{T})^{\frac{2}{3}})\cap o(\frac{1}{\sqrt{T}})$.
Intuitively, the choice of $\{\alpha_T\}_{T\geq 1}$ must ensure that $\mathbf{X}$ is sparse, so that the warden cannot suspect the existence of information symbols, but not so sparse that Alice and Bob cannot extract a long enough key from their observation. We shall show that our choice of $\{\alpha_T\}_{T\geq1}$  satisfies simultaneously both requirements. In a second phase, Bob measures his received quantum states in some basis and, based on the output of the measurements, generates two messages $W$ and $K$, representing public information reconciliation and secret key, respectively. Bob subsequently sends $W$ through the public channel, and Alice  recovers $K$ using $W$ and $\mathbf{X}$. Although the second phase of the protocol seems deceptively similar to a standard application of information reconciliation and privacy amplification, there exists a technical difficulty because of the specific distributions of Alice's and Bob's  observations, which precludes the use of standard tools. More precisely, in the finite-length analysis resulting from standard information reconciliation and privacy amplification, penalty terms appear that depend on the second order statistics of the conditional information density of Bob and Eve's observations given $\mathbf{X}$, which  scale as $\omega(\sqrt{T})$; However, the scaling of the covert throughput is known to be $o(\sqrt{T})$, which is dominated by those penalties and therefore prohibits key expansion. We instead resort to a technique called likelihood encoder~\cite{Song2016}, in which the encoders used to generate $W$ and $K$ are derived from different principles. In particular, instead of information reconciliation and privacy amplification, we use channel coding and channel resolvability to only analyze quantities depending on mutual information, which has the same scaling as the number of bits generated by a covert protocol.

The analysis of protocols with forward and reverse reconciliation leads to Theorem~\ref{th:main_quantum_passive} below, which proof is given in Appendix~\ref{sec:proof-theorem-passive}. 
\begin{theorem}
\label{th:main_quantum_passive}
Let $\{\ket{y}^B\}$ be any orthonormal basis for  $\calH^B$, and define $\widetilde{\rho}^{BE}_x \eqdef \sum_y \pr{\ket{y}\bra{y}^B\otimes I^E}\rho^{BE}_x\pr{\ket{y}\bra{y}^B\otimes I^E}$. Assume that $\calH^B$ and $\calH^E$ have finite dimension and $0<\chi_2\pr{\widetilde{\rho}^E_1\| \widetilde{\rho}^E_0}<\infty$. We have
\begin{align}
\label{eq:passive-forward-througput}
C_{\textnormal{qck}} \geq  \sqrt{\frac{2}{\chi_2\pr{\widetilde{\rho}^E_1\| \widetilde{\rho}^E_0}}} \pr{\avgD{\widetilde{\rho}_1^B}{\widetilde{\rho}_0^B}-\avgD{\widetilde{\rho}_1^E}{\widetilde{\rho}_0^E}},
\end{align}
and if $\widetilde{\rho}^{BE}_0 = \widetilde{\rho}^B_0 \otimes \widetilde{\rho}^E_0$, then
\begin{multline}
C_{\textnormal{qck}} \geq \sqrt{\frac{2}{\chi_2\pr{\widetilde{\rho}^E_1\| \widetilde{\rho}^E_0}}}\left( \D{\widetilde{\rho}^{BE}_1}{\widetilde{\rho}^{BE}_0} - \D{\widetilde{\rho}^E_1}{\widetilde{\rho}^E_0}\right.\\\left. - \D{\widetilde{\rho}^{BE}_1}{\widetilde{\rho}^B_1 \otimes \widetilde{\rho}^E_1}\right), 
\end{multline}
which simplifies when $\widetilde{\rho}^{BE}_1 = \widetilde{\rho}^B_1 \otimes \widetilde{\rho}^E_1$ as
\begin{align}
\label{eq:passive-reverse-througput-simp}
C_{\textnormal{qck}} \geq  \sqrt{\frac{2}{\chi_2\pr{\widetilde{\rho}^E_1\| \widetilde{\rho}^E_0}}}\avgD{\widetilde{\rho}_1^B}{\widetilde{\rho}_0^B}.
\end{align}
\end{theorem}
While this result certainly does not hold for the most general quantum setting, note that the covert secret key throughputs predicted hold with a precise definition of covertness that explicitly includes the public communication and demonstrate the existence of efficient protocols that allow key expansion. Perhaps  more importantly, as apparent in the proof of the result, such protocols do \emph{not} rely on a secret key to determine the instances in which Alice transmits non-zero states; in contrast, our proof shows the existence of reconciliation and key-extraction algorithms capable of \emph{extracting} the diffuse secret correlations created by Alice's sparse transmission of non-innocent states.

\begin{figure}
  \centering
   \includegraphics[width=\linewidth]{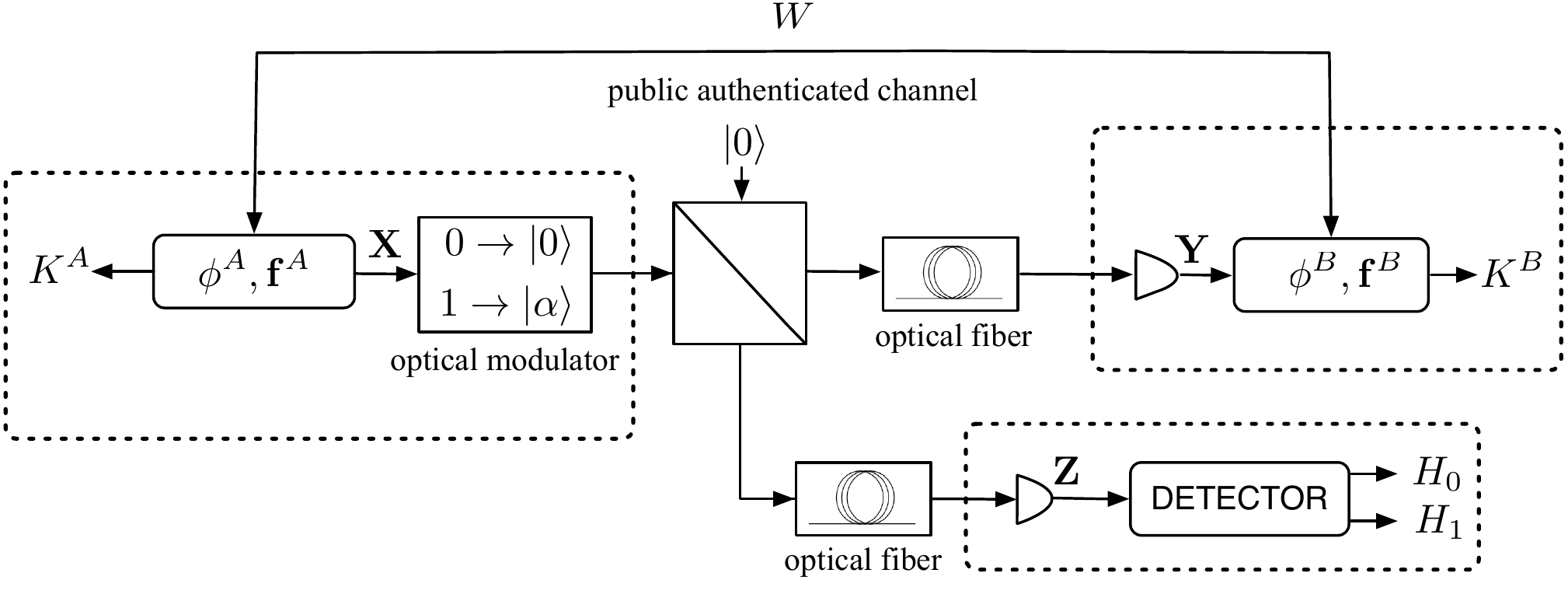}
  \caption{Simplified model of a lossy bosonic channel.}
  \label{fig:beam-splitter}
\end{figure}

\begin{figure}
  \centering
   \includegraphics[width=\linewidth]{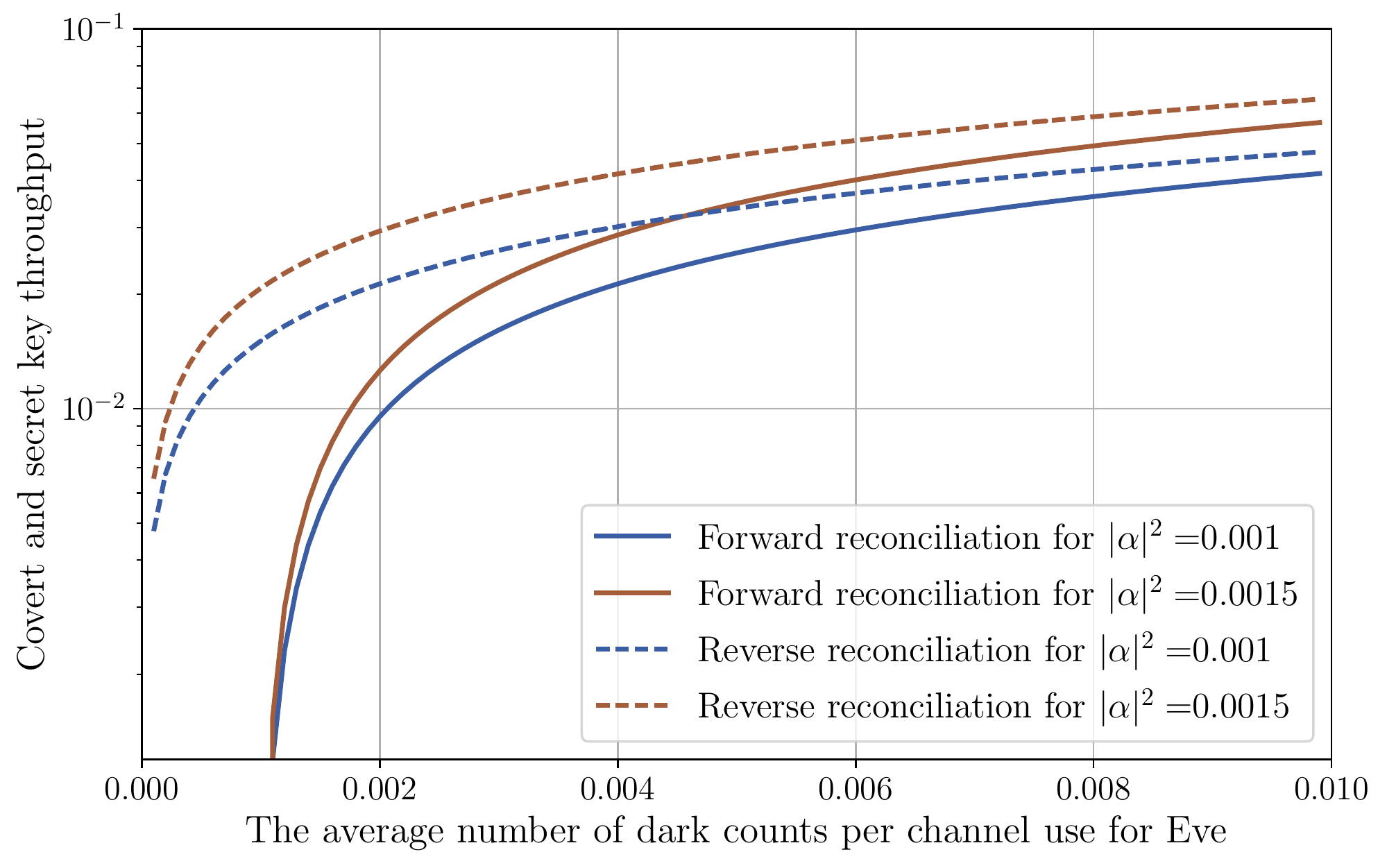}
  \caption{Covert and secret key generation throughput as a function of Eve's dark count rate.}
  \label{fig:dark-count}
\end{figure}

\begin{figure}
  \centering
  \includegraphics[width=\linewidth]{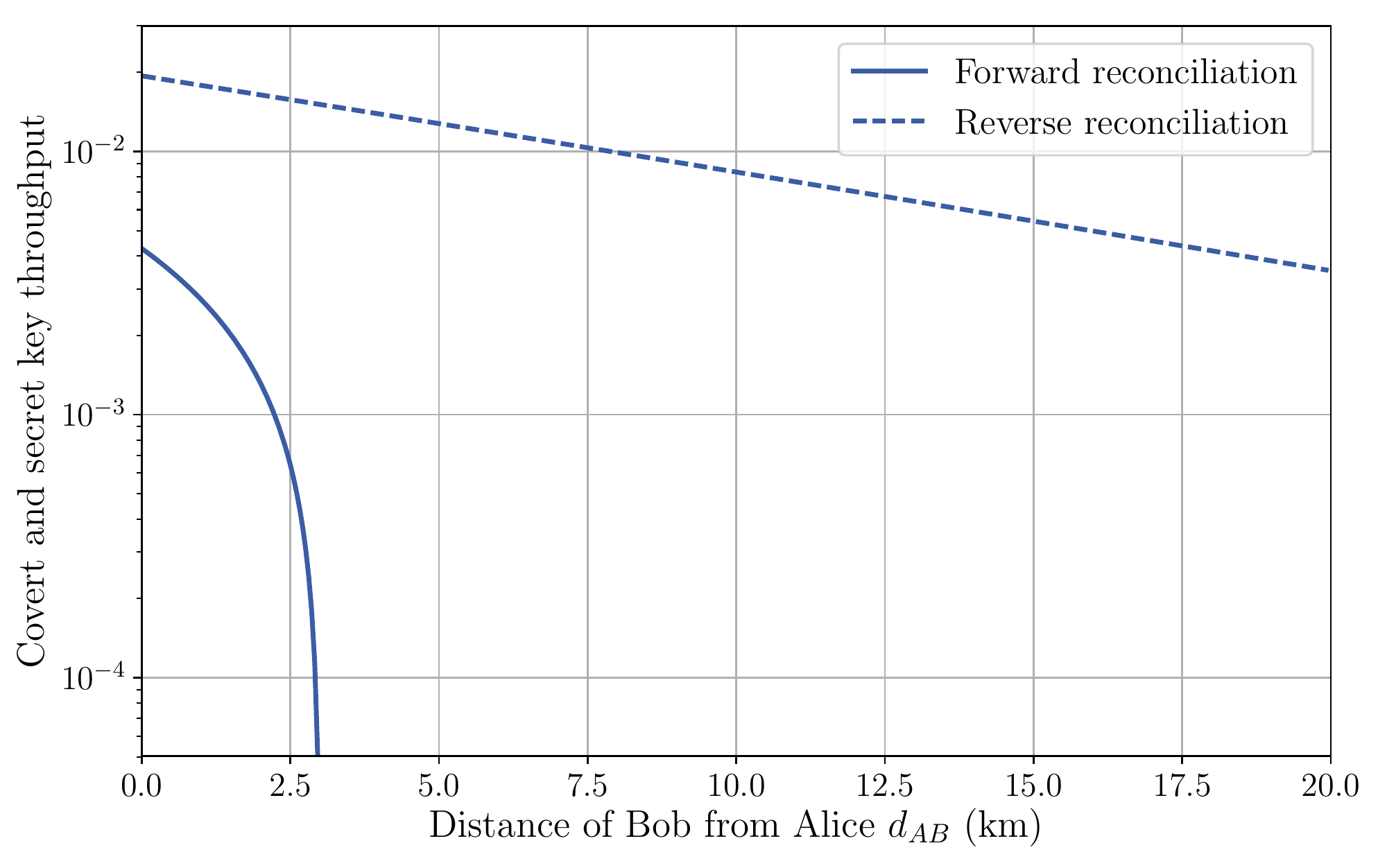}
  \caption{Covert and secret key generation throughput for a lossy bosonic channel.}
  \label{fig:dark-count-distance}
\end{figure}
As an illustration, we consider the situation depicted in Fig.~\ref{fig:beam-splitter} in which the input port of a balanced beam-splitter is in control of Alice while Bob and Eve are each connected to one of the output ports through optical fibers of length $d_{AB}$ and $d_{AE}$, respectively, and loss $\gamma$~dB/km. We further assume that the second input port is in the vacuum state, and that Alice uses the vacuum state $\ket{0}$ and a coherent state $\ket{\alpha}$ as the innocent and the information symbol, respectively. Bob and Eve measure their output ports with photodetectors to count the number of photons at each channel use. The photodetectors suffer from dark count 
that is beneficial for covert communication since detection of photons at Eve does not necessarily imply the existence communication. Let $\eta_B$ and $\eta_E$ be Bob and Eve's photodetector efficiency, respectively, and $\lambda_B$ and $\lambda_E$ be Bob and Eve's photodetector dark count rate, respectively. The achievable covert and secret key throughputs can be obtained by substituting the quantities
\begin{align}
\widetilde{\eta}_{B} &\eqdef \eta_B 10^{-\frac{d_{AB}  \gamma}{10} }\\
\widetilde{\eta}_{E} &\eqdef \eta_E10^{-\frac{d_{AE} \gamma}{10} }\\
\chi_2\pr{\widetilde{\rho}^E_1\| \widetilde{\rho}^E_0} &= e^{\frac{(\lambda_E + |\alpha|^2 \widetilde{\eta}_E)^2}{\lambda_E} -\lambda_E +2|\alpha|^2 \widetilde{\eta}_E}\\
\avgD{\widetilde{\rho}_1^B}{\widetilde{\rho}_0^B} &= (\lambda_B + |\alpha|^2 \widetilde{\eta}_B)\log\pr{\lambda_B + |\alpha|^2 \widetilde{\eta}_B} - |\alpha|^2 \widetilde{\eta}_B\\
\avgD{\widetilde{\rho}_1^E}{\widetilde{\rho}_0^E} &= (\lambda_E + |\alpha|^2 \widetilde{\eta}_E)\log\pr{\lambda_E + |\alpha|^2 \widetilde{\eta}_E} - |\alpha|^2 \widetilde{\eta}_E
\end{align}
in \eqref{eq:passive-forward-througput} and \eqref{eq:passive-reverse-througput-simp} for forward and reverse reconciliation, respectively. Note that the output states of this channel belong to infinite-dimensional spaces and, strictly speaking, one cannot directly apply Theorem~\ref{th:main_quantum_passive}. Nevertheless, since for the number states $\{\ket{n}\}_{n\geq 0}$, $\bra{n}\rho\ket{n}$ decays exponentially for all output states $\rho$,  one can construct a sequence of channels with finite-dimensional output states for which the quantities used in \eqref{eq:passive-forward-througput} and \eqref{eq:passive-reverse-througput-simp}, as well as the performance of any covert and secret key generation protocol, tend to those of the original channel.

We illustrate in Fig.~\ref{fig:dark-count} the achievable covert and secret key throughput as a function of Eve's photodetector dark count rate $\lambda_E$ for $\gamma = 0.2 $~dB/km, $\eta_B = \eta_E = 0.97$, $\lambda_B = 0.001$, and $d_{AB} = d_{AE} = 3$~km. In Fig.~\ref{fig:dark-count-distance}, we also illustrate the achievable covert and secret key throughput as a function of the distance of Bob to Alice $d_{AB}$ for  $|\alpha|^2 = 0.001$, $\gamma = 0.2$~dB/km, $\eta_B = \eta_E = 0.97$, $\lambda_B = \lambda_E = 0.001$, and $d_{AE} = 3$~km. As expected, the secret and covert key throughputs are orders of magnitude lower than their counterparts without covertness constraint. This is an unfortunate but unavoidable byproduct of the covertness constraint, which severely limits how many useful bits can be embedded in transmitted signals.

\section{Covert and secret key generation over unknown cq-channel}
\label{sec:unkown-channels}
We now relax the assumption that Alice and Bob have full-knowledge of the communication channel. To this end, we assume that to transmit ``0'' and ``1'', Alice prepares two states $\rho_0^{\widetilde{A}}$ and $\rho_1^{\widetilde{A}}$, respectively, and sends these states to Bob through an unknown but fixed quantum channel $\calE_{{\widetilde{A}}\to B}$. We know that Eve's information about the communication is limited to the output of the complementary channel of $\calE_{{\widetilde{A}}\to B}$ denoted by $\calE_{{\widetilde{A}}\to E} \eqdef \calE_{{\widetilde{A}}\to B}^ \dagger$ \cite{wilde2013quantum}. In this setting, we define two cq-channels from Alice to Bob and Eve as $x\mapsto\rho_x^B \eqdef \calE_{{\widetilde{A}}\to B}(\rho_x^{\widetilde{A}})$ and $x\mapsto\rho_x^E\eqdef \calE_{{\widetilde{A}}\to E}(\rho_x^{\widetilde{A}})$. To communicate covertly in this scenario, we propose a protocol consisting of two phases: in the first phase, Alice and Bob  covertly perform quantum tomography to derive bounds on the required parameters of the channels, and in the second phase, after the class of possible channels is sufficiently narrowed down, Alice and Bob run a \emph{universal} covert code,  the use of which is critical since there would always be an  error in the estimation phase. More precisely, by the central limit theorem, with high probability, the estimation error would be $\Omega(\frac{1}{\sqrt{T}})$ when the estimation is taking place over $O(T)$ channel uses. Thus, if a protocol is guaranteed to sustain a certain performance for only the \emph{estimated channel}, the estimation error could potentially lead to a  significant deviation in the predicted performance when the protocol is executed over $T$ uses of the \emph{true channel}. 

Before stating the main result of this section, we recall the definition of the \emph{$\chi$ representation} of a quantum channel from \cite{nielsen2002quantum}. If $\ket{1}, \cdots, \ket{d}$ is an orthonormal basis for $\calH^B$, we let $\widetilde{E}_{d(n-1) + m} \eqdef \ket{n}\bra{m}$ so that $\widetilde{E}_1, \cdots, \widetilde{E}_{d^2}$ forms an orthonormal basis for $\calL(\calH)$. By \cite{nielsen2002quantum}, there exists coefficients $\chi_{j, k}$ such that
\begin{align}
\calE(\rho)= \sum_{j, k}\widetilde{E}_j \rho\widetilde{E}_k^\dagger \chi_{j, k}.
\end{align}
The matrix $\chi$ is defined as the matrix with entries $\chi_{j,k}$. The following then holds.
\begin{theorem}
\label{th:main-collective}
Let $\widetilde{\lambda}^{\chi}$, $\widetilde{\lambda}^{B}$ and $\widetilde{\lambda}^E$ be fixed in $]0, 1]$. Let $\zeta > 0$ and let $\{\alpha_T\}_{T\geq1}$ be such that
\begin{align}
  \alpha_T \in\omega\pr{\pr{\frac{\log T}{T}}^{\frac{2}{3}}}\cap o\pr{\frac{1}{\sqrt{T}}}.\label{eq:condition_alpha}
\end{align}
There exists a vanishing sequence $\{\epsilon_T\}_{T\geq 1}$ and a sequence of $(\epsilon_T,\epsilon_T,\mu_T)$ covert and secret key generation protocols such that for all  quantum channels $\calE_{{\widetilde{A}}\to B}$ with
\begin{align}
\label{eq:collective-lambdae}
\lambda_{\min}(\calE_{{\widetilde{A}}\to E}(\rho_0^{\widetilde{A}})) \geq \widetilde{\lambda}^E,
\end{align}
we have
\begin{align}
\mu_T &\leq (1+\epsilon_T) \frac{\alpha_T^2 \chi_2(\calE_{{\widetilde{A}}\to E}(\rho_1^{\widetilde{A}})\|\calE_{{\widetilde{A}}\to E}(\rho_0^{\widetilde{A}}))T}{2}.
\end{align}
Additionally, if Eq. \eqref{eq:collective-lambdae} holds as well as
\begin{align}
\label{eq:collective-chi-min}
\lambda_{\min}(\chi) \geq \widetilde{\lambda}^{\chi},\\
\label{eq:collective-lambdab}
\lambda_{\min}(\calE_{{\widetilde{A}}\to B}(\rho_0^{\widetilde{A}})) \geq \widetilde{\lambda}^B,
\end{align}
then with probability at least $1-\epsilon_T$, the length of the generated key is at least
\begin{multline}
 (1-\zeta) (\D{\calE_{{\widetilde{A}}\to B}(\rho_1^{\widetilde{A}})}{\calE_{{\widetilde{A}}\to B}(\rho_0^{\widetilde{A}})} \\-\D{\calE_{{\widetilde{A}}\to E}(\rho_1^{\widetilde{A}})}{\calE_{{\widetilde{A}}\to E}(\rho_0^{\widetilde{A}})} )\alpha_T T.\label{eq:key_amount}
\end{multline}
\end{theorem}
Theorem~\ref{th:main-collective} has a slightly more complicated form than Theorem~\ref{th:main_quantum_passive} because of the statistical uncertainty associated to the estimation phase. Nevertheless, the interpretation remains intuitive and as follows. Firstly, the parameter $\alpha_T$ defined in~\eqref{eq:condition_alpha} controls the fraction of symbols ``1'' transmitted over $T$ channel uses, and should be understood as close to but slightly less than $1/\sqrt{T}$. For this choice, the covertness parameter $\mu_T$ vanishes with $T$ while the number of secret key bits covertly generated scales almost as $\Omega(\sqrt{T})$. Secondly, the condition \eqref{eq:collective-lambdae} states that the channel to the adversary should be noisy enough to allow covert operation a priori. This condition should not be surprising, as we know that covert communication is impossible in a general setting \cite{Bash2015a}. The conditions~\eqref{eq:collective-chi-min} and~\eqref{eq:collective-lambdab} are technical conditions  to avoid extreme cases, which are necessary to ascertain the reliability of our estimation protocol and can be set to the technological limits of Eve's apparatus. As apparent in the proof, these three conditions only affect the sequence $\{\epsilon_T\}_{T\geq 1}$ but not  the asymptotic covert and secret key rate guaranteed in~\eqref{eq:key_amount}. Finally, Alice and Bob can test whether  $\lambda_{\min}(\calE_{{\widetilde{A}}\to E}(\rho_0^{\widetilde{A}})) \geq \widetilde{\lambda}^E$ with high probability at the end the phase and abort the protocol if the condition is violated. In that case, we cannot guarantee the covertness constraint as we have sent too many non-innocent symbols in the estimation phase but we may still avoid detection by aborting the key-generation phase of the protocol.

\begin{figure}
\centering
\includegraphics[width=1\linewidth]{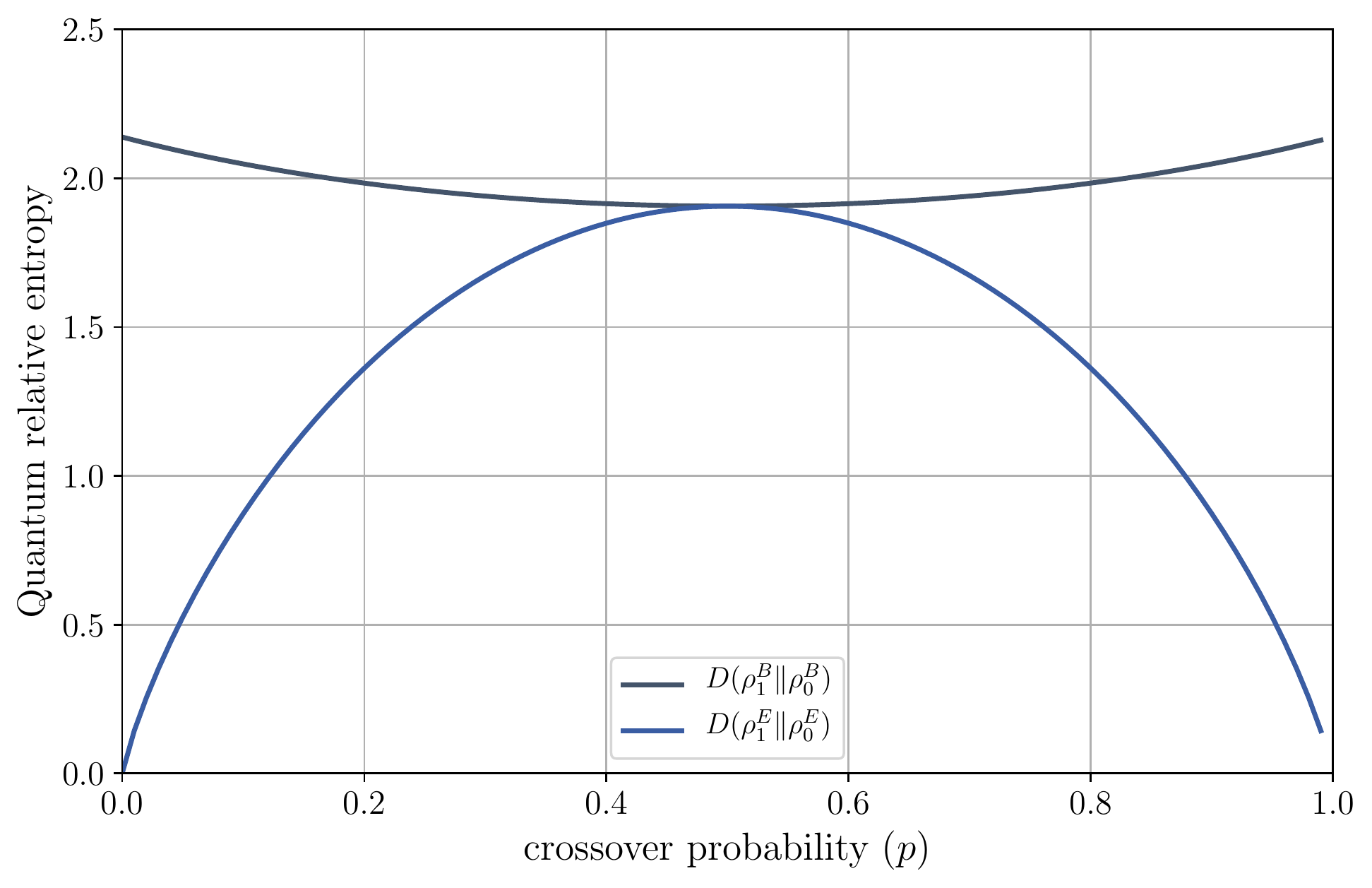}
\centering
\caption{Quantum relative entropy of main and warden channels}
\label{fig:example}
\end{figure}
We illustrate the result of Theorem~\ref{th:main-collective} for a specific quantum channel, $\calE_{{\widetilde{A}}\to B}$, from Alice to Bob. We assume that $\calE_{{\widetilde{A}}\to B}$ is a phase flip channel with flipping probability $p$, namely, $\calE_{{\widetilde{A}}\to B}(\rho^{\widetilde{A}}) = (1-p) \rho +p\sigma_z\rho \sigma_z$, and the matrix representation of Alice's transmitted states in the computational basis is
\begin{align}
\rho^{\widetilde{A}}_0 &= \begin{bmatrix}0.95&0\\0&0.05\end{bmatrix}\\
\rho^{\widetilde{A}}_1 &= \begin{bmatrix}0.2&0.3\\0.3&0.8\end{bmatrix}.
\end{align}
For $x\in\{0, 1\}$, we define $\rho^B_x \eqdef \calE_{{\widetilde{A}}\to B}(\rho^{\widetilde{A}}_x)$ and $\rho^E_x \eqdef \calE_{{\widetilde{A}}\to B}^\dagger(\rho^{\widetilde{A}}_x)$, and in Fig.~\ref{fig:example}, we show $\D{\rho^B_1}{\rho^B_0}$ and $\D{\rho^E_1}{\rho^E_0}$ for different values of $p$. By Theorem~\ref{th:main-collective}, the number of generated covert and secret key bits is on the order of $(\D{\rho^B_1}{\rho^B_0} - \D{\rho^E_1}{\rho^E_0}) \alpha_T T$, which scales as $O(\alpha_T T)$ except for $p=0.5$.
\section*{Acknowledgement}
This work was supported in part by Nation Science Foundation under the award 1527074.

\appendix

\section{Proof of Theorem~\ref{th:main_quantum_passive}}
\label{sec:proof-theorem-passive}
We prove Theorem~\ref{th:main_quantum_passive} by  generalizing  the proof of \cite[Theorem 1]{Bloch2017} to the quantum setting. The most challenging part of this generalization is to establish a channel resolvability result for cq-channels for distributions suitable for covert communications. We first introduce some preliminary concepts regarding covert communications mostly borrowed from \cite{Bloch2016a}. We also note that the use of standard proof techniques for secret key generation such as source coding with side information and privacy amplification is challenging for covert communication as discussed in Section~\ref{sec:quantum_passive_results}. We therefore resort to the likelihood encoder technique \cite{Song2016} in which we first define an auxiliary problem that can be analyzed using channel coding approaches, for which designing a code for the main problem is reduced to the design of code for the auxiliary problem.
\subsection{Preliminaries}
We define here required quantities used for our achievability proof. Suppose Alice sends iid symbols through her cq-channel $x\mapsto \rho_x^{BE}$ with each symbol  distributed according to  $Q_X\sim$ Bernoulli$(\alpha_T)$ for $\alpha_T\in(0, 1)$. Upon receiving each state, Bob makes a measurement in a fixed orthonormal basis $\{\ket{y}^B\}$ for $\calH^B$ to obtain a classical symbol $y$. In the following, we define equivalent cq-channels from Bob to Alice and Eve that results in to the same joint state for the three parties.

\begin{definition}
\label{def:q-covert-process}
Let $\alpha_T \in\omega\pr{\pr{\frac{\log T}{T}}^{\frac{2}{3}}}\cap o\pr{\frac{1}{\sqrt{T}}}$. We define 
\begin{align}
Q_{Y|X}(y|x) &\eqdef \bra{y}^B \rho_x^B \ket{y}^B,\displaybreak[0]\\
\widetilde{\rho}^{BE}_x &\eqdef \sum_y\pr{\ket{y}\bra{y}^B\otimes I^E}\rho^{BE}_x\pr{\ket{y}\bra{y}^B\otimes I^E},\displaybreak[0]\\
\widetilde{\rho}^{ABE} &\eqdef \sum_{x} Q_X(x) \ket{x}\bra{x}^A \otimes \widetilde{\rho}^{BE}_x,\displaybreak[0]\\
\widetilde{\rho}_{x,y}^E &\eqdef \frac{\textnormal{tr}_B\pr{\pr{\ket{y}\bra{y}^B\otimes I^E} \rho_x^{BE}\pr{\ket{y}\bra{y}^B\otimes I^E}}}{Q_{Y|X}(y|x)}\\
\widetilde{\rho}_y^{AE} &\eqdef \sum_{x} Q_{X|Y}(x|y) \ket{x}\bra{x}^A \otimes \widetilde{\rho}_{x, y}^E.
\end{align}
Note that the state $\widetilde{\rho}^{ABE}$ is the joint state of all parties after Bob's measurement, which is classical for both Alice and Bob, and $\widetilde{\rho}^{BE}_x$, $\widetilde{\rho}^{AE}_y$, and $\widetilde{\rho}^{E}_{x, y}$ are the corresponding conditional quantum states.
\end{definition}
The following lemma establishes useful properties of $\rho_0^{BE}$ under the assumption $\widetilde{\rho}_{0}^{BE} =  \widetilde{\rho}^B_0 \otimes \widetilde{\rho}^E_0$.
\begin{lemma}
\label{lm:rho_0}
If $\widetilde{\rho}_{0}^{BE} =  \widetilde{\rho}^B_0 \otimes \widetilde{\rho}^E_0$ then, for all $y$, it holds that $\widetilde{\rho}_{0,y}^E =  \widetilde{\rho}^E_0$. Furthermore, we have
\begin{multline}
I(Q_Y, \widetilde{\rho}_y^E) = \alpha_T\left(\D{\widetilde{\rho}_1^B}{\widetilde{\rho}_0^B} + \D{\widetilde{\rho}_1^E}{\widetilde{\rho}_0^E} - \right.\\\left.\D{\widetilde{\rho}_1^{BE}}{\widetilde{\rho}_0^{BE}}+ \D{\widetilde{\rho}_1^{BE}}{\widetilde{\rho}_1^{B} \otimes \widetilde{\rho}_1^{E}}\right) + O(\alpha_T^2).\label{eq:quan_i_yz}
\end{multline}
\end{lemma}
\begin{proof}
  By the spectral decomposition theorem, there exist orthonormal bases $\{\ket{y}^B\}$ and $\{\ket{z}^E\}$ for $\calH^B$ and $\calH^E$, respectively, such that
\begin{align}
\widetilde{\rho}_{0}^{B} &= \sum_{y} \lambda_y \ket{y}\bra{y}^B\\
\widetilde{\rho}_{0}^{E} &= \sum_{z} \lambda_z \ket{z}\bra{z}^E\\
\widetilde{\rho}_{0}^{BE}  &= \sum_{y, y', z, z'} \lambda_{yy'zz'} \ket{y}\bra{y'}^B\otimes \ket{z}\bra{z'}^E.
\end{align}
Our assumption that $\widetilde{\rho}_{0}^{BE} =  \widetilde{\rho}^B_0 \otimes \widetilde{\rho}^E_0$ implies that $ \lambda_{yy'zz'} = \lambda_y\lambda_z \indic{y=y', z= z'}$. Furthermore, for any $y$, we have by definition 
\begin{align}
\widetilde{\rho}_{0,y}^E &\eqdef \frac{\text{tr}_B\pr{\pr{\ket{y}\bra{y}^B\otimes I^E} \rho_0^{BE}\pr{\ket{y}\bra{y}^B\otimes I^E}}}{Q_{Y|X}(y|0)}\displaybreak[0]\\
&= \frac{1}{Q_{Y|X}(y|0)}\text{tr}_B\left(\pr{\ket{y}\bra{y}^B\otimes I^E}\nonumber\right.\displaybreak[0]\\
&\phantom{=}\times\left. \pr{ \sum_{y',z'} \lambda_{y'}\lambda_{z'} \ket{y'}\bra{y'}^B\otimes \ket{z'}\bra{z'}^E}\nonumber\right.\displaybreak[0]\\
&\phantom{=}\left.\times \pr{\ket{y}\bra{y}^B\otimes I^E}\right)\displaybreak[0]\\
&= \frac{\text{tr}_B\pr{\sum_{z'} \lambda_{y}\lambda_{z'} \ket{y}\bra{y}^B\otimes \ket{z'}\bra{z'}^E}}{Q_{Y|X}(y|0)}\displaybreak[0]\\
&= \frac{\lambda_{y}}{Q_{Y|X}(y|0)}\sum_{z'}\lambda_{z'}  \ket{z'}\bra{z'}^E\displaybreak[0]\\
&=  \frac{\lambda_{y}}{Q_{Y|X}(y|0)} \widetilde{\rho}_{0}^{E} \label{eq:rho_z_0}.
\end{align}
We also know that $\tr{\widetilde{\rho}_{0}^{E}} = \tr{\widetilde{\rho}_{0,y}^E } = 1$, which together with \eqref{eq:rho_z_0} yields $\widetilde{\rho}_{0}^{E} = \widetilde{\rho}_{0,y}^E $.

To prove \eqref{eq:quan_i_yz}, notice that
\begin{multline}
I(Q_Y, \widetilde{\rho}_y^{E}) \\
\begin{split}
&= \avgI{B;E}_{\widetilde{\rho}}\\
&= \avgI{A;B}_{\widetilde{\rho}} +  \avgI{A;E}_{\widetilde{\rho}} -  \avgI{A;BE}_{\widetilde{\rho}} +  \avgI{B;E|A}_{\widetilde{\rho}}.
\end{split}
\end{multline}
Moreover, for $\widetilde{\rho}_{\alpha_T}^B\eqdef (1-\alpha_T) \widetilde{\rho}_0^B +\alpha_T \widetilde{\rho}_1^B $, we can write
\begin{align}
&\avgI{A;B}_{\widetilde{\rho}} \\
&\phantom{==}= H(\widetilde{\rho}_{\alpha_T}^B ) - (1-\alpha_T)  H(\widetilde{\rho}_0^B) - \alpha_TH(  \widetilde{\rho}_1^B)\\
&\phantom{==}= -\textnormal{tr}(\widetilde{\rho}_{\alpha_T}^B\log\pr{\widetilde{\rho}_{\alpha_T}^B} - (1-\alpha_T)\widetilde{\rho}_0^B  \log(\widetilde{\rho}_0^B)\nonumber \\
&\phantom{==}\phantom{= ==============}- \alpha_T \widetilde{\rho}_1^B\log(\widetilde{\rho}_1^B))\\
&\phantom{==}= -\textnormal{tr}(\widetilde{\rho}_{\alpha_T}^B\pr{\log\pr{\widetilde{\rho}_{\alpha_T}^B} - \log\pr{\widetilde{\rho}_0^B} + \log\pr{\widetilde{\rho}_0^B}}\nonumber\\
&\phantom{==}\phantom{=} - (1-\alpha_T)\widetilde{\rho}_0^B  \log(\widetilde{\rho}_0^B) - \alpha_T \widetilde{\rho}_1^B\log(\widetilde{\rho}_1^B))\\
&\phantom{==}= -\textnormal{tr}(\widetilde{\rho}_{\alpha_T}^B \pr{\log \widetilde{\rho}_{\alpha_T}^B - \log \widetilde{\rho}_{0}^B}\nonumber\\
&\phantom{==}\phantom{==========}  - \alpha_T \widetilde{\rho}_{1}^B \pr{\log \widetilde{\rho}_{1}^B- \log \widetilde{\rho}_{0}^B}  )\\
&\phantom{==}= \alpha_T \D{\widetilde{\rho}_{1}^B}{\widetilde{\rho}_{0}^B} - \D{\widetilde{\rho}_{\alpha_T}^B}{\widetilde{\rho}_{0}^B}\\
&\phantom{==}\stackrel{(a)}{=} \alpha_T \D{\widetilde{\rho}_{1}^B}{\widetilde{\rho}_{0}^B}  + O(\alpha_T^2),
\end{align}
where $(a)$ follows from \cite[Equation (19)]{Sheikholeslami2016}. Similarly, we obtain
\begin{align}
\avgI{A;E}_{\widetilde{\rho}} &= \alpha_T \D{\widetilde{\rho}_{1}^E}{\widetilde{\rho}_{0}^Z}  + O(\alpha_T^2),\\
\avgI{A;BE}_{\widetilde{\rho}} &= \alpha_T \D{\widetilde{\rho}_{1}^{BE}}{\widetilde{\rho}_{0}^{BE}}  + O(\alpha_T^2).
\end{align}
Since $X$ is classical, \cite[Equation (11.92)]{wilde2013quantum} yields that  
\begin{align}
  \avgI{B;E|A}_{\widetilde{\rho}} 
  &= (1-\alpha_T)  \avgI{B;E}_{\widetilde{\rho}_0} + \alpha_T   \avgI{B;E}_{\widetilde{\rho}_1}\\
  &\stackrel{(a)}{=} \alpha_T   \avgI{B;E}_{\widetilde{\rho}_1}\\
  &= \alpha_T \D{\widetilde{\rho}_1^{BE}}{\widetilde{\rho}_1^{B} \otimes \widetilde{\rho}_1^{E}},
\end{align}
where $(a)$ follows from our assumption that $\widetilde{\rho}_{0}^{BE} =  \widetilde{\rho}^B_0 \otimes \widetilde{\rho}^E_0$. This completes the proof of \eqref{eq:quan_i_yz}.
\end{proof}

\subsection{One-shot results}
We recall here one-shot results for classical channel coding and classical channel resolvability (Lemma~\ref{lm:oneshot_channel}) and quantum channel resolvability (Lemma~\ref{lm:quantum-resolve}) that play a central role on our analysis. Given a classical channel $(\calX, W_{Y|X}, \calY)$, a message $W$  uniformly distributed over $\intseq{1}{M}$, and an encoder $f:\intseq{1}{M}\to \calX$, let $\widehat{P}_{WXY}(w, x, y) \eqdef \frac{1}{M}\indic{f(w) = x}W_{Y|X}(y|x)$ be the induced  Probability Mass Function (PMF)  of $W$, $X$, and $Y$, and $\widehat{W} \eqdef \arg \max_{w\in\intseq{1}{M}} W_{Y|X}(Y|f(w))$ be the maximum likelihood decoder at the output.
\begin{lemma}[One-shot Bounds]
\label{lm:oneshot_channel}
 If $F$ is a random encoder such that $\{F(w)\}_{w\in\intseq{1}{M}}$ are iid according to a distribution $P_X$ over $\calX$, then for any $\gamma \in \mathbb{R}$, we have
\begin{multline}
\label{eq:rel-oneshot}
\E[F]{\P{\widehat{W} \neq W}}\\ \leq \P[P_X \times W_{Y|X}]{\log \frac{{W}_{Y|X}(Y|X)}{({W}_{Y|X}\circ P_X)(Y)}\leq \gamma} + \frac{M}{2^\gamma},
\end{multline}
and
\begin{multline}
\label{eq:res-bound22}
\E[F]{\V{\widehat{P}_Y;W_{Y|X}\circ P_X}} \\\leq \P[P_X \times W_{Y|X}]{\log \frac{W_{Y|X}(Y|X)}{(W_{Y|X}\circ P_X)(Y)}\geq \gamma} + \sqrt{\frac{2^\gamma}{M}},
\end{multline}
where $(W_{Y|X}\circ P_X)(y) \eqdef \sum_x P_X(x)W_{Y|X}(y|x)$.
\end{lemma}
\begin{proof}
See \cite{Polyanskiy2010} for \eqref{eq:rel-oneshot} and \cite{hayashi2006general} for \eqref{eq:res-bound22}.
\end{proof}
Let  $y\mapsto \rho_y$ denote a cq-channel and $P_Y$ be a {PMF} over $\calY$. If $\overline{\rho} \eqdef \sum_y P_Y(y) \rho_y$, our objective is to find an encoder $f:\intseq{1}{M}\to \calY$ such that $\|\overline{\rho} - \widehat{\rho}\|_1$ be small, where $\widehat{\rho} \eqdef \frac{1}{M}\sum_{i=1}^M \rho_{f(i)}$.
\begin{lemma}
\label{lm:quantum-resolve}
If $F:\intseq{1}{M} \to \calY$ is a random encoder whose codewords are {iid} according to $P_Y$, then for all $s\leq0$ and $\gamma $, we have
\begin{align}
\label{eq:q-res-bound}
\E[F]{\|\overline{\rho} - \widehat{\rho}\|_1 } \leq 2 \sqrt{2^{\gamma s + \phi(s)}} + \sqrt{\frac{2^{\gamma} \nu}{M}},
\end{align}
where $\phi(s) \eqdef \log\pr{\sum_y P_Y(y) \tr{\rho_y^{1-s} \overline{\rho}^s}}$ and $\nu$ is  the number of distinct eigenvalues of  $\overline{\rho}$.
\end{lemma}
\begin{proof}
See \cite[Lemma 9.2]{hayashi2006quantum}.
\end{proof}

\subsection{An auxiliary problem}
To show the existence of good codes for our main problem, we use the likelihood encoder technique \cite{Song2016}, and in particular, define an auxiliary problem for which we can exploit channel coding instead of source coding. We then show how these two problems are related in Section~\ref{sec:thm-passive-proof}. Consider a cq-channel $y\mapsto \widetilde{\rho}^{AE}_y$ from Bob to Alice and Eve as in Definition~\ref{def:q-covert-process}. Bob encodes three uniformly distributed messages $W_1\in\intseq{1}{M_1}$, $W_2\in\intseq{1}{M_2}$, and $W_3\in\intseq{1}{M_3}$ into a codeword $\mathbf{Y}$ using an encoder $f:\intseq{1}{M_1}\times \intseq{1}{M_2} \times\intseq{1}{M_3} \to \calY^T$, transmits the codeword $\mathbf{Y}$ over the cq-channel, and sends $W_2$ publicly. Alice subsequently performs a measurement on her received state $\rho_{\mathbf{Y}}^A$ in a fixed basis $\{\ket{x}\}$ to obtain $\mathbf{X}$, and uses $\mathbf{X}$ and $W_2$ to decode $W_1$ as $\widehat{W}_1$. If $P^a_{\mathbf{Y}}$ denotes the induced {PMF} of $\mathbf{Y}$, and $\rho_a^{\mathbf{A}\mathbf{B}\mathbf{E}W_1W_2W_3\widehat{W}_1}$ is the joint state in the auxiliary problem, our objective is to ensure that $\P{\widehat{W}_1 \neq W_1}$, 
$\V{P^a_{\mathbf{Y}}; Q_Y^{\proddist T}}$, and $\|\rho^{\mathbf{E}W_1W_2} -\rho^{\mathbf{E}} \otimes \rho^{W_1W_2} \|_1$ are small.
\begin{lemma}
\label{lm:auxiliary_prob}
If for some $\zeta>0$
\begin{align}
\log M_3 &= \lfloor  (1+\zeta)I(Q_Y, \widetilde{\rho}_y^E) T \rfloor,\\
\log M_1 + \log M_2 + \log M_3 &= \lceil (1+\zeta) H(Q_Y) T \rceil,\\
\log M_1 +\log M_3&= \lfloor (1-\zeta)I(Q_Y, Q_{X|Y})T\rfloor,\\
\end{align}
then there exists a sequence of codes and a positive constant $\xi$ such that
\begin{align}
\P{\widehat{W}_1 \neq W_1} &\leq 2^{-\xi \alpha_T T},\\
\label{eq:py_estimation}
\V{P^a_{\mathbf{Y}}; Q_Y^{\proddist T}} &\leq 2^{-\xi  T},\\
 \|\rho^{\mathbf{E}W_1W_2} - \rho^{\mathbf{E}}\otimes\rho^{W_1W_2}  \|_1&\leq 2^{-\omega(\log T)}.\label{eq:q-aux-sec}
\end{align}
\end{lemma}
\begin{proof}
Let $F:\intseq{1}{M_1}\times \intseq{1}{M_2}\times \intseq{1}{M_3}$ be a random encoder whose codewords are drawn independently according to $Q_Y^{\proddist T}$. By construction, Alice can assume that each symbol $X_i$ is received as the output of a {DMC} $(\calY, Q_{X|Y}, \calX)$ with input $Y_i$, and, therefore, Lemma~\ref{lm:oneshot_channel} implies that
 \begin{multline}
 \E[F]{\P{\widehat{W}_1 \neq W_1}} \\
 \begin{split}
 &=  \frac{1}{M_2} \sum_{w_2} \E[F]{\P{\widehat{W}_1 \neq W_1 |W_2 = w_2}} \\
 &\stackrel{(a)}{\leq} \P[Q_{X|Y}^{\proddist T}\times Q_Y^{\proddist T}]{\sum_{t=1}^T\log \frac{Q_{X|Y}(X_t|Y_t)}{Q_X(X_t)}\leq \gamma} + \frac{M_1 M_3}{2^\gamma}\\
 &= \P[Q_{XY}^{\proddist T}]{\sum_{t=1}^T\log \frac{Q_{Y|X}(Y_t|X_t)}{Q_Y(Y_t)}\leq \gamma} + \frac{M_1 M_3}{2^\gamma},
 \end{split}
 \end{multline}
where $(a)$ follows  from applying Lemma~\ref{lm:oneshot_channel} to the subcodebook $\{F(w_1, w_2, w_3):w_1\in\intseq{1}{M_1}, w_3 \in \intseq{1}{M_3}\}$ for a particular $w_2$. By choosing
\begin{align}
\log M_1 + \log M_3&= \lfloor (1-\zeta)I(Q_X, Q_{Y|X})T\rfloor\\
\gamma &= \left(1-\frac{\zeta}{2}\right)I(Q_X, Q_{Y|X})T,
\end{align}
and using Bernstein's inequality \cite{bernstein1924modification}, we obtain
\begin{multline}
\P[Q_{XY}^{\proddist T}]{\sum_{t=1}^T\log \frac{Q_{Y|X}(Y_t|X_t)}{Q_Y(Y_t)}\leq \gamma} + \frac{M_1 M_3}{2^\gamma}\\
\begin{split}
 &\leq \exp\left(-\frac{-\frac{1}{8}\zeta^2I(Q_Y, Q_{X|Y})^2T}{\Var{\log \frac{Q_{Y|X}(Y|X)}{Q_Y(Y)} } + \frac{1}{3}C_3 \zeta \avgI{X;Y}}\right) \\
 &\phantom{====================}+ 2 ^{-\frac{\zeta}{2}\avgI{X;Y}T}\\
&\leq 2^{-\xi \alpha_T T},
\end{split}
\end{multline}
for some $\xi>0$. Next, by using Lemma~\ref{lm:oneshot_channel} for the  channel $(\calY, Q_{Y'|Y}, \calY)$ with $Q_{Y'|Y}(y'|y) \eqdef \indic{y'=y}$ and the distribution $Q_Y$, we obtain
\begin{multline}
\E[F]{ \V{P_{\mathbf{Y}}^a; Q_Y^{\proddist T}}} \\
 \leq \P[Q_Y^{\proddist T}]{\sum_{t=1}^T\log \frac{1}{Q_Y(Y_t)}\geq \gamma} +\sqrt{ \frac{2^\gamma}{M_1M_2M_3}}.
\end{multline}
By choosing
\begin{align}
\log M_1 + \log M_2 + \log M_3 &= \lceil (1+\zeta) \avgH{Y}T\rceil\\
\gamma &= \left(1+\frac{\zeta}{2}\right) \avgH{Y} T
\end{align}
and using Hoeffding's inequality \cite{Hoeffding2}, with $\mu_Y \eqdef \min_{y:Q_Y(y) > 0} Q_Y(y)$, we obtain
\begin{multline}
\P[Q_Y^{\proddist T}]{\sum_{t=1}^T\log \frac{1}{Q_Y(Y_t)}\geq \gamma} + \sqrt{\frac{2^\gamma}{M_1M_2M_3}}\\
\begin{split}
 &\leq \exp\left(-\frac{\zeta^2\avgH{Y}^2T}{2\log^2(\mu_Y)}\right) + 2^{-\frac{\zeta}{2}\avgH{Y}T }\\
&\leq 2^{-\xi T},
\end{split}
\end{multline}
for $\xi >0 $ small enough. 

Since $W_1$ and $W_2$ are classical, we can write 
\begin{align}
\rho^{W_1W_2\mathbf{E}} = \frac{1}{M_1 M_2}\sum_{w_1, w_2} \ket{w_1w_2}\bra{w_1w_2} \otimes \rho_{w_1w_2}^{\mathbf{E}},
\end{align}
to upper-bound $\E[F]{\|\rho^{\mathbf{E}W_1W_2} - \rho^{\mathbf{E}}\otimes\rho^{W_1W_2 } \|_1}$, we apply Lemma~\ref{lm:quantum-resolve} and obtain
\begin{multline}
\E[F]{\|\rho^{W_1W_2\mathbf{E}} - \rho^{W_1W_2}\otimes \pr{\widetilde{\rho}^{E}}^{\proddist T} \|_1}\\
\begin{split}
 &= \frac{1}{M_1M_2}\sum_{w_1,w_2}\E[F]{\|\rho^{\mathbf{E}}_{w_1, w_2} -  \pr{\widetilde{\rho}^{E}}^{\proddist T} \|_1}\\
&\leq \sqrt{2^{\gamma s + T \phi(s)}} + \sqrt{\frac{2^\gamma \nu}{M_3}},
\end{split}
\end{multline}
where $\nu$ is the number of distinct eigenvalues of $\pr{\widetilde{\rho}^E}^{\proddist T}$, and 
\begin{align}
\phi(s) = \log \pr{\sum_y Q_Y(y) \tr{\pr{\widetilde{\rho}_y^E}^{1-s} \pr{\widetilde{\rho}^E}^s }}.
\label{eq:phi-res-def}
\end{align}
Upon choosing
\begin{align}
\log M_3 &= \lfloor  I(Q_Y, \widetilde{\rho}_y^E)T + \zeta  \alpha_T T \rfloor,\\
\gamma &=I(Q_Y, \widetilde{\rho}_y^E)T+ \frac{\zeta}{2} \alpha_T T,
\end{align}
we obtain
\begin{multline}
\sqrt{2^{\gamma s + T \phi(s)}} + \sqrt{\frac{2^\gamma \nu}{M_3}} \\
\begin{split}
&\leq \sqrt{2^{s\alpha_T T\pr{\frac{I(Q_Y, \widetilde{\rho}_y^E)}{\alpha_T} + \frac{\zeta}{2} + \frac{\phi(s)}{s\alpha_T}}}} + \sqrt{2^{-\frac{\zeta}{2} \alpha_T T} \nu}\\
&\stackrel{(a)}{\leq} \sqrt{2^{s\alpha_T T\pr{\frac{I(Q_Y, \widetilde{\rho}_y^E)}{\alpha_T} + \frac{\zeta}{2} + \frac{\phi(s)}{s\alpha_T}}}} \\
&\phantom{=}+ \sqrt{2^{-\frac{\zeta}{2} \alpha_T T} (T+1)^{\dim \calH^E}}\\
&\leq \sqrt{2^{s\alpha_T T\pr{\frac{I(Q_Y, \widetilde{\rho}_y^E)}{\alpha_T} + \frac{\zeta}{2} + \frac{\phi(s)}{s\alpha_T}}}} + \frac{1}{2}2^{-\xi \alpha_T  T},\label{eq:res-bound}
\end{split}
\end{multline}
where $(a)$ follows from \cite[Lemma 3.7]{hayashi2006quantum}. We now introduce  the following technical lemma to simplify the above expression.
\begin{lemma}
\label{lm:exponent-res}
Suppose $s<0$; there exists a constant $B\geq0$ such that for $T$ large enough and $|s|$ small enough, we have
\begin{align}
\phi(s) > -I(Q_Y, \widetilde{\rho}_y^E)s - B(\alpha_T s^2 - s^3).
\end{align}
\end{lemma}
\begin{proof}
See Appendix~\ref{sec:error-exponent}.
\end{proof}
Applying Lemma~\ref{lm:exponent-res} to \eqref{eq:res-bound}, we obtain
\begin{multline}
\sqrt{2^{s\alpha_T T\pr{\frac{I(Q_Y, \widetilde{\rho}_y^E)}{\alpha_T} + \frac{\zeta}{2} + \frac{\phi(s)}{s\alpha_T}}}}\\
\begin{split} 
&\leq \sqrt{2^{s\alpha_T T\pr{\frac{I(Q_Y, \widetilde{\rho}_y^E)}{\alpha_T} + \frac{\zeta}{2} + \frac{-I(Q_Y, \widetilde{\rho}_Y^E)s - B(\alpha_T s^2 - s^3)}{s\alpha_T}}}} \\
&= \sqrt{2^{s\alpha_T T\pr{ \frac{\zeta}{2} + \frac{ B(\alpha_T s - s^2)}{\alpha_T}}}} 
\end{split}
\end{multline}
By choosing $s = o(\sqrt{\alpha_T}) \cap \omega(\frac{\log T}{T\alpha_T})$\footnote{To find such $s$, it is required that $\sqrt{\alpha_T} =  \omega(\frac{\log T}{T\alpha_T})$ or equivalently $\alpha_T = \omega\pr{\pr{\frac{\log T}{T}}^{\frac{2}{3}}}$}, the above expression goes to zero faster than any polynomial.
Therefore, for a random encoder, we have
\begin{align}
\E[F]{\P{W_1 \neq \widehat{W}_1} } &\leq 2^{-\xi \alpha_T T}\label{eq:rel-aux}\\
\E[F]{\V{P^a_{\mathbf{Y}}; Q_Y^{\proddist T}}} &\leq 2^{-\xi T}\label{eq:y-app-aux}\\
\E[F]{\|\rho^{W_1W_2\mathbf{E}} -  \pr{\widetilde{\rho}^{E}}^{\proddist T}\otimes\rho^{W_1W_2} \|_1} &\leq 2^{-\omega(\log T)},\label{eq:secrecy-aux}
\end{align}
if
\begin{align}
\log M_3 &= \lfloor  (1+\zeta)I(Q_Y, \widetilde{\rho}_y^E) T \rfloor,\\
\log M_1 + \log M_2 + \log M_3 &= \lceil (1+\zeta) H(Q_Y) T \rceil,\\
\log M_1 +\log M_3&= \lfloor (1-\zeta)I(Q_Y, Q_{X|Y})T\rfloor.
\end{align}
Upon defining the events
\begin{align}
\calE_1 &\eqdef \{\P{W_1 \neq \widehat{W}_1} \leq 4\times2^{-\xi \alpha_T T}\},\\
\calE_2 &\eqdef \{\V{P^a_{\mathbf{Y}}; Q_Y^{\proddist T}} \leq 4\times2^{-\xi T}\},\\
\calE_3 &\eqdef \{\|\rho^{W_1W_2\mathbf{E}} -  \pr{\widetilde{\rho}^{E}}^{\proddist T}\otimes\rho^{W_1W_2} \|_1 \leq 4\times2^{-\omega(\log T)}\},
\end{align}
and using Markov inequality, we have
\begin{multline}
\P[F]{\calE_1\cap \calE_2\cap \calE_3} \\
\begin{split}
&\geq 1 - \P[F]{\calE_1^c} - \P[F]{\calE_2^c} - \P[F]{\calE_3^c}\\
 &\geq 1- \frac{\E[F]{\P{W_1 \neq \widehat{W}_1} }}{2^{-\xi \alpha_T T}} - \frac{\E[F]{\V{P^a_{\mathbf{Y}}; Q_Y^{\proddist T}}}}{ 4\times2^{-\xi T}}\\
 &\phantom{====}-\frac{\E[F]{\|\rho^{W_1W_2\mathbf{E}} -  \pr{\widetilde{\rho}^{E}}^{\proddist T}\otimes\rho^{W_1W_2} \|_1}}{4\times2^{-\omega(\log T)}}\\
 &\geq \frac{1}{4}.
 \end{split}
\end{multline}
Therefore, there exists a realization $f$ of $F$ with
\begin{align}
\P{W_1 \neq \widehat{W}_1} &\leq 4\times2^{-\xi \alpha_T T},\\
\V{P^a_{\mathbf{Y}}; Q_Y^{\proddist T}} &\leq 4\times2^{-\xi T},\\
\|\rho^{W_1W_2\mathbf{E}} -  \pr{\widetilde{\rho}^{E}}^{\proddist T}\otimes\rho^{W_1W_2} \|_1 &\leq 4\times2^{-\omega(\log T)}.
\end{align}
\end{proof}

\subsection{Proof of Theorem~\ref{th:main_quantum_passive}}
\label{sec:thm-passive-proof}
Using the likelihood encoder technique, we first prove that 
\begin{multline}
C_{\textnormal{qck}} \geq \sqrt{\frac{2}{\chi_2\pr{\widetilde{\rho}^E_1\| \widetilde{\rho}^E_0}}} (\D{\widetilde{\rho}^{BE}_1}{\widetilde{\rho}^{BE}_0} \\- \D{\widetilde{\rho}^E_1}{\widetilde{\rho}^E_0}- \D{\widetilde{\rho}^{BE}_1}{\widetilde{\rho}^B_1 \otimes \widetilde{\rho}^E_1}).
\end{multline}
Consider a specific code for the auxiliary problem and let $\widetilde{\rho}^{\mathbf{ABE}W_1W_2\widehat{W}_1} $ be the corresponding induced joint quantum state. Because all random variables $W_1$, $W_2$, $\mathbf{X}$, and $\mathbf{Y}$ are classical, we can define their induced joint {PMF} denoted by $\widetilde{P}_{W_1W_2\mathbf{X}\mathbf{Y}}$. We then use the conditional {PMF}s $\widetilde{P}_{W_1W_2|\mathbf{Y}}$ and $\widetilde{P}_{\widehat{W}_1|\mathbf{X}W_2 }$ as the encoder and decoder, respectively, in the main problem resulting in the induced joint quantum state $\widehat{ \rho}^{\mathbf{ABE}W_1W_2W_3\widehat{W}_1}$. By our construction, we can decompose both joint states $\widetilde{\rho}^{\mathbf{ABE}W_1W_2W_3\widehat{W}_1} $ and $ \widehat{\rho}^{\mathbf{ABE}W_1W_2\widehat{W}_1}$  as 
\begin{multline}
\widetilde{ \rho}^{\mathbf{ABE}W_1W_2W_3\widehat{W}_1} = \sum_{w_1,w_2, \widehat{w}_1, \mathbf{y}, \mathbf{x}}  \widetilde{P}_{\mathbf{Y}}(\mathbf{y})\\
 \times \widetilde{P}_{W_1W_2|\mathbf{Y}}(w_1, w_2|\mathbf{y}) Q_{X|Y}^{\proddist T}(\mathbf{x}|\mathbf{y}) \widetilde{P}_{\widehat{W}_1|\mathbf{X}W_2 }(\widehat{w}_1|\mathbf{x}, w_2)\\
 \times \ket{\mathbf{y}\mathbf{x}w_1w_2\widehat{w}_1}\bra{\mathbf{y}\mathbf{x}w_1w_2\widehat{w}_1}\otimes \widetilde{\rho}_{\mathbf{x},\mathbf{y}}^{\mathbf{E}},
\end{multline}
and
\begin{multline}
\widehat{ \rho}^{\mathbf{ABE}W_1W_2W_3\widehat{W}_1} 
= \sum_{w_1,w_2, \widehat{w}_1, \mathbf{y}, \mathbf{x}}  Q_Y^{\proddist T}(\mathbf{y}) \\
\times \widetilde{P}_{W_1W_2|\mathbf{Y}}(w_1, w_2|\mathbf{y}) Q_{X|Y}^{\proddist T}(\mathbf{x}|\mathbf{y})  \widetilde{P}_{\widehat{W}_1|\mathbf{X}W_2 }(\widehat{w}_1|\mathbf{x}, w_2)\\
\times \ket{\mathbf{y}\mathbf{x}w_1w_2\widehat{w}_1}\bra{\mathbf{y}\mathbf{x}w_1w_2\widehat{w}_1}\otimes \widetilde{\rho}_{\mathbf{x},\mathbf{y}}^{\mathbf{E}}.
\end{multline}
Since they differ  only in the distribution of $\mathbf{Y}$, we have
\begin{multline}
\|\widetilde{\rho}^{\mathbf{ABE}W_1W_2W_3\widehat{W}_1}  - \widehat{\rho}^{\mathbf{ABE}W_1W_2W_3\widehat{W}_1} \|_1 \\
\leq 2\V{\widetilde{P}^a_{\mathbf{Y}}; Q_Y^{\proddist T}}
 \stackrel{(a)}{\leq} 2^{-\xi T},
\end{multline}
where $(a)$ follows from \eqref{eq:y-app-aux}. Thus, we  upper-bound the probability of error in the main problem as
\begin{multline}
\P[\widehat{P}]{W_1 \neq W_2}\\
\begin{split}
& \leq \P[\widetilde{P}]{W_1 \neq W_2}  \\
&\phantom{=}+ \|\widetilde{\rho}^{\mathbf{ABE}W_1W_2W_3\widehat{W}_1}  - \widehat{\rho}^{\mathbf{ABE}W_1W_2W_3\widehat{W}_1} \|_1 \\
&\leq 2^{-\zeta \alpha_T T} + 2^{-\zeta T},
\end{split}
\end{multline}
and upper-bound the sum of secrecy and covertness as
\begin{align}
S+ C 
&\eqdef\D{\widehat{\rho}^{W_1 W_2 \mathbf{E}}}{{\rho}_{\text{unif}}^{W_1}\otimes \widehat{\rho}^{W_2\mathbf{E}}}\nonumber\displaybreak[0]\\
&\phantom{=========} + \D{\widehat{\rho}^{W_2\mathbf{E}}}{{\rho}_{\text{unif}}^{W_2}\otimes \rho_0^{\proddist T}}\displaybreak[0]\\
&= \D{\widehat{\rho}^{W_1 W_2 \mathbf{E}}}{{\rho}_{\text{unif}}^{W_1W_2}\otimes \widehat{\rho}^{\mathbf{E}}} + \D{\widehat{\rho}^{\mathbf{E}}}{\rho_0^{\proddist T}}\displaybreak[0]\\
&\stackrel{(a)}{=} \D{\widehat{\rho}^{W_1 W_2 \mathbf{E}}}{{\rho}_{\text{unif}}^{W_1W_2}\otimes \widehat{\rho}^{\mathbf{E}}}\nonumber \displaybreak[0]\\
&\phantom{======} +\frac{1}{2} \alpha_T^2\chi_2(\rho_1^E\|\rho_0^E)T + O(\alpha_T^3T)\displaybreak[0]\\
&\stackrel{(b)}{\leq}  \|\widehat{\rho}^{W_1 W_2 \mathbf{E}} - {\rho}_{\text{unif}}^{W_1W_2}\otimes \widehat{\rho}^{\mathbf{E}} \|_1\nonumber\displaybreak[0]\\
&\times \log \frac{M_1M_2\pr{\dim \calH^E}^T}{\frac{1}{M_1M_2}\lambda_{min}(\widetilde{\rho}^E)^T \|\widehat{\rho}^{W_1 W_2 \mathbf{E}} - {\rho}_{\text{unif}}^{W_1W_2}\otimes \widehat{\rho}^{\mathbf{E}} \|_1 }\nonumber \displaybreak[0]\\
& \phantom{======} +\frac{1}{2} \alpha_T^2\chi_2(\rho_1^E\|\rho_0^E)T + O(\alpha_T^3T)\displaybreak[0]\\
&= \|\widehat{\rho}^{W_1 W_2 \mathbf{E}} - {\rho}_{\text{unif}}^{W_1W_2}\otimes \widehat{\rho}^{\mathbf{E}} \|_1\nonumber\displaybreak[0]\\
&\phantom{===} \times\pr{O(T)-\log  \|\widehat{\rho}^{W_1 W_2 \mathbf{E}} - {\rho}_{\text{unif}}^{W_1W_2}\otimes \widehat{\rho}^{\mathbf{E}} \|_1  }\nonumber\displaybreak[0]\\
&\phantom{======} +\frac{1}{2} \alpha_T^2\chi_2(\rho_1^E\|\rho_0^E)T + O(\alpha_T^3T)\displaybreak[0]\\
&\stackrel{(c)}{\leq}\pr{2^{-\zeta \alpha_T T }+ 2^{-\zeta T}}O(T)\nonumber\displaybreak[0]\\
&\phantom{=====}+\frac{1}{2} \alpha_T^2\chi_2(\rho_1^E\|\rho_0^E)T + O(\alpha_T^3T),
\label{eq:sec_cov}
\end{align}
where $(a)$ follows from \cite[Lemma 7]{Sheikholeslami2016}, $(b)$ follows from Lemma~\ref{lm:div-bound-l1}, and $(c)$ follows from
\begin{multline}
\|\widehat{\rho}^{W_1 W_2 \mathbf{E}} - {\rho}_{\text{unif}}^{W_1W_2}\otimes \widehat{\rho}^{\mathbf{E}} \|_1 \\
\begin{split}
&\leq  \|\widetilde{\rho}^{W_1 W_2 \mathbf{E}} - {\rho}_{\text{unif}}^{W_1W_2}\otimes \widehat{\rho}^{\mathbf{E}} \|_1 
+ \|\widehat{\rho}^{W_1 W_2 \mathbf{E}} - \widetilde{\rho}^{W_1 W_2 \mathbf{E}} \|_1 \\
&\leq 2^{-\zeta \alpha_T T} + 2^{-\zeta T}.
\end{split}
\end{multline}
The throughput of the coding scheme is  lower-bounded by \eqref{eq:passive-throughput} shown below.
\begin{widetext}

\begin{align}
\frac{\log M_1}{\sqrt{T C}} &\geq \frac{ \log M_1}{\sqrt{T \pr{ \pr{2^{-\zeta \alpha_T T }+ 2^{-\zeta T}}O(T)+\frac{1}{2} \alpha_T^2\chi_2(\rho_1^E\|\rho_0^E)T + O(\alpha_T^3T)}}} \\
&\geq \sqrt{\frac{2}{\chi_2(\rho_1^E\|\rho_0^E)}}\frac{ \lfloor (1-\zeta)\avgI{A;B}_{\widetilde{\rho}}T\rfloor - \lceil  \avgI{B;E}_{\widetilde{\rho}} T + \zeta\alpha_T T\rceil }{T\alpha_T(1+o(1))}\\
&= \sqrt{\frac{2}{\chi_2\pr{\widetilde{\rho}^E_1\| \widetilde{\rho}^E_0}}} \pr{\D{\widetilde{\rho}^{BE}_1}{\widetilde{\rho}^{BE}_0} - \D{\widetilde{\rho}^E_1}{\widetilde{\rho}^E_0} - \D{\widetilde{\rho}^{BE}_1}{\widetilde{\rho}^B_1 \otimes \widetilde{\rho}^E_1}} + o(1).\label{eq:passive-throughput}
\end{align}
\end{widetext}
We now turn to the proof of
\begin{align}
C_{\textnormal{qck}} \geq\sqrt{\frac{2}{\chi_2\pr{\widetilde{\rho}^E_1\| \widetilde{\rho}^E_0}}} \pr{ \avgD{\widetilde{\rho}_1^B}{\widetilde{\rho}_0^B}-\avgD{\widetilde{\rho}_1^E}{\widetilde{\rho}_0^E}}.
\end{align}
Note that if $\avgD{\widetilde{\rho}_1^B}{\widetilde{\rho}_0^B} \leq \avgD{\widetilde{\rho}_1^E}{\widetilde{\rho}_0^E}$, the result is trivial. Therefore, we can assume that $\avgD{\widetilde{\rho}_1^B}{\widetilde{\rho}_0^B} > \avgD{\widetilde{\rho}_1^E}{\widetilde{\rho}_0^E}$. Let $M_1$ and $M_2$ be such that
\begin{align}
\log M_1 + \log M_2 &= \lfloor (1-\zeta) I(Q_X, \widetilde{\rho}_y^B) \rfloor,\label{eq:throughput-forward}\\
\log M_2 &= \lceil (1+\zeta) I(Q_X, \widetilde{\rho}_y^E) \rceil.\label{eq:throughput-forward2}
\end{align}
The protocol is then as follows. Alice chooses a random binary string of length $\log M_1 + \log M_2$ and transmits this string through a covert code introduced in \cite{Sheikholeslami2016}. Alice and Bob  subsequently extract the first $\log M_1$ bits of the string as the key. The reliability and covertness proof follows exactly from \cite{Sheikholeslami2016}. For secrecy, note that
\begin{multline}
 \D{\rho^{\mathbf{E}MS^X}}{\rho^{\mathbf{E}M} \otimes \rho^{S^X}_{\text{unif}}} \\
 \begin{split}
&\eqdef \D{\rho^{\mathbf{E}S^X}}{\rho^{\mathbf{E}} \otimes \rho^{S^X}_{\text{unif}}} \\
&=  \frac{1}{M_1} \sum_{w_1 = 1}^{M_1}\D{\rho^{\mathbf{E}}_{w_1}}{\rho^{\mathbf{E}}}.
\end{split}
\end{multline}
Similar to the proof of \eqref{eq:secrecy-aux}, one can show that the above expression is upper-bounded by $2^{-\omega(\log T)}$ provided that $\log M_2 = \lceil (1+\zeta) I(Q_X, \widetilde{\rho}_y^E) \rceil$. Lower-bounding the throughput as in \eqref{eq:passive-throughput}  using \eqref{eq:throughput-forward} and \eqref{eq:throughput-forward2} concludes the proof.

\section{Proof of Theorem~\ref{th:main-collective}}
\label{sec:proof-theorem-unknown}

\subsection{Universal covert communication}
\label{sec:universal-covert-code}
The following theorem shows that knowing  bounds on $\lambda_{\min}(\rho_0^B)$, $\lambda_{\min}(\rho_0^E)$,  $\D{\rho_1^B}{\rho_0^B}$ and $\D{\rho_1^E}{\rho_0^E}$ is all that is required to covertly generate a secret key.
\begin{theorem}
\label{th:universal-covert}
Let $D^B$, $D^E$, $\widetilde{\lambda}^B$, and $\widetilde{\lambda}^E$ be fixed numbers and $\{\alpha_T\}_{T\geq 1}$ be as in Definition~\ref{def:q-covert-process}. For any $\zeta>0$, there exists a sequence of codes $\{\calC_T\}_{T\geq 1}$ such that for all cq-channels $x\mapsto \rho_x^{BE}$ satisfying
\begin{align}
\label{eq:db-cond}
\D{\rho_1^B}{\rho_0^B} &\geq D^B,\\
\label{eq:dw-cond}
\D{\rho_1^E}{\rho_0^E} &\leq D^E,\\
\label{eq:lb-cond}
\lambda_{\min}(\rho_0^B) &\geq \widetilde{\lambda}^B,\\
\label{eq:lw-cond}
\lambda_{\min}(\rho_0^E) &\geq \widetilde{\lambda}^E,
\end{align}
we have
\begin{align}
P_e &\leq 2T^{-5},\\
S &\leq L_1T^{-4},\\
C &\leq \frac{\alpha_T^2\chi_2(\rho_1^E\|\rho_0^E)}{2}T + L_1T^{-4} + 2\sqrt{L_1}\log \frac{2}{\widetilde{\lambda}^E}T^{-1}\nonumber\\
&\phantom{=}+ L_2\alpha_T^3T ,\\
\log M &= (1-2\zeta)(D^B - D^E)\alpha_T T,
\end{align}
where $L_1, L_2>0$ depend on the $\dim \calH^E$ and $\widetilde{\lambda}^E$.
\end{theorem}
The remainder of this section is dedicated to the proof of the above result. We first adapt a result from \cite{bjelakovic2009classical}, which shows that for any class of cq-channels, there exists a \emph{finite} class of cq-channels that approximates the main class with high precision.
\begin{lemma}
\label{lm:cont-disc-compound-cq}
Consider a compound cq-channel $x \mapsto \rho_x^B(\theta)$ where $x \in \cal X$, $\rho_x^B \in \calD(\calH)$, $\calH$ is a $d$-dimensional Hilbert space, and $\theta\in\Theta$ is an arbitrary index set. There exists a constant $K>0$  that depends only on $d$ such that for  all $T\in\mathbb{N}$, there exists another compound cq-channel $x\mapsto\rho_x^B(\widetilde{\theta})$ with $x\in\calX$, $\rho_x^B \in \calD(\calH)$, and $\widetilde{\theta} \in \widetilde{\Theta}$ such that
\begin{enumerate}
\item the set $\widetilde{\Theta}$ is finite, i.e.,
\begin{align}
\label{eq:first-cond}
|\widetilde{\Theta}| \leq K^{|\calX|} T^{6|\calX| d^2};
\end{align}
\item for all $\theta\in\Theta$, there exists a $\widetilde{\theta}\in\widetilde{\Theta}$ such that for all $\mathbf{x} \in \calX^T$, we have
\begin{align}
\label{eq:second-cond}
\|\rho_{\mathbf{x}}^{\mathbf{B}}(\theta) -\rho_{\mathbf{x}}^{\mathbf{B}}(\widetilde{\theta})  \|_1 \leq T^{-5};
\end{align}
\item 
for all {PMFs} $P_X$ over $\calX$, we have
\begin{align}
\label{eq:third-cond}
 \min_{\widetilde{\theta}\in \widetilde{\Theta}} I(P_X, \rho_x^B(\widetilde{\theta})) \geq \inf_{\theta \in \Theta} I(P_X, \rho_x^B(\theta)) -2T^{-6}\log\pr{T^6 d}.
\end{align}
\end{enumerate}
\end{lemma}
\begin{proof}
We modify the proof provided in \cite{bjelakovic2009classical} to derive a tighter upper-bound on the approximation error of the new compound channel at the expense of increasing its size. By \cite[Theorem 5.5]{bjelakovic2009classical}, for all $\kappa > 0$, there exists a partition of all cq-channels from $\calX$ to $\calD(\calH)$  denoted by $\Pi = \{\pi_1, \cdots, \pi_n\}$ such that $n\leq K^{|\calX|} \kappa^{-|\calX|d^2}$, where $K$ only depends on the dimension of $\calH$, d, and the diameter of $\Pi$ is at most $\kappa$, i.e., for all $i\in\intseq{1}{n}$, for any two channels $x\mapsto\rho_x^B$ and $x\mapsto\widetilde{\rho}_x^B$ in $\pi_i$, for any $x\in\calX$, we have $\|\rho_x^B-\widetilde{\rho}_x^B\|_1 \leq \kappa$. Setting $\kappa = T^{-6}$, this implies that there exists a partition of size at most $K^{|\calX|}T^{6|\calX|d^2}$ and diameter at most $T^{-6}$. We construct the new compound cq-channel $x\mapsto \rho_x^B(\widetilde{\theta})$ by selecting an arbitrary channel from each $\pi_i$ whose intersection with $\{x\mapsto \rho_x^B(\theta): \theta\in\Theta\}$ is non-empty.
 We now show that this compound channel satisfies the conditions mentioned in the statement of the lemma. Since we select at most one channel from each $\pi_i$, $|\widetilde{\Theta}|\leq n \leq  K^{|\calX|} T^{6|\calX| d^2}$, and thus, we have \eqref{eq:first-cond}. To prove \eqref{eq:second-cond}, consider any $\theta\in\Theta$. By our construction, there should be a $\widetilde{\theta}\in\widetilde{\Theta}$  such that $x\mapsto \rho_x^B(\widetilde{\theta})$ and $x\mapsto \rho_x^B(\theta)$ belong to the same $\pi_i$. Therefore, for any $\mathbf{x} \in \calX^T$, we have
\begin{multline}
\|\rho^{\mathbf{B}}_{\mathbf{x}}(\theta) -\rho^{\mathbf{B}}_{\mathbf{x}}(\widetilde{\theta})  \|_1\\
\begin{split}
&= \| \rho_{x_1}^B(\theta) \otimes \cdots \otimes \rho_{x_T}^B(\theta) -\rho_{x_1}^B(\widetilde{\theta}) \otimes \cdots \otimes \rho_{x_T}^B(\widetilde{\theta})  \|_1\\
&\leq \sum_{t=1}^T \|\rho_{x_t}^B(\theta)  - \rho_{x_t}^B(\widetilde{\theta})  \|_1\\
&\stackrel{(a)}{\leq} T^{-5},
\end{split}
\end{multline}
where $(a)$ follows since $x\mapsto \rho_x^B(\theta)$ and $x\mapsto \rho_x^B(\widetilde{\theta})$ belong to the same $\pi_i$, and the diameter of the partition is less than $T^{-6}$. Finally, let $P_X$ be any {PMF} over $\calX$; to lower-bound $ \min_{\widetilde{\theta}\in \widetilde{\Theta}} I(P_X, \rho_x^B(\widetilde{\theta}))$ as in \eqref{eq:third-cond}, take any  $\widetilde{\theta}\in\widetilde{\Theta}$ and consider $\theta$ such that $x\mapsto \rho_x^B(\theta)$ and $x\mapsto \rho_x^B(\widetilde{\theta})$ belong to the same $\pi_i$. To complete the lemma, it is enough to show that
\begin{align}
I(P_X, \rho_x^B(\widetilde{\theta})) \geq I(P_X, \rho_x^B({\theta})) - 2T^{-6}\log\pr{T^6 d}.
\end{align}
To this end, we have
\begin{multline}
I(P_X, \rho_x^B(\widetilde{\theta})) \\
\begin{split}
&= H\pr{\sum_x P_X(x) \rho_x^B(\widetilde{\theta})} -\sum_x P_X(x)  H\pr{\rho_x^B(\widetilde{\theta})}\\
&\stackrel{(a)}{\geq} H\pr{\sum_x P_X(x) \rho_x^B({\theta})} -\sum_x P_X(x)  H\pr{\rho_x^B({\theta})}\\
&\phantom{=} - 2T^{-6}\log\pr{T^6 d}\\
&= I(P_X, \rho_x^B(\widetilde{\theta})) -  2T^{-6}\log\pr{T^6 d},
\end{split}
\end{multline}
where $(a)$ follows from Fannes's inequality which states that for any two $\rho$ and $\sigma$ in $\calD(\calH)$, if $\|\rho - \sigma\|_1\leq \delta \leq e^{-1}$, we have $|H(\rho) - H(\sigma)| \leq \delta \log(d\delta^{-1})$.
\end{proof}
\subsubsection{Universal reliability result}
We next prove a universal reliability result suitable for covert communications. Note that we cannot use the result of \cite{bjelakovic2009classical} directly since the  input distribution  used to analyze covert communications changes with the block-length. Indeed, our inspection of the proof of \cite{bjelakovic2009classical} suggests that the technique cannot be adapted for the covert case since the the penalty arising from the approximation of a class of channels dominates the number of bits that one can transmit covertly, which scales as $O\pr{\sqrt{T}}$. Therefore, we use a different approach based on the quantum universal decoder introduced by Hayashi in \cite{hayashi2009universal}. We first state the following lemma from \cite{bjelakovic2009classical} which is a general achievability result for cq-channels.
\begin{lemma} [{\cite[Theorem 5.4]{bjelakovic2009classical}}]
\label{lm:one-shot-universal-cq-reliability}

Let $x\mapsto \rho_x^B$ be any cq-channel with input set $\calX$, and $M$ be a positive integer. For all $x$, let $\Gamma_x$ be an operator on $\calH^B$ with $0\leq \Gamma_x \leq I$, and $P_X$ be a probability distribution over $\calX$. If $F:\intseq{1}{M} \to \calX$ is a random encoder whose codewords are {iid} according to $P_X$,  there exists a ``universal" decoder corresponding to a POVM $\{\Lambda_w\}_{w=1}^M$ depending on the operators $\Gamma_x$ and the encoder $F$ (not on the channel) such that the average probability of error satisfies
\begin{multline}
\label{eq:cq-one-shot-rel}
\E[F]{\sum_{w=1}^M \pr{1-\tr{\rho_{F(w)}^B\Lambda_w}}} \\
\leq 2\sum_x P_X(x) \tr{\rho_x^B\Gamma_x } + 4 M \sum_{x}  P_X(x) \tr{\rho^B \Gamma_{x}},
\end{multline}
where $\rho^B \eqdef \sum_x P_X(x) \rho_x^B$.
\end{lemma}

We next consider a stationary memoryless cq-channel  $x\mapsto \rho_x^B$ with $T$ channel uses and for each codeword $\mathbf{x}\in\calX^T$, we aim to construct the operator $\Gamma_{\mathbf{x}}$ \emph{independent} of the channel such that we would be able to upper-bound the right hand side of \eqref{eq:cq-one-shot-rel}. We shall follow the approach in \cite{hayashi2009universal}, which is based on the following result from representation theory.
\begin{theorem}[Schur-Weyl Duality]
Let $H$ be a $d$-dimensional Hilbert space over $\mathbb{C}$. For any $T\geq 1$, we have the decomposition
\begin{align}
H^{\proddist T} = \mathop{\oplus}_{\mathbf{t} \in Y_T^d} \calU_{\mathbf{t}} \otimes \calV_{\mathbf{t}},
\end{align}
where $Y_T^d \eqdef \{(t_1, \cdots, t_d)\in \mathbb{Z}^d: t_1 \geq \cdots \geq t_d \geq 0, \sum_{i=1}^d t_i = T\}$, $\calU_{\mathbf{t}}$ is an irreducible representation of $\text{SU}(d)$,  and $\calV_{\mathbf{t}}$ is an irreducible representation of the $T^{th}$ order symmetric group.
\end{theorem}

In \cite{hayashi2009universal}, for all $\mathbf{t}\in Y_T^d$ and all $T$, the author introduced several quantum states that satisfy universal matrix inequalities for all density matrices and  all cq-channels. Since those quantum states are a substantial ingredient of the construction of our universal decoder, we state here their definition and properties from \cite{hayashi2009universal}.
\begin{definition}
For $\mathbf{t} \in Y_Y^d$, let $I_{\mathbf{t}}$ be the projection onto the subspace $\calU_{\mathbf{t}} \otimes \calV_{\mathbf{t}}$. Define
\begin{align}
\sigma_\mathbf{t} &\eqdef \frac{1}{\textnormal{dim} (\calU_{\mathbf{t}} \otimes \calV_{\mathbf{t}})} I_{\mathbf{t}}\\
\sigma_{U, T} &\eqdef \sum_{\mathbf{t} \in Y_T^d} \frac{1}{|Y_T^d|} \sigma_{\mathbf{t}}. 
\end{align}
Moreover, for $\mathbf{x}' = (0, \cdots, 0, 1, \cdots, 1) \in \calX^T$ with $\wt{\mathbf{x}'} = m$, we define $\sigma_{\mathbf{x}'} \eqdef \sigma_{U, T-m} \otimes \sigma_{U, m}$. For any $\mathbf{x} \in \calX^T$ with $\wt{\mathbf{x}} = m$, we suppose $\mathbf{x} = \pi \mathbf{x}'$ where $\pi$ is a permutation of $T$ elements and define $\sigma_{\mathbf{x}} \eqdef U_{\pi} \sigma_{\mathbf{x}'} U_{\pi}^\dagger$ where $U_{\pi}$ is the unitary representation of $\pi$. 
\end{definition}
\begin{lemma}
\label{lm:universal-dec-prop}
For any density matrix $\rho$ on $\calH$ and any cq-channel $x\mapsto \rho_x^B$, we have
\begin{align}
T^{\frac{d(d-1)}{2}} |Y_T^d| \sigma_{U, T} &\succeq \rho^{\proddist T},\\
T^{|\calX|\frac{d(d-1)}{2}} |Y_T^d| \sigma_\mathbf{x} \succeq \rho_{\mathbf{x}}^{\mathbf{B}}.
\end{align}
\end{lemma}
\begin{proof}
See \cite[Equation (6) and (7)]{hayashi2009universal}.
\end{proof}

\begin{lemma}
\label{lm:universal-reliability}
Fix $\zeta$ and $\widetilde{\lambda}$ in $]0, 1[$. Let $x\mapsto \rho_x^{B}(\theta)$ be a compound cq-channel with $\theta\in\Theta$ and $x \in \calX = \{0, 1\}$ such that  $\lambda_{\min}(\rho_0^B) \geq \widetilde{\lambda}$ for all $\theta\in \Theta$.  For a fixed $T$, let 
\begin{align}
\log M \eqdef \lfloor (1-\zeta) \alpha_T  \inf_{\theta \in \Theta} \D{\rho_1^B(\theta)}{ \rho_0^B(\theta)} T\rfloor,
\end{align}
and $F:\intseq{1}{M} \to \calX^T$ be a random encoder such that $F(1), \cdots, F(M)$ are {iid} according to $P_X^{\proddist T}$ with $P_X=$ Bernoulli($\alpha_T$) and $\alpha_T$ as in Definition~\ref{def:q-covert-process}.

Then, there exists $T_0$ that depends only on $\dim \calH$, $\zeta$, and $\widetilde{\lambda}$ such that for all $T\geq T_0$,
\begin{align}
\P[F]{\forall \theta\in\Theta, P_e(\theta) \leq 2T^{-5}} \geq  \frac{2}{3}.
\end{align}
\end{lemma}
\begin{proof}
We first consider the compound cq-channel $x\mapsto \rho_x^B(\widetilde{\theta})$ obtained by applying Lemma~\ref{lm:cont-disc-compound-cq} to the compound cq-channel $x\mapsto \rho_x^B(\theta)$. By Lemma~\ref{lm:one-shot-universal-cq-reliability}, for each $\widetilde{\theta}\in \widetilde{\Theta}$, the expectation of the probability of error with respect to random coding is upper-bounded by
\begin{align}
\label{eq:rel-compound-n}
2\sum_{\mathbf{x}}P_X^{\proddist T}(\mathbf{x}) \tr{\rho_{\mathbf{x}} \Gamma_{\mathbf{x}}} +4M \sum_{\mathbf{x}} P_X^{\proddist T}(\mathbf{x}) \tr{\pr{\rho^{B}}^{\proddist T} \Gamma_{\mathbf{x}}},
\end{align}
where $\Gamma_{\mathbf{x}} \eqdef \{\sigma_{\mathbf{x}} - \gamma \sigma_{U, T} \succeq 0\}$ \cite{hayashi2009universal}. To upper-bound the first term in \eqref{eq:rel-compound-n}, we split the summation into three parts based on the weight of the codeword $\mathbf{x}$. In particular, for two thresholds $w_\ell < T\alpha_T < w_u \leq  2T\alpha_T$, we obtain with a  Chernoff bound
\begin{align}
\sum_{\mathbf{x}: \textnormal{wt}({\mathbf{x}}) < w_\ell}P_X^{\proddist T}(\mathbf{x}) \tr{\rho_{\mathbf{x}} \Gamma_{\mathbf{x}}}
&\leq \sum_{\mathbf{x}: \textnormal{wt}({\mathbf{x}}) < w_\ell}P_X^{\proddist T}(\mathbf{x})\\
&= \P[P_X^{\proddist T}]{\wt{\mathbf{X}} \leq w_\ell}\\
&\leq e^{-\frac{1}{2} \pr{1-\frac{w_\ell}{T\alpha_T}}^2 T\alpha_T},
\end{align}
and analogously
\begin{align}
\sum_{\mathbf{x}: \textnormal{wt}({\mathbf{x}}) > w_u}P_X^{\proddist T}(\mathbf{x}) \tr{\rho_{\mathbf{x}} \Gamma_{\mathbf{x}}}
&\leq e^{-\frac{1}{3} \pr{\frac{w_u }{T\alpha_T} - 1}^2 T\alpha_T}.
\end{align}
To upper-bound the remaining terms, for $Q_X \sim$ Bernoulli($p$), let us define
\begin{multline}
\phi(s, p) \eqdef -(1-s) \\
\times \log\pr{\tr{\pr{\sum_x Q_X(x) \pr{\rho_x^B(\widetilde{\theta})}^{1-s}}^{\frac{1}{1-s}}}}.
\end{multline}
Then, by \cite[Equation (18)]{hayashi2009universal}, we have
\begin{multline}
\sum_{\mathbf{x}: w_\ell \leq \textnormal{wt}({\mathbf{x}}) \leq w_u}P_X^{\proddist T}(\mathbf{x}) \tr{\rho_{\mathbf{x}} \Gamma_{\mathbf{x}}}\\
\begin{split}
&{\leq} \sum_{\mathbf{x}: w_\ell \leq \textnormal{wt}({\mathbf{x}}) \leq w_u}P_X^{\proddist T}(\mathbf{x}) \min_{s\in[0, 1]} (T+1)^{d+sd(d-1)}\\
&\phantom{=========}\times|Y_T^d|^{2s} \gamma^s e^{-T\phi\pr{s, \frac{\textnormal{wt}({\mathbf{x}})}{T}}}\\
&\leq (T+1)^{d^2}|Y_T^d|^{2} \max_{w\in\intseq{w_\ell}{w_u}} \min_{s\in[0, 1]} \gamma^se^{-T\phi\pr{s, \frac{w}{T}}}.\label{eq:mid-weight-phi}
\end{split}
\end{multline}
We introduce a result bounding $\phi(s, p)$ for small $s$ and $p$.
\begin{lemma}
\label{lm:exponent-reliability}
For all $\widetilde{\lambda}, \widetilde{s},\widetilde{p} \in [0, 1]$, there exists a universal constant $B>0$ such that for all cq-channels $x\mapsto \rho^B_x$ with $\lambda_{\min}(\rho^B_0) \geq \widetilde{\lambda}$ and for  all $s\leq \widetilde{s}$ and $p\leq \widetilde{p} $ , we have 
\begin{align}
\phi(s, p) \geq s I(p) - B(p s^2 + s^3),
\end{align}
where $I(p) \eqdef I(Q_X, \rho_x)$ with $Q_X \sim $ Bernoulli($p$). Furthermore, for small enough $p$, we have
\begin{align}
I(p) \geq p \D{\rho_1^B}{\rho_0^B} - Bp^2.
\end{align} 
\end{lemma}
\begin{proof}
See Appendix~\ref{sec:error-exponent}.
\end{proof}
Applying Lemma~\ref{lm:exponent-reliability} to \eqref{eq:mid-weight-phi}, we obtain for all $s$ small enough,
\begin{multline}
 (T+1)^{d^2}|Y_T^d|^{2} \max_{w\in\intseq{w_\ell}{w_u}} \gamma^se^{-T\phi\pr{s, \frac{w}{T}}} \\
 \begin{split}
 &\leq  (T+1)^{d^2}|Y_T^d|^{2} \max_{w\in\intseq{w_\ell}{w_u}} \gamma^se^{-T\pr{s I\pr{\frac{w}{T}} -B\pr{\frac{w}{T} s^2 + s^3}}} \\
  &\leq  (T+1)^{d^2}|Y_T^d|^{2}\\
  &\phantom{=}\times \max_{w\in\intseq{w_\ell}{w_u}} \gamma^se^{-T\pr{\frac{w}{T}s\D{\rho_1^B}{\rho_0^B}  -B\pr{\frac{w^2}{T^2}+\frac{w}{T} s^2 + s^3}}} \\
    &\leq  (T+1)^{d^2}|Y_T^d|^{2}\\
  &\phantom{=}\times \gamma^se^{-T\pr{\frac{w_\ell}{T}s\D{\rho_1^B(\widetilde{\theta})}{\rho_0^B(\widetilde{\theta})}  -B\pr{\frac{w_u^2}{T^2}+\frac{w_u}{T} s^2 + s^3}}}.
 \end{split}
\end{multline}
To upper-bound the second term in \eqref{eq:rel-compound-n}, we use the operator inequality $A\{A\succeq 0\} \succeq 0$ for any Hermitian operator $A$. Hence, we have for all $\mathbf{x}$
\begin{align}
(\sigma_{\mathbf{x}} - \gamma \sigma_{U, T}) \Gamma_{\mathbf{x}} \succeq 0.
\end{align}
This implies that
\begin{multline}
\left(\sigma_{\mathbf{x}} - \gamma \sigma_{U, T} + \frac{\gamma}{T^{d(d-1)}|Y_T^d| } \pr{\rho^B}^{\proddist T}\right.\\
\left. -  \frac{\gamma}{T^{d(d-1)}|Y_T^d| } \pr{\rho^B}^{\proddist T}\right) \Gamma_{\mathbf{x}} \succeq 0.
\end{multline}
Thus, we have
\begin{multline}
\tr{\pr{\sigma_{\mathbf{x}} -   \frac{\gamma}{T^{d(d-1)}|Y_T^d| } \pr{\rho^B}^{\proddist T}} \Gamma_{\mathbf{x}}}\\
\begin{split}
& \geq \tr{ \pr{ \gamma \sigma_{U, T} - \frac{\gamma}{T^{d(d-1)}|Y_T^d| } \pr{\rho^B}^{\proddist T} } \Gamma_{\mathbf{x}}} \\
&\stackrel{(a)}{\geq} 0,
\end{split}
\end{multline}
where $(a)$ follows since by Lemma \ref{lm:universal-dec-prop}, $ \pr{ \gamma \sigma_{U, T} - \frac{\gamma}{T^{d(d-1)}|Y_T^d| } \pr{\rho^B}^{\proddist T} } \succeq 0$. Accordingly, we conclude that
\begin{align}
 \sum_{\mathbf{x}} P_X^{\proddist T}(\mathbf{x}) \tr{\pr{\rho^{B}}^{\proddist T} \Gamma_{\mathbf{x}}} 
 &\leq \frac{T^{d(d-1)}|Y_T^d|}{\gamma}.
\end{align}
Substituting the derived upper-bounds in \eqref{eq:rel-compound-n}, we obtain

\begin{multline}
\E[F]{P_e(\widetilde{\theta})} \leq 2\left(e^{-\frac{1}{2} \pr{1-\frac{w_\ell}{T\alpha_T}}^2 T\alpha_T}\right.\\
 \left.+ e^{-\frac{1}{3} \pr{\frac{w_u }{T\alpha_T} - 1}^2 T\alpha_T}+(T+1)^{d^2}|Y_T^d|^{2} \gamma^s\right.\\
 \left. \times e^{-T\pr{\frac{w_\ell}{T}s\D{\rho_1^B(\widetilde{\theta})}{\rho_0^B(\widetilde{\theta})}  -B\pr{\frac{w_u^2}{T^2}+\frac{w_u}{T} s^2 + s^3}}}\right) \\
 + 4M \frac{T^{d(d-1)}|Y_T^d|}{\gamma}.
\end{multline}
By choosing 
\begin{align}
w_\ell &= T\alpha_T - (T\alpha_T)^{\frac{2}{3}},\\
w_u &=  T\alpha_T + (T\alpha_T)^{\frac{2}{3}},\\
\gamma &= \left \lfloor \pr{1-\frac{\zeta}{2}} \alpha_T  \inf_{\theta \in \Theta} \D{\rho_1^B(\theta)}{ \rho_0^B(\theta)} T\right\rfloor, \\
s &= o(\sqrt{\alpha_T}) \cap \omega(\frac{\log T}{T\alpha_T}),
\end{align}
we obtain
\begin{align}
\label{eq:ub-e-pe}
\E[F]{P_e(\widetilde{\theta})}\leq 2^{-\omega(\log T)},
\end{align}
where the term $-\omega(\log T)$ depends on $\widetilde{\lambda}$, $\zeta$, and $\dim \calH$.
By Markov's inequality and the union bound, we have
\begin{align}
\P[F]{\forall \widetilde{\theta}\in\widetilde{\Theta}, P_e(\widetilde{\theta})\leq 3|\widetilde{\Theta}|\E[F]{P_e(\widetilde{\theta})}} \geq \frac{2}{3}.
\end{align}
By Lemma~\ref{lm:cont-disc-compound-cq}, $|\widetilde{\Theta}|$ is upper-bounded by a polynomial in $T$. This together with \eqref{eq:ub-e-pe} implies that  $3|\widetilde{\Theta}|\E[F]{P_e(\widetilde{\theta})}  = 2^{-\omega(\log T)}$. Finally, by Lemma~\ref{lm:cont-disc-compound-cq}, for all $\theta\in\Theta$, there exists $\widetilde{\theta}\in\widetilde{\Theta}$ such that $P_e(\theta) \leq P_e(\widetilde{\theta}) + T^{-5}$. Thus, for large enough $T$, we have
\begin{align}
\P[F]{\forall \theta\in\Theta, P_e(\theta) \leq 2T^{-5}} \geq  \frac{2}{3}.
\end{align}
\end{proof}
\subsubsection{Universal resolvability result}
We next prove an asymptotic resolvability result for covert distributions.  
\begin{lemma}
\label{lm:aym-resolvability}
Fix $\widetilde{\lambda}$ and $\zeta$ in $]0, 1[$. Consider a cq-channel $x\mapsto \rho_x^E$ with $x\in\calX = \{0, 1\}$ such that $\lambda_{\min}(\rho_0^E) \geq \widetilde{\lambda}$. Let $P_X$ be the covert distribution in Definition~\ref{def:q-covert-process}, $M'$ be an integer satisfying
\begin{align}
M' \geq \lceil (1+\zeta) \alpha_T \D{\rho_1^E}{\rho_0^E} T\rceil,
\end{align}
and $F:\intseq{1}{M'} \to \calX^T$ be a random encoder such that all codewords are distributed according to $P_X^{\proddist T}$ independently. Then, we have
\begin{align}
\E[F]{\left\|\widehat{\rho}^{\mathbf{E}}-\pr{\rho^E}^{\proddist T}\right\|_1} \leq 2^{-\omega(\log T)},
\end{align}
where the constant hidden in $\omega(\log T)$  depends only on $\zeta$, $\widetilde{\lambda}$, and $\dim \calH$, $\widehat{\rho}^{\mathbf{E}} \eqdef \frac{1}{M'}\sum_{i=1}^{M'} \rho^{\mathbf{E}}_{F(i)} $ and $\rho^E \eqdef \sum_x P_X(x)\rho_x^E$.
\end{lemma}
\begin{proof}
The proof is akin to the resolvability part of the proof of Lemma~\ref{lm:auxiliary_prob} specialized to the channel  from Alice to Eve. By Lemma~\ref{lm:quantum-resolve}, we have
\begin{align}
\E[F]{\|\widehat{\rho}^{\mathbf{E}} -  \pr{{\rho}^{E}}^{\proddist T} \|_1}&\leq \sqrt{2^{\gamma s + T \phi(s, \alpha_T)}} + \sqrt{\frac{2^\gamma \nu}{M'}},
\end{align}
where $\nu$ is the number of distinct eigenvalues of $\pr{\rho^E}^{\proddist T}$, and for $Q_X\sim$ Bernoulli($p$), we define  \begin{align}
\phi(s, p) \eqdef \log \pr{\sum_x P_X(x) \tr{\pr{\rho_x^E}^{1-s} \pr{\rho^E}^s }}.
\end{align}
For $\gamma = \alpha_T \D{\rho_1^E}{\rho_0^E} T + \frac{\zeta}{2}\alpha_TT$, we have
\begin{multline}
\sqrt{2^{\gamma s + T \phi(s, \alpha_T)}} + \sqrt{\frac{2^\gamma \nu}{M'}}\\
\begin{split} 
&\leq \sqrt{2^{s\alpha_T T\pr{\D{\rho_1^E}{\rho_0^E} + \frac{\zeta}{2} + \frac{\phi(s, \alpha_T)}{s\alpha_T}}}} + \sqrt{2^{-\frac{\zeta}{2} \alpha_T T} \nu}\\
&\stackrel{(a)}{\leq} \sqrt{2^{s\alpha_T T\pr{\D{\rho_1^E}{\rho_0^E} + \frac{\zeta}{2} + \frac{\phi(s, \alpha_T)}{s\alpha_T}}}}\\
&\phantom{=} + \sqrt{2^{-\frac{\zeta}{2} \alpha_T T} (T+1)^{\dim \calH^E}}\\
&\leq \sqrt{2^{s\alpha_T T\pr{\D{\rho_1^E}{\rho_0^E} + \frac{\zeta}{2} + \frac{\phi(s, \alpha_T)}{s\alpha_T}}}} + \frac{1}{2}2^{-\xi \alpha_T  T},\label{eq:res-bound2}
\end{split}
\end{multline}
where $(a)$ follows from \cite[Lemma 3.7]{hayashi2006quantum} and $\xi$ is small positive number. The following lemma is the counterpart of Lemma~\ref{lm:exponent-res} for the channel from Alice to Eve.
\begin{lemma}
\label{lm:exponent-res2}
Fix $\widetilde{s} < 0$, $\widetilde{p}\in[0, 1]$, and $\widetilde{\lambda} \in [0, 1]$. There exists a universal constant $B>0$ such that for all cq-channels $x\mapsto \rho_x^E$, $p\in[0, \widetilde{p}]$, and $s\in [\widetilde{s}, 0]$, we have
\begin{align}
\phi(s, p) > -I(p)s - B(p s^2 - s^3),
\end{align}
where $I(p)\eqdef I(P_X, \rho_x^E)$.
\end{lemma}
\begin{proof}
See Appendix~\ref{sec:error-exponent}.
\end{proof}
Applying Lemma~\ref{lm:exponent-res2} to \eqref{eq:res-bound2}, we obtain
\begin{multline}
\sqrt{2^{s\alpha_T T\pr{\D{\rho_1^E}{\rho_0^E} + \frac{\zeta}{2} + \frac{\phi(s, \alpha_T)}{s\alpha_T}}}} \\
\begin{split}
&\leq \sqrt{2^{s\alpha_T T\pr{\D{\rho_1^E}{\rho_0^E} + \frac{\zeta}{2} + \frac{-\alpha_T \D{\rho_1^E}{\rho_0^E}s - B(\alpha_T^2+ \alpha_T s^2 - s^3)}{s\alpha_T}}}} \\
&= \sqrt{2^{s\alpha_T T\pr{ \frac{\zeta}{2} + \frac{ B(\alpha_T^2+ \alpha_T s - s^2)}{\alpha_T}}}} 
\end{split}
\end{multline}
By choosing $s = o(\sqrt{\alpha_T}) \cap \omega(\frac{\log T}{T\alpha_T})$\footnote{To find such $s$, it is required that $\sqrt{\alpha_T} =  \omega(\frac{\log T}{T\alpha_T})$ or equivalently $\alpha_T = \omega\pr{\pr{\frac{\log T}{T}}^{\frac{2}{3}}}$}, the above expression goes to zero faster than any polynomial.
\end{proof}
\begin{lemma}
\label{lm:universal-resolvability}
Fix $\zeta$ and $\widetilde{\lambda}$ in $]0, 1[$. Let $x\mapsto \rho_x^E(\theta)$ be a compound cq-channel with $x\in \calX = \{0, 1\}$ and $\theta\in\Theta$ such that for all $\theta\in\Theta$, $\lambda_{\min}(\rho_0^E) \geq \widetilde{\lambda}$.  Let $P_X$ be as in Lemma~\ref{lm:aym-resolvability}. Let $M'$ be an integer satisfying 
\begin{align}
M' \geq \lceil (1+\zeta) \alpha_T\sup_{\theta\in\Theta}\D{\rho_1^E(\theta)}{\rho_0^E(\theta)} T\rceil,
\end{align}
and $F:\intseq{1}{M}\times \intseq{1}{M'} \to \calX^T$ be a random encoder such that all codewords are independently distributed according to $P_X^{\proddist T}$. Then, there exists $T_0$ depending only on $\dim \calH$, $\zeta$, and $\widetilde{\lambda}$ such that for all $T\geq T_0$, we have
\begin{align}
\P[F]{\forall \theta\in\Theta, \frac{1}{M} \sum_{w=1}^M \left\|\widehat{\rho}^{\mathbf{E}}_w-\pr{\rho^E(\theta)}^{\proddist T}\right\|_1 \leq 2T^{-5}} \geq \frac{2}{3}
\end{align}
where $\widehat{\rho}^{\mathbf{E}}_w \eqdef \frac{1}{M'}\sum_{i=1}^{M'} \rho^{\mathbf{E}}_{F(w,i)} $ and $\rho^E(\theta) \eqdef \sum_x P_X(x)\rho_x^E(\theta)$.
\end{lemma}
\begin{proof}
We again consider the compound cq-channel $x\mapsto \rho_x^E(\widetilde{\theta})$ from Lemma~\ref{lm:cont-disc-compound-cq}. By Lemma~\ref{lm:aym-resolvability}, for all $\widetilde{\theta}\in\widetilde{\Theta}$, we have
\begin{multline}
\label{eq:ub-e-pe}
\E[F]{\frac{1}{M} \sum_{w=1}^M \left\|\widehat{\rho}^{\mathbf{E}}_w-\pr{\rho^E(\widetilde{\theta})}^{\proddist T}\right\|_1 }\\
\begin{split}
&= \frac{1}{M} \sum_{w=1}^M\E[F]{ \left\|\widehat{\rho}^{\mathbf{E}}_w-\pr{\rho^E(\widetilde{\theta})}^{\proddist T}\right\|_1 }\\
 &\leq 2^{-\omega(\log T)}.
 \end{split}
\end{multline}
By Markov's inequality and the union bound, we have
\begin{multline}
\mathbb{P}_F\left(\forall \widetilde{\theta}\in\widetilde{\Theta}, \frac{1}{M} \sum_{w=1}^M \left\|\widehat{\rho}^{\mathbf{E}}_w-\pr{\rho^E(\widetilde{\theta})}^{\proddist T}\right\|_1 \right.\\
\left.\leq 3|\widetilde{\Theta}|\E[F]{\frac{1}{M} \sum_{w=1}^M \left\|\widehat{\rho}^{\mathbf{E}}_w-\pr{\rho^E(\widetilde{\theta})}^{\proddist T}\right\|_1}\right) \geq \frac{2}{3}.
\end{multline}
Since $|\widetilde{\Theta}|$ is upper-bounded by a polynomial in $T$, we have  
\begin{align}
 3|\widetilde{\Theta}|\E[F]{\frac{1}{M} \sum_{w=1}^M \left\|\widehat{\rho}^{\mathbf{E}}_w-\pr{\rho^E(\widetilde{\theta})}^{\proddist T}\right\|_1} = 2^{-\omega(\log T)}.
\end{align} 
Finally, by Lemma~\ref{lm:cont-disc-compound-cq}, for all $\theta\in\Theta$, there exists $\widetilde{\theta}\in\widetilde{\Theta}$ such that 
\begin{multline}
\frac{1}{M} \sum_{w=1}^M \left\|\widehat{\rho}^{\mathbf{E}}_w-\pr{\rho^E({\theta})}^{\proddist T}\right\|_1\\
 \leq \frac{1}{M} \sum_{w=1}^M \left\|\widehat{\rho}^{\mathbf{E}}_w-\pr{\rho^E(\widetilde{\theta})}^{\proddist T}\right\|_1 + T^{-5}.
\end{multline}
 Thus, for large enough $T$, we have
\begin{align}
\P[F]{\forall \theta\in\Theta, \frac{1}{M} \sum_{w=1}^M \left\|\widehat{\rho}^{\mathbf{E}}_w-\pr{\rho^E({\theta})}^{\proddist T}\right\|_1\leq 2T^{-5}} \geq  \frac{2}{3}.
\end{align}
\end{proof}
\subsubsection{Proof of Theorem~\ref{th:universal-covert}}

We are now ready to provide the proof of the main result of this section. Our code construction is  similar to \cite{Bloch2016a}, which uses wiretap coding to ensure the security of a covert message. Fix $\zeta$, $\widetilde{\lambda}^B$, $\widetilde{\lambda}^E$, $D^B$, and $D^E$, and let $\Theta$ be an arbitrary indexing of all cq-channels $x\mapsto\rho_x^{BE}$ satisfying \eqref{eq:db-cond}-\eqref{eq:lw-cond} for which the corresponding cq-channel to $\theta\in\Theta$ is $x\mapsto\rho_x^{BE}(\theta)$. Considering the sequence $\{\alpha_T\}_{T\geq 1}$ from Definition~\ref{def:q-covert-process}, for a fixed large enough $T$, let $P_X$ be Bernoulli($\alpha_T$); let $F:\intseq{1}{M}\times\intseq{1}{M'}\to\calX^T$ be a random encoder whose codewords are {iid} according to $P_X^{\proddist T}$ that encodes two messages $W$ and $W'$ uniformly distributed over $\intseq{1}{M}$ and $\intseq{1}{M'}$, respectively, to a codeword $\mathbf{X}$. By Lemma~\ref{lm:universal-reliability}, for
\begin{align}
\log M  + \log M'& =\lfloor (1-\zeta)\alpha_T \inf_{\theta\in\Theta}\D{\rho_1^{B}}{\rho_0^B}T \rfloor\\
&\geq \lfloor (1-\zeta)\alpha_T D^BT\rfloor,
\end{align}
we have
\begin{align}
\label{eq:p-rel-uni}
\P[F]{\forall \theta\in\Theta, P_e(\theta) \leq 2T^{-5}} \geq \frac{2}{3},
\end{align}
where $P_e(\theta)$ is the probability that at least one of the messages $W$ and $W'$ is not decoded correctly at the receiver when the cq-channel corresponding to index $\theta$ is used. Moreover, by Lemma~\ref{lm:universal-resolvability}, for
\begin{align}
\log M' &= \lceil (1+\zeta) \alpha_T\sup_{\theta\in\Theta}\D{\rho_1^E(\theta)}{\rho_0^E(\theta)} T\rceil\\
&\leq  \lceil (1+\zeta) \alpha_T D^E T\rceil,
\end{align}
we have
\begin{align}
\label{eq:p-res-uni}
\P[F]{\forall \theta\in\Theta, \frac{1}{M} \sum_{w=1}^M \D{\widehat{\rho}^{\mathbf{E}}_w}{\pr{\rho^E(\theta)}^{\proddist T}} \leq  2T^{-5}} \geq \frac{2}{3},
\end{align}
where $\widehat{\rho}^{\mathbf{E}}_w  $ and $\rho^E(\theta)$ are defined in the statement of Lemma~\ref{lm:universal-resolvability}. Inequalities \eqref{eq:p-rel-uni} and \eqref{eq:p-res-uni} imply that there exists a realization $f$  of $F$  such that for all $\theta\in\Theta$,
\begin{align}
P_e(\theta)& \leq 2T^{-5},\\
\frac{1}{M} \sum_{w=1}^M \left\|\widehat{\rho}^{\mathbf{E}}_w-\pr{\rho^E(\theta)}^{\proddist T}\right\|_1&\leq  2T^{-5}.
\end{align}
Hence, by Lemma~\ref{lm:div-bound-l1}, we upper-bound the quantum relative entropy between the induced quantum states and $\pr{\rho^E(\theta)}^{\proddist T}$ as
\begin{multline}
\frac{1}{M} \sum_{w=1}^M \D{\widehat{\rho}^{\mathbf{E}}_w}{\pr{\rho^E(\theta)}^{\proddist T}}\\
\begin{split}
&\leq 2T^{-5}\log \frac{d^{T}}{\pr{\lambda_{\min}(\rho^E(\theta))}^T 2T^{-5}}\\
&= 2T^{-4}\pr{\log \frac{d}{\lambda_{\min}(\rho^E(\theta))} + \frac{5\log T - \log 2 }{T}}.
\end{split}
\end{multline}
To lower-bound the minimum eigenvalue of $\rho^E(\theta)$, we use Lemma~\ref{lm:lambda-min-max} to obtain for large $T$,
\begin{align}
\lambda_{\min}(\rho^E(\theta)) 
&= \lambda_{\min}(\alpha_T\rho^E_1(\theta) + (1-\alpha_T)\rho^E_0(\theta))\\
&\geq  \lambda_{\min}((1-\alpha_T)\rho^E_0(\theta)) - \|\alpha_T\rho^E_1(\theta)\|_1 \\
&\geq  (1-\alpha_T)\lambda_{\min}(\rho^E_0(\theta)) - \alpha_T\\ 
&\geq \frac{\widetilde{\lambda}^E}{2}.
\end{align}
Therefore, for some constant $L_1 > 0$ depending on $d$ and $\widetilde{\lambda}^E$, we have
\begin{align}
\frac{1}{M} \sum_{w=1}^M \D{\widehat{\rho}^{\mathbf{E}}_w}{\pr{\rho^E(\theta)}^{\proddist T}} &\leq L_1T^{-4}.
\end{align}
To analyze the secrecy of the protocol, since there is no public communication and $W$ is the key extracted at Alice's end, the information leakage to the adversary is
\begin{align}
\D{\rho^{\mathbf{E}W}}{\rho^{\mathbf{E}} \otimes \rho^{W}_{\text{unif}}} 
&\stackrel{(a)}{=} \D{\rho^{\mathbf{E}W}}{\rho^{\mathbf{E}} \otimes \rho^{W}_{\text{unif}}}\\
&\leq \D{\rho^{\mathbf{E}W}}{\pr{\rho^E(\theta)}^{\proddist T}\otimes \rho^{W}_{\text{unif}}}\\
&= \frac{1}{M} \sum_{w=1}^M \D{\widehat{\rho}^{\mathbf{E}}_w}{\pr{\rho^E(\theta)}^{\proddist T}}\\
 &\leq L_1T^{-4},
\end{align}
where $(a)$ follows since there is no public communication. For the covertness, first note that by convexity of quantum relative entropy, we have
\begin{align}
\D{\rho^{\mathbf{E}}}{\pr{\rho^{E}(\theta)}^{\proddist T}} \leq \frac{1}{M} \sum_{w=1}^M \D{\widehat{\rho}^{\mathbf{E}}_w}{\pr{\rho^E(\theta)}^{\proddist T}}\leq L_1T^{-4}.
\end{align}
 We can subsequently bound $\D{\rho^{\mathbf{E}}}{\pr{\rho^{E}_0(\theta)}^{\proddist T}}$ as
\begin{align}
&\D{\rho^{\mathbf{E}}}{\pr{\rho^{E}_0(\theta)}^{\proddist T}}\\\displaybreak[0]
&\nonumber= \D{\rho^{\mathbf{E}}}{\pr{\rho^{E}(\theta)}^{\proddist T}} +\D{\pr{\rho^{E}(\theta)}^{\proddist T}}{\pr{\rho^{E}_0(\theta)}^{\proddist T}} \\ \displaybreak[0]
&\nonumber\phantom{=}+\textnormal{tr}\left(\pr{\rho^{\mathbf{E}}- \pr{\rho^{E}(\theta)}^{\proddist T}}\right.\\
&\left.\phantom{=}\times\pr{\log \pr{\rho^{E}(\theta)}^{\proddist T} - \log\pr{\rho^{E}_0(\theta)}^{\proddist T} }\right)\\ \displaybreak[0]
&\nonumber\leq \D{\rho^{\mathbf{E}}}{\pr{\rho^{E}(\theta)}^{\proddist T}} +\D{\pr{\rho^{E}(\theta)}^{\proddist T}}{\pr{\rho^{E}_0(\theta)}^{\proddist T}} \\ \displaybreak[0]
&\phantom{=}+\left\|\rho^{\mathbf{E}}- \pr{\rho^{E}(\theta)}^{\proddist T}\right\|_1 T\pr{\log \frac{1}{\lambda_{\min}(\rho^E_0(\theta)) \lambda_{\min}(\theta)} }\\ \displaybreak[0]
&\nonumber\stackrel{(a)}{\leq}\D{\rho^{\mathbf{E}}}{\pr{\rho^{E}(\theta)}^{\proddist T}} +\D{\pr{\rho^{E}(\theta)}^{\proddist T}}{\pr{\rho^{E}_0(\theta)}^{\proddist T}} \\ \displaybreak[0]
&\phantom{=}+\sqrt{ \D{\rho^{\mathbf{E}}}{\pr{\rho^{E}(\theta)}^{\proddist T}} } T\\\displaybreak[0]
&\phantom{=}\times\pr{\log \frac{1}{\lambda_{\min}(\rho^E_0(\theta)) \lambda_{\min}(\rho^E(\theta))} }\\ \displaybreak[0]
&\nonumber\leq L_1T^{-4} + \D{\pr{\rho^{E}(\theta)}^{\proddist T}}{\pr{\rho^{E}_0(\theta)}^{\proddist T}}\\\displaybreak[0]
&\phantom{=} + \sqrt{L_1T^{-4}}T\log \frac{1}{\lambda_{\min}(\rho^E_0(\theta)) \lambda_{\min}(\rho^E(\theta))}\\\displaybreak[0] 
&\nonumber= \D{\pr{\rho^{E}(\theta)}^{\proddist T}}{\pr{\rho^{E}_0(\theta)}^{\proddist T}}  + L_1T^{-4}\\\displaybreak[0]
&\phantom{=} + \sqrt{L_1}\log \frac{1}{\lambda_{\min}(\rho^E_0(\theta)) (\rho^E(\theta))}T^{-1},\\\displaybreak[0]
&\nonumber\stackrel{(b)}{\leq}\frac{\alpha_T^2\chi_2(\rho_1^E(\theta))\|\rho_0^E(\theta)}{2}T + L_2\alpha_T^3T + L_1T^{-4} \\\displaybreak[0]
&\phantom{=}+ \sqrt{L_1}\log \frac{1}{\lambda_{\min}(\rho^E_0) \lambda_{\min}(\rho^E(\theta))}T^{-1},\\\displaybreak[0]
&\nonumber\leq \frac{\alpha_T^2\chi_2(\rho_1^E(\theta))\|\rho_0^E(\theta)}{2}T + L_2\alpha_T^3T + L_1T^{-4}\\\displaybreak[0]
&\phantom{=} + 2\sqrt{L_1}\log \frac{2}{\widetilde{\lambda}^E}T^{-1},
\end{align}
where $(a)$ follows from Pinsker inequality, and $(b)$ follows from \cite{Sheikholeslami2016}.

\subsection{Covert Quantum Tomography}
\label{sec:covert-quant-tomogr}

\subsubsection{Instantiation of a covert estimation protocol}
\label{sec:link-covert-estim}

We now detail how Alice and Bob can covertly form estimates of $\D{\rho_1^B}{\rho_0^B}$ and $\D{\rho_1^W}{\rho_0^W}$. By our discussion at the beginning of Section~\ref{sec:unkown-channels}, if the channel from Alice to Bob is $\calE_{{\widetilde{A}}\to B}$, the goal of the estimation phase would be to  first verify the conditions \eqref{eq:collective-chi-min} and \eqref{eq:collective-lambdab}, and if they hold, to estimate $D^B(\calE)\eqdef\D{\calE_{{\widetilde{A}}\to B}(\rho_1^{\widetilde{A}})}{\calE_{{\widetilde{A}}\to B}(\rho_0^{\widetilde{A}})}$ and $D^E(\calE)\eqdef\D{\calE_{{\widetilde{A}}\to B}^\dagger(\rho_1^{\widetilde{A}})}{\calE_{{\widetilde{A}}\to B}^\dagger(\rho_0^{\widetilde{A}})}$. The protocol will be aborted otherwise. We shall use standard quantum tomography \cite{nielsen2002quantum} and adapt it to be covert. We start the description of the estimation phase by formally defining an estimation protocol. Suppose Alice and Bob have access to private randomness $R$ distributed according to $P_R$ over $\calR$ and use $T'$ channel uses for the estimation phase. The estimation protocol consists of an encoder function $f:\calR \to \calD(\calH)^{T'}$ for Alice, a POVM $\mathbf{M}_r = \{M_r^j\}_{j\in \calJ}$ for each $r\in\calR$ applied by Bob to his received state $\rho^{\mathbf{B}}$ when $R=r$ and results in an output $j$ in $\calJ$ the set of all possible outputs of the measurement, one function $H: \calJ\to\{0, 1\}$ used by Bob to verify that \eqref{eq:collective-chi-min} and \eqref{eq:collective-lambdab} hold, and two estimators $\widehat{D}^B:\calJ\to\mathbb{R}$ and $\widehat{D}^W:\calJ\to\mathbb{R}$ used by Bob to form estimations of $\smash{\D{\calE_{{\widetilde{A}}\to B}(\rho_1^{\widetilde{A}})}{\calE_{{\widetilde{A}}\to B}(\rho_0^{\widetilde{A}})}}$ and $\smash{\D{\calE_{{\widetilde{A}}\to B}^\dagger(\rho_1^{\widetilde{A}})}{\calE_{{\widetilde{A}}\to B}^\dagger(\rho_0^{\widetilde{A}})}}$, respectively.

 We now explicitly instantiate a covert estimation protocol. Consider any number of channel uses $T'$ and any quantum channel $\calE: \calL(\calH) \to \calL(\calH)$ where $\calH$ is a $d$-dimensional Hilbert space. Let $\widetilde{E}_1, \cdots, \widetilde{E}_{d^2}$ be defined as discussed in the beginning of Section~\ref{sec:unkown-channels}, i.e., for an orthonormal basis $\ket{1}, \cdots, \ket{d}$, we let $\widetilde{E}_{d(n-1) + m} \eqdef \ket{n}\bra{m}$.  Our goal is to estimate $\calE(\widetilde{E}_n)$ for all $n\in\intseq{1}{d^2}$ from which we would have a complete characterization of the quantum channel $\calE$. To do so, the main idea is that Alice would send some states through Pulse-Position Modulation (PPM) to Bob for which Bob performs quantum state tomography. More concretely, Alice and Bob first agree on two integers $q$ and $\ell$ such that $q\ell \leq T'$ and sample an {iid} sequence  $U_1, \cdots, U_\ell$  from their private randomness where each $U_i$ is uniformly distributed over $\intseq{1}{q}$. Alice then transmits the innocent state $\rho_0^{\widetilde{A}}$ on the $i^{th}$ channel uses unless 
\begin{align}
i \in \calI \eqdef \{U_1, U_2 + q, \cdots, U_\ell + q (\ell-1)\}.
\end{align}
To determine the state that should be sent by Alice on the positions in $\calI$, let us define the vectors 
\begin{align}
\ket{n, m, +} &\eqdef \frac{\ket{n}+\ket{m}}{\sqrt{2}}\\
\ket{n, m, -} &\eqdef \frac{\ket{n}+i\ket{m}}{\sqrt{2}}
\end{align}
and consider pure states 
\begin{multline}
\calS \eqdef \{\ket{n, m, +}\bra{n, m, +}:
n\neq m\}\\ \cup \{\ket{n, m, -}\bra{n, m, -}:  n\neq m\} \cup \{\ket{n}\bra{n}:n\in \intseq{1}{d}\},
\end{multline}
where $\card{\calS} =2d^2-d$.  On the positions in $\calI$,  in an arbitrary but known order, Alice transmits   each  state in $\calS$ $\lfloor \ell / \card{S}\rfloor$ times. Then, for each state $\rho\in \calS$, Bob receives $\lfloor \ell / \card{S}\rfloor$ independent copies of $\calE(\rho)$, and performs a POVM defined by $\{\widetilde{\rho}, I - \widetilde{\rho}\}$  for each operator $\widetilde{\rho} \in \calS$, $\widetilde{\ell} \eqdef \lfloor \lfloor \ell / \card{S}\rfloor / \card{\calS}\rfloor$ times. Let $\widehat{N}(\rho, \widetilde{\rho})$ be the number of times the result of the measurement $\{\widetilde{\rho}, I - \widetilde{\rho}\}$ on $\calE(\rho)$ corresponds to $\widetilde{\rho}$ and let $\widehat{f}(\rho, \widetilde{\rho}) \eqdef \widehat{N}(\rho, \widetilde{\rho})/\widetilde{\ell}$. Bob subsequently estimates  $\calE(\rho)$ for each $\rho\in \calS$ as
\begin{align}
\widehat{\calE}(\rho)
&\eqdef\sum_{n\neq m} \ket{n}\bra{m} \left(\widehat{f}(\rho, \ket{n, m, +}\bra{n, m, +}) \right. \nonumber\\
&\left.\phantom{===}- i\widehat{f}(\rho, \ket{n, m, -}\bra{n, m, -}) - \frac{1-i}{2}\widehat{f}(\rho, \ket{n}\bra{n})\right.\\
&\left.\phantom{===}- \frac{1-i}{2}\widehat{f}(\rho, \ket{m}\bra{m})\right)+\sum_n \ket{n}\bra{n} \widehat{f}(\rho, \ket{n}\bra{n}).\label{eq:calEhat-rho}
\end{align}
Since $\{\widetilde{E}_j: j \in \intseq{1}{d^2}\}$ is an orthonormal basis for $\calL(\calH)$, we can write $\widehat{\calE}(\rho) \eqdef \sum_{j}\widetilde{E}_j \widehat{\lambda}_{\rho, j}$ for some unique $\widehat{\lambda}_{\rho, j}$.
Then, for $n, m\in\intseq{1}{d}$, we define
\begin{align}
\widehat{\calE}(\widetilde{E}_{d(n-1)+m}) \label{eq:lambda-hat-def}
&\eqdef \begin{cases}\widehat{\calE}(\ket{n, m, +}\bra{n, m, +})  \quad &n\neq m,\\
+ i \widehat{\calE}(\ket{n, m, -}\bra{n, m, -})\\
 - \frac{1+i}{2}\widehat{\calE}(\ket{n}\bra{n}) \\
 - \frac{1+i}{2} \widehat{\calE}(\ket{m}\bra{m})\\
 \\
\widehat{\calE}(\ket{n}\bra{n})\quad&n=m,
\end{cases}
\end{align}
which is enough to characterize a quantum channel. We can similarly write $ {\calE}(\widetilde{E}_{d(n-1)+m}) = \sum_{j}\widetilde{E}_j{\lambda}_{d(n-1)+m,j}$ for some unique ${\lambda}_{d(n-1)+m,j}$. We next attempt to form an estimation of the $\chi$-representation of the channel $\calE$, $\{\chi_{j, k}\}$, defined at the beginning of Section~\ref{sec:unkown-channels}. By \cite{nielsen2002quantum}, for some fixed $\kappa_{j,k}^{j',k'}$,
\begin{align}
\chi_{j ,k} = \sum_{j', k'} \kappa_{j,k}^{j',k'} \lambda_{j',k'}.
\end{align}
We thus define $\widehat{\chi}_{j,k} \eqdef \sum_{j', k'} \kappa_{j,k}^{j',k'} \widehat{\lambda}_{j',k'}$. Finally, for some $\tau>0$, we define
\begin{align}
H &\eqdef \indic{\lambda_{\min}(\widehat{\chi}) \geq \widetilde{\lambda}^{\chi} -\tau \text{ and } \lambda_{\min}(\widehat{\calE}(\rho_0^{\widetilde{A}})) \geq \widetilde{\lambda}^{B} - \tau },\\
\widehat{D}^B &\eqdef \D{\widehat{\calE}(\rho_1^{\widetilde{A}})}{\widehat{\calE}(\rho_0^{\widetilde{A}})} - \tau,\\
\widehat{D}^E &\eqdef \D{\widehat{\calE}^\dagger(\rho_1^{\widetilde{A}})}{\widehat{\calE}^\dagger(\rho_0^{\widetilde{A}})}+\tau.
\end{align}

The next theorem establishes bounds on the performance of the described covert estimation protocol.
\begin{theorem}
\label{th:estimation-analysis}
There exist some $\xi > 0$ that depends on $\tau$, $d$, $\widetilde{\lambda}^\chi$, $\widetilde{\lambda}^B$, and $\widetilde{\lambda}^E$ such that
\begin{multline}
\D{\rho^{\mathbf{E}W} }{(\rho_0^{\mathbf{E}}) \otimes \rho_{\text{unif}}^W} \leq \frac{\ell}{q}\pr{\frac{\dim \calH^E}{\widetilde{\lambda}^E} - 1}, 
\label{eq:d-est-ph}
\end{multline}
\begin{multline}
\P{H=0|\lambda_{\min}(\chi) \geq \widetilde{\lambda}^{\chi}, \lambda_{\min}(\calE(\rho^{\widetilde{A}}_0)) \geq \widetilde{\lambda}^B} \\\leq 2^{-\xi n},\label{eq:H0-prob}
\end{multline}
\begin{multline}
\mathbb{P}\left(H=1|\lambda_{\min}(\chi) \leq \widetilde{\lambda}^{\chi} - 2\tau\right.\\
 \left.\textnormal{ or } \lambda_{\min}(\calE(\rho^{\widetilde{A}}_0)) \leq \widetilde{\lambda}^B - 2\tau\right) \leq 2^{-\xi n}\label{eq:H1-prob},
\end{multline}
\begin{multline}
\label{eq:p-par-est}
\mathbb{P}\left(D^B(\calE)- 2\tau \leq \widehat{D}^B \leq D^B(\calE),\right.\\
\left.D^E(\calE) \leq \widehat{D}^E \leq D^E(\calE) + 2\tau\right.\\
\left.|\lambda_{\min}(\chi) \geq \widetilde{\lambda}^{\chi} -2\tau,\right.
 \left.\lambda_{\min}(\calE(\rho^{\widetilde{A}}_0)) \geq \widetilde{\lambda}^B-2\tau\right)\\\geq 1-  2^{-\xi \ell}.
\end{multline}
\end{theorem}
\begin{figure}
  \centering
  \includegraphics[width=\linewidth]{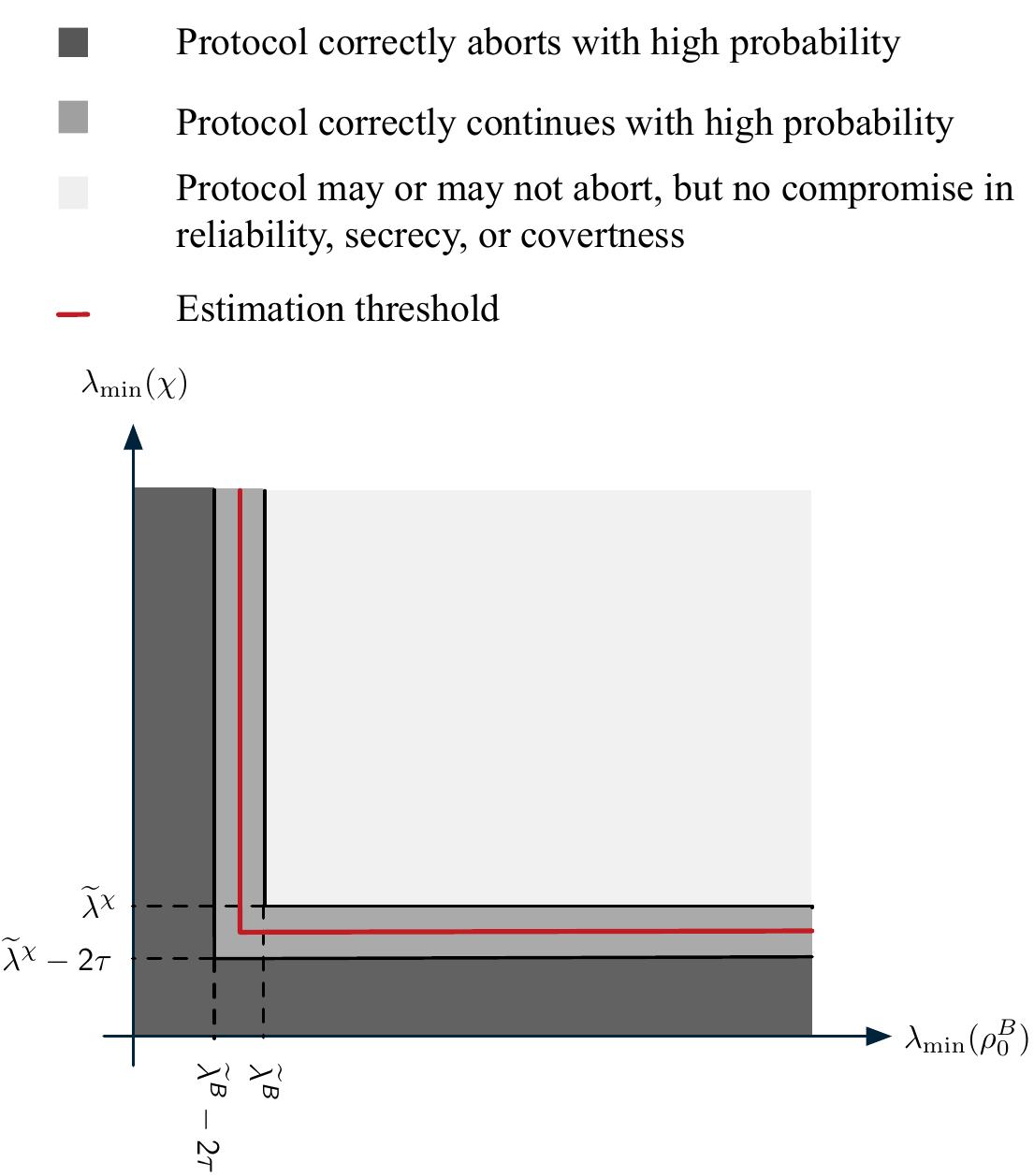}
  \caption{Testing the conditions \eqref{eq:collective-chi-min} and \eqref{eq:collective-lambdab}}
  \label{fig:estimation}
\end{figure}
We shall prove Theorem~\ref{th:estimation-analysis} in Section~\ref{sec:proof-estimation}. Note that~\eqref{eq:d-est-ph} characterizes the covertness of the estimation protocol by bounding the relative entropy between the state induced by the estimation protocol and the state in which there is no communication. \eqref{eq:H0-prob} and~\eqref{eq:H1-prob} characterize the robustness of estimation since \eqref{eq:H0-prob} bounds the probability that the channel satisfies the required condition \eqref{eq:collective-chi-min} and \eqref{eq:collective-lambdab} but Alice and Bob abort the protocol while \eqref{eq:H1-prob} bounds the probability that Alice and Bob run the key generation phase but the channel does not satisfy the required conditions. Finally,~\eqref{eq:p-par-est} characterizes the accuracy of the estimation by bounding the probability that the estimated parameters of the channel are close to their true values.

As depicted in Fig.~\ref{fig:estimation}, there is a technical subtlety in verifying \eqref{eq:collective-chi-min} and \eqref{eq:collective-lambdab} because the channel estimation error in finite number of channel uses prevents us from testing with absolute certainty that \eqref{eq:collective-chi-min} and \eqref{eq:collective-lambdab} hold. In other words, there could exist a set of channels for which, based on the estimation error, Alice and Bob may or may not abort the protocol; regardless, the protocol ensures that if the key generation phase is executed, it is reliable, secure, and covert. 

We conclude this section by analyzing the performance of a covert key generation protocol obtained by combining the covert estimation protocol with the universal code introduced in Section~\ref{sec:universal-covert-code}. More precisely, Alice and Bob first perform the described estimation protocol $\calP$ over $T'$ channel uses.  Using $O(\log T')$ channel uses and $O(\log T')$ bits of private common randomness, Bob transmits the one-time-padded $H$, $\widehat{D}^B$, and $\widehat{D}^{W}$ over the public channel. If $H=0$, Alice abort the protocol. If $H=1$, after obtaining $\widehat{D}^B$ and $\widehat{D}^E$, Alice and Bob run the universal code $\calC_T$ introduced in Theorem~\ref{th:universal-covert} for $D^B = \widehat{D}^B$, $D^E = \widehat{D}^E$, and the lower-bounds on the minimum eigenvalue of $\rho_0^B$ and $\rho_0^E$, $\widetilde{\lambda}^B-2\tau$ and $\widetilde{\lambda}^E$, respectively. The rationale behind the conservative choice for the minimum eigenvalue of $\rho_0^B$ is that, because to the estimation error, Alice and Bob might accept the channels for which $\lambda_{\min}(\rho_0^B)$ is slightly less than $\widetilde{\lambda}_B$. We characterize the reliability, secrecy, and covertness of the overall protocol in the next lemma and provide the poof in Section~\ref{sec:proof-estimation}.
\begin{lemma}
\label{lm:cascade-analysis}
For all channels $\calE_{{\widetilde{A}}\to B}$ if we only know $\lambda_{\min}(\calE_{{\widetilde{A}}\to E}(\rho^{\widetilde{A}}_0)) \geq \widetilde{\lambda}^E$, we have
 \begin{align}
 P_e &\leq \P{H= 1},\displaybreak[0]\\
 S &\leq   \P{H= 1}\pr{T \log \frac{1}{\widetilde{\lambda}^E} + \ell^{\max}},\displaybreak[0]\\
 C &\leq \P{H=1} T \log \frac{1}{\widetilde{\lambda}^E} + \delta,
 \end{align}
 where $L_1, L_2>0$ depend on $\dim \calH^E$ and $\widetilde{\lambda}^E$, and $\ell^{\max}$ is the maximum length of the key.
 
In addition, for a quantum channel $\calE_{{\widetilde{A}}\to B}$ with $\lambda_{\min}(\calE_{{\widetilde{A}}\to E}(\rho^{\widetilde{A}}_0)) \geq \widetilde{\lambda}^E$ and $\lambda_{\min}(\calE_{{\widetilde{A}}\to B}(\rho^{\widetilde{A}}_0)) \geq \widetilde{\lambda}^B-2\tau$, and an estimation protocol $\calP$, define $\epsilon\eqdef \P{H=1 \textnormal{ and } (D^B(\calE)\leq \widehat{D}^B \textnormal{ or } D^E(\calE) \geq \widehat{D}^E )}$ and $\delta \eqdef \D{\rho^{\mathbf{E}W} }{(\rho_0^{\mathbf{E}}) \otimes \rho_{\text{unif}}^W}$. For the protocol described above, we have
 \begin{align}
 P_e& \leq 2T^{-5} + \epsilon,\\
 S &\leq L_1T^{-4} + \epsilon \pr{T \log \frac{1}{\widetilde{\lambda}^E} + \ell^{\max}},\\
 C &\leq  \frac{\alpha_T^2\chi_2(\rho_1^E(\theta))\|\rho_0^E(\theta)}{2}T + L_2\alpha_T^3T + L_1T^{-4}\nonumber \\
 &\phantom{=}+ 2\sqrt{L_1}\log \frac{2}{\widetilde{\lambda}^E}T^{-1} + \epsilon T \log \frac{1}{\widetilde{\lambda}^E} + \delta.
 \end{align}

\end{lemma}

\subsubsection{Proof of Theorem~\ref{th:estimation-analysis}  and Lemma~\ref{lm:cascade-analysis}}
\label{sec:proof-estimation}


To show that the desired parameters of the channel are approximated properly by their associated estimators, we first show that the estimated channel $\widehat{\calE}$ defined in \eqref{eq:lambda-hat-def} is close to the true channel $\calE$, i.e.,  for all $j$ and $k$, with high probability, $\widehat{\lambda}_{j, k}$ is close to $\lambda_{j, k}\eqdef \tr{\widetilde{E}_k^\dagger \calE(\widetilde{E}_j)}$.

\begin{lemma}
\label{lm:lambda-bound}
For all $\gamma > 0$ and $\kappa_{\max} \eqdef \max_{j,k,j',k'} |\kappa_{j,k}^{j',k'}|$, we have
\begin{align}
\label{eq:lambad-bound}
\P{\exists j, k: |\lambda_{j, k} - \widehat{\lambda}_{j, k}| \geq \gamma} \leq  16 d^4 e^{-\frac{1}{256} \widetilde{\ell}\gamma^2},
\end{align}
and
\begin{align}
\label{eq:chi-bound}
\P{\exists j, k: |\chi_{j, k} - \widehat{\chi}_{j, k}| \geq  d^2 \kappa_{\max}\gamma} \leq  16 d^4 e^{-\frac{1}{256} \widetilde{\ell}\gamma^2}.
\end{align}
\end{lemma}
\begin{proof}
We only prove \eqref{eq:lambad-bound} as \eqref{eq:chi-bound} then follows from the definition of $\widehat{\chi}_{j, k}$. Notice first that, by our construction,  the distribution of $\widehat{N}(\rho,\widetilde{\rho}) = \widetilde{\ell} \widehat{f}(\rho,\widetilde{\rho})$ is Binomial$(\tr{ \widetilde{\rho}\calE(\rho)}, \widetilde{\ell})$ for all $\rho, \widetilde{\rho}\in\calS$. Therefore,  Hoeffding's inequality yields for all $\gamma > 0 $ that
\begin{align}
\label{eq:general-hoeff}
\P{|\widehat{f}(\rho, \widetilde{\rho}) - \tr{\widetilde{\rho}\calE(\rho)}| \geq \gamma} \leq 2 \exp{-2\widetilde{\ell} \gamma^2}.
\end{align}
For all $n$, $m$, $n'$, and $m'$, using the equality 
\begin{multline}
\label{eq:e-expansion}
\widetilde{E}_{d(n-1) + m} = \ket{n, m,+}\bra{n, m,+} \\+ i \ket{n, m,-}\bra{n, m, -}- \frac{1+i}{2}\ket{n}\bra{n} - \frac{1+i}{2}\ket{m}\bra{m},
\end{multline}
we expand $\lambda_{j, k}$ and $\widehat{\lambda}_{j, k}$ in terms of $\tr{\widetilde{\rho}\calE(\rho)}$ and $\widehat{f}(\rho, \widetilde{\rho})$, respectively, and apply \eqref{eq:general-hoeff}. More precisely, by definition of $\lambda_{d(n-1)+m, d(n'-1)+m'} $, we have
\begin{multline}
\lambda_{d(n-1)+m, d(n'-1)+m'} \\
\begin{split}
&= \tr{\widetilde{E}_{ d(n'-1)+m'}^\dagger \calE(\widetilde{E}_{d(n-1)+m})}\\
&= \tr{\widetilde{E}_{ d(n'-1)+m'}^\dagger \calE(\ket{n, m,+}\bra{n, m,+})} \\
&\phantom{=}+ i \tr{\widetilde{E}_{ d(n'-1)+m'}^\dagger \calE(\ket{n, m,-}\bra{n, m,-})}\\
&\phantom{=}-\frac{1+i}{2}\tr{\widetilde{E}_{ d(n'-1)+m'}^\dagger \calE(\ket{n}\bra{n})}\\
&\phantom{=} -\frac{1+i}{2} \tr{\widetilde{E}_{ d(n'-1)+m'}^\dagger \calE(\ket{m}\bra{m})}.
\end{split}
\end{multline}
We now fix $n', m'\in\intseq{1}{d}$ and $\rho\in\calS$ and for simplicity, let $j\eqdef d(n'-1)+m'$, $\ket{+} \eqdef \ket{n', m',+}$, $\ket{-}\eqdef \ket{n', m', -}$. Then, by \eqref{eq:e-expansion},	
\begin{multline}
\tr{\widetilde{E}_j^\dagger \calE(\rho)} = \tr{\ket{+}\bra{+}\calE(\rho)}  - i \tr{\ket{-}\bra{-}\calE(\rho)}\\
  -\frac{1-i}{2}\tr{\ket{n}\bra{n}\calE(\rho)}-\frac{1-i}{2}\tr{\ket{m}\bra{m}\calE(\rho)}.
\end{multline}
Therefore, we obtain the upper-bound in \eqref{eq:p-lambda-fig}.
\begin{widetext}
\begin{multline}
\P{|\widehat{\lambda}_{\rho, j} -\tr{\widetilde{E}_j \calE(\rho)} | \geq \gamma}\\
\begin{split} 
&=\mathbb{P}\pr{|\widehat{f}(\rho, \ket{+}\bra{+}) - i\widehat{f}(\rho, \ket{-}\bra{-}) - \frac{1-i}{2}\widehat{f}(\rho, \ket{n'}\bra{n'})- \frac{1-i}{2}\widehat{f}(\rho, \ket{m'}\bra{m'}) - \tr{\widetilde{E}_j \calE(\rho)}| \geq \gamma}\\
 &\leq \P{|\widehat{f}(\rho, \ket{+}\bra{+})  - \tr{\ket{+}\bra{+} \calE(\rho)}| \geq \frac{\gamma}{4}} + \P{|\widehat{f}(\rho, \ket{-}\bra{-})  - \tr{\ket{-}\bra{-} \calE(\rho)}| \geq \frac{\gamma}{4}} +\\
 &~~~~~\P{|\widehat{f}(\rho, \ket{n'}\bra{n'})  - \tr{\ket{n'}\bra{n'} \calE(\rho)}| \geq \frac{\gamma}{2\sqrt{2}}} + \P{|\widehat{f}(\rho, \ket{m'}\bra{m'})  - \tr{\ket{m'}\bra{m'} \calE(\rho)}| \geq \frac{\gamma}{2\sqrt{2}}}\\
 &\leq  4e^{-\frac{1}{8}\widetilde{\ell}\gamma^2}.\label{eq:p-lambda-fig}
 \end{split}
\end{multline}
\end{widetext}
Similarly, to analyze the second term in the right hand side of \eqref{eq:calEhat-rho}, we have
\begin{align}
\P{|\widehat{f}(\rho, \ket{n'}\bra{n'}) - \tr{\ket{n'}\bra{n'}\calE(\rho)}| \geq \gamma} \leq e^{-2\widetilde{\ell} \gamma^2}.
\end{align}
Thus, the union bound implies that
\begin{multline}
\P{\exists j: |\widehat{\lambda}_{\rho, j} - \tr{\widetilde{E}_j\calE(\rho)}| \geq \gamma } \\
\leq d(d-1)4e^{-\frac{1}{8} \widetilde{\ell} \gamma^2} + de^{-2\widetilde{\ell} \gamma^2} \leq 4d^2e^{-\frac{1}{8} \widetilde{\ell} \gamma^2}.
\end{multline}
Moreover, because we have 
\begin{multline}
\widehat{\lambda}_{j, k} = \widehat{\lambda}_{\ket{+}\bra{+}, k} \\+ i\widehat{\lambda}_{\ket{-}\bra{-}, k} - \frac{1+i}{2}\widehat{\lambda}_{\ket{n'}\bra{n'}, k} -\frac{1+i}{2} \widehat{\lambda}_{\ket{m'}\bra{m'}, k},
\end{multline}
we obtain
\begin{align}
\P{\exists j, k: |\lambda_{j, k} - \widehat{\lambda}_{j, k}| \geq \gamma} \leq  16 d^4 e^{-\frac{1}{256} \widetilde{\ell}\gamma^2}.
\end{align}
\end{proof}

\begin{lemma}
\label{lm:approx-states}
For any $\rho \in \calD(\calH)$ and $0<\gamma < \frac{\lambda_{\min}(\chi)}{d^5 \kappa_{\max}}$, we have
\begin{multline}
\P{\|\calE(\rho) -\widehat{\calE}(\rho) \|_1\geq d^6\kappa_{\max} \gamma } \\
\leq 16 d^4 e^{-\frac{1}{256} \widetilde{\ell}\gamma^2},
\end{multline}
\begin{multline}
\P{\|\calE^\dagger(\rho) -\widehat{\calE}^\dagger(\rho) \|_1 \geq  \frac{d^{18} \kappa_{\max}^2 \lambda_{\max}(\widetilde{\rho})\sqrt{\lambda_{\max}(\chi)} \gamma^2 }{2\pr{\lambda_{\min}(\chi) - d^{5}\kappa_{\max} \gamma}} } \\
\leq 16 d^4 e^{-\frac{1}{256} \widetilde{\ell}\gamma^2},
\end{multline}
where $\widetilde{\rho}\eqdef \sum_{j,k} \tr{\widetilde{E}_j \rho \widetilde{E}_k^\dagger} \ket{j}\bra{k}$.
\end{lemma}
\begin{proof}
Using the triangle inequality, we obtain
\begin{align}
\|\calE(\rho) -\widehat{\calE}(\rho) \|_1 
&= \left\| \sum_{j,k} \widetilde{E}_j\rho\widetilde{E}_k^\dagger \chi_{j,k} -\sum_{j,k} \widetilde{E}_j\rho\widetilde{E}_k^\dagger \widehat{\chi}_{j,k} \right\|_1 \\
&\leq \sum_{j,k} \left\| \widetilde{E}_j\rho\widetilde{E}_k^\dagger \right\|_1|\chi_{j,k} -  \widehat{\chi}_{j,k}|  \\
&\leq \sum_{j,k} |\chi_{j,k} -  \widehat{\chi}_{j,k}|.\label{eq:calE-chi}
\end{align}
Furthermore,
\begin{multline}
\|\calE^\dagger(\rho) -\widehat{\calE}^\dagger(\rho) \|_1\\
\begin{split}
&= \|\sqrt{\chi}^* \widetilde{\rho} \sqrt{\chi}^* -\sqrt{\widehat{\chi}}^* \widetilde{\rho} \sqrt{\widehat{\chi}}^*  \|_1\\
&= \|\sqrt{\chi}^* \widetilde{\rho} (\sqrt{\chi}^*-\sqrt{\widehat{\chi}}) -(\sqrt{\widehat{\chi}}^*-\sqrt{\chi}^*) \widetilde{\rho} \sqrt{\widehat{\chi}}^*  \|_1\\
&\leq  \|\sqrt{\chi}^* \widetilde{\rho} (\sqrt{\chi}^*-\sqrt{\widehat{\chi}})\|_1+ \| (\sqrt{\widehat{\chi}}^*-\sqrt{\chi}^*) \widetilde{\rho} \sqrt{\widehat{\chi}}^*  \|_1\\
&\stackrel{(a)}{\leq} \sigma_{\max}(\sqrt{\chi}^* \widetilde{\rho}) \|\sqrt{\chi}^*-\sqrt{\widehat{\chi}}\|_1 \\
&\phantom{=}+  \sigma_{\max}(\widetilde{\rho} \sqrt{\widehat{\chi}}^* )\|\sqrt{\widehat{\chi}}^*-\sqrt{\chi}^*\|_1\\
&\leq  \sigma_{\max}(\widetilde{\rho})\pr{\lambda_{\max}(\sqrt{\chi}) +\lambda_{\max}(\sqrt{\widehat{\chi}})}\|\sqrt{\chi} - \sqrt{\widehat{\chi}}\|_1\\
&\stackrel{(b)}{\leq}  \sigma_{\max}(\widetilde{\rho})\pr{2\sigma_{\max}(\sqrt{\chi})) + \|\sqrt{\chi} - \sqrt{\widehat{\chi}} \|_1}\\
&\phantom{=}\times \|\sqrt{\chi} - \sqrt{\widehat{\chi}}\|_1,
\end{split}
\end{multline}
where $(a)$ follows from Lemma~\ref{lm:product-l1} in Appendix~\ref{sec:tech-lemmas}, and $(b)$ follows from Lemma~\ref{lm:lambda-min-max} in Appendix~\ref{sec:tech-lemmas}. To upper-bound $\|\sqrt{\chi} - \sqrt{\widehat{\chi}}\|_1$, let us define $F(x) \eqdef \sqrt{\chi + x(\widehat{\chi} - \chi)}$; then, we have
\begin{align}
\|\sqrt{\chi} - \sqrt{\widehat{\chi}}\|_1
&= \|F(0) - F(1)\|_1\\
&\leq \sup_{x\in[0, 1]} \|F'(x)\|_1.
\end{align}
Applying Lemma~\ref{lm:norm-derivative} for $f(\mu) = \sqrt{\mu}$ and $A(x) = \chi + x(\widehat{\chi} - \chi)$, we obtain
\begin{align}
\|F'(x)\|_1 
&= \left\|\frac{d}{dx}f(A(x)) \right\|_1\\
&\leq \sup_{\mu\in[\lambda_{\min}(A(x)), \lambda_{\max}(A(x))]} d^4 |f'(\mu)| \|\chi - \widehat{\chi}\|_1\\
&=  \sup_{\mu\in[\lambda_{\min}(A(x)), \lambda_{\max}(A(x))]} d^4\left|\frac{1}{2\sqrt{\mu}}\right| \|\chi - \widehat{\chi}\|_1\\
&= \frac{d^4\|\chi - \widehat{\chi}\|_1}{2\sqrt{\lambda_{\min}(A(x))}}.
\end{align}
Moreover, by Lemma~\ref{lm:lambda-min-max}, we know that 
\begin{align}
\lambda_{\min}(A(x)) 
&= \lambda_{\min}(\chi +x (\widehat{\chi} - \chi))\\
&\geq \lambda_{\min}(\chi) - \|\widehat{\chi} - \chi\|_1.
\end{align}
Hence, we have
\begin{align}
\|\sqrt{\chi} - \sqrt{\widehat{\chi}}\|_1 \leq \frac{d^4\|\chi - \widehat{\chi}\|_1}{2\sqrt{\lambda_{\min}(\chi) - \|\widehat{\chi} - \chi\|_1}}
\end{align}
If for all $j, k$, we have $|\chi_{jk}- \widehat{\chi}_{jk}| \leq d^2\kappa_{\max}\gamma$, then $\|\chi - \widehat{\chi}\|_1 \leq d^5\kappa_{\max} \gamma$. Thus, \eqref{eq:chi-bound} yields the upper-bound in \eqref{eq:app-bound-cale-dag}.
\begin{widetext}

\begin{align}
\P{\|\calE^\dagger(\rho) -\widehat{\calE}^\dagger(\rho) \|_1 \geq  \frac{d^{9} \kappa_{\max} \gamma \lambda_{\max}(\widetilde{\rho}) }{2\sqrt{\lambda_{\min}(\chi) - d^{5}\kappa_{\max} \gamma}} \pr{2\sqrt{\lambda_{\max}(\chi)}\frac{d^{9}\kappa_{\max} \gamma }{2\sqrt{\lambda_{\min}(\chi) - d^5\kappa_{\max} \gamma}}}} &\leq 16 d^4 e^{-\frac{1}{256} \widetilde{\ell}\gamma^2}.\label{eq:app-bound-cale-dag}
\end{align}
\end{widetext}
\end{proof}

\begin{proof}[Proof of Theorem~\ref{th:estimation-analysis}]

\textbf{Covertness analysis}: Let $\rho^{\mathbf{E}}_i$ denote  Eve's state during the channel uses from $q(i-1) + 1$ to $qi$. Since, $U_1, \cdots, U_\ell$ are independent, then
\begin{align}
\D{\rho^{\mathbf{E}}}{\pr{\rho_0^E}^{\proddist T'}} = \sum_{i=1}^{\ell} \D{\rho^{\mathbf{E}}_i}{\pr{\rho_0^{E}}^{\proddist q}}.
\end{align}
We now focus on the block from channel use $q(i-1) + 1$ to $qi$. Define $\overline{\rho}$ as the state sent by Alice on the position $q(i-1) + U_i$ and $\rho(j) \eqdef \pr{\rho_0^E}^{\proddist (j-1)} \otimes \overline{\rho} \otimes  \pr{\rho_0^E}^{\proddist (q - j)}$. One can check that $\rho^{\mathbf{E}}_i = \frac{1}{q} \sum_{j=1}^q \rho(j)$. Thus, we have
\begin{align}
\D{\rho^{\mathbf{E}}_i}{\pr{\rho_0^{E}}^{\proddist q}}
&\stackrel{(a)}{\leq} \tr{\pr{\rho^{\mathbf{E}}_i}^2 \pr{\pr{\rho_0^{E}}^{\proddist q}}^{-1} } - 1\\
&=\tr{\pr{\frac{1}{q} \sum_{j=1}^q \rho(j)}^2 \pr{\pr{\rho_0^{E}}^{\proddist q}}^{-1} } - 1\\
&= \frac{1}{q^2} \sum_{j=1}^q \sum_{\widetilde{j}=1}^q \tr{\rho(j) \rho(\widetilde{j})  \pr{\pr{\rho_0^{E}}^{\proddist q}}^{-1}},
\end{align}
where $(a)$ follows from \cite{ruskai1990convexity}. Note that for $j <\widetilde{j}$, we have \eqref{eq:tr-bound}.
\begin{widetext}

\begin{align}
\tr{\rho(j) \rho(\widetilde{j})  \pr{\pr{\rho_0^{E}}^{\proddist q}}^{-1}} 
&=  \tr{\pr{\pr{\rho_0^E}^{\proddist (j-1)} \otimes \overline{\rho} \otimes  \pr{\rho_0^E}^{\proddist (q - j)}} \pr{\pr{\rho_0^E}^{\proddist (\widetilde{j}-1)} \otimes \overline{\rho} \otimes  \pr{\rho_0^E}^{\proddist (q - \widetilde{j})}}\pr{\pr{\rho_0^{E}}^{\proddist q}}^{-1} }\\
&= \tr{\pr{\rho_0^E}^{\proddist (j-1)} \otimes  \pr{\overline{\rho} \rho_0^E  \pr{\rho_0^E}^{-1}} \otimes  \pr{\rho_0^E}^{\proddist (\widetilde{j} - j -1)} \otimes \pr{\rho_0^E \overline{\rho} \pr{\rho_0^E}^{-1} } \otimes  \pr{\rho_0^E}^{\proddist (q - \widetilde{j})}}\\
&= \tr{\overline{\rho}} \tr{\rho_0^E \overline{\rho} \pr{\rho_0^E}^{-1} }\pr{\tr{\rho_0^E}}^{q-2}\\
&= 1.\label{eq:tr-bound}
\end{align}
\end{widetext}
Similarly, one can show that for $j > \widetilde{j}$, we have $\tr{\rho(j) \rho(\widetilde{j})  \pr{\pr{\rho_0^{E}}^{\proddist q}}^{-1}}  = 1$. Furthermore, when $j = \widetilde{j}$, we have
\begin{multline}
\tr{\rho(j) \rho(\widetilde{j})  \pr{\pr{\rho_0^{E}}^{\proddist q}}^{-1}} \\
\begin{split}
&= \tr{\pr{\pr{\rho_0^E}^{\proddist (j-1)} \otimes \overline{\rho} \otimes  \pr{\rho_0^E}^{\proddist (q - j)}} ^2\pr{\pr{\rho_0^{E}}^{\proddist q}}^{-1} }\\
&= \tr{\pr{\rho_0^E}^{\proddist (j-1)} \otimes \pr{\overline{\rho}^2 \pr{\rho_0^E}^{-1}} \otimes  \pr{\rho_0^E}^{\proddist (q - j)} }\\
&= \tr{\pr{\overline{\rho}^2 \pr{\rho_0^E}^{-1}} } \tr{\rho_0^E}^{q-1}\\
&=  \tr{\pr{\overline{\rho}^2 \pr{\rho_0^E}^{-1}} }.
\end{split}
\end{multline}
Therefore, we obtain
\begin{multline}
\frac{1}{q^2} \sum_{j=1}^q \sum_{\widetilde{j}=1}^q \tr{\rho(j) \rho(\widetilde{j})  \pr{\pr{\rho_0^{E}}^{\proddist q}}^{-1}} - 1 \\
\begin{split}
&= \frac{1}{q^2} \pr{q(q-1) + q  \tr{\pr{\overline{\rho}^2 \pr{\rho_0^E}^{-1}} }} - 1\\
&= \frac{1}{q}\pr{\tr{\pr{\overline{\rho}^2 \pr{\rho_0^E}^{-1}} } - 1}\\
&\leq \frac{1}{q}\pr{\frac{\dim \calH^E}{\widetilde{\lambda}^E} - 1}
\end{split}
\end{multline}

\textbf{Error analysis}: To prove \eqref{eq:H0-prob} and \eqref{eq:H1-prob}, it is enough to show that
\begin{multline}
\mathbb{P}\left(|\lambda_{\min}({\chi}) - \lambda_{\min}(\widehat{\chi})| \leq \tau\right.\\
 \left.\textnormal{ and } |\lambda_{\min}(\widehat{\calE}(\rho_0^A)) - \lambda_{\min}({\calE}(\rho_0^{\widetilde{A}}))| \leq \tau\right) \geq 1-2^{-\xi \ell}.
\end{multline}
To this end, note that
\begin{align}
\P{|\lambda_{\min}({\chi}) - \lambda_{\min}(\widehat{\chi})| \leq \tau} 
&\stackrel{(a)}{\geq}  \P{{{\norm{\chi - \widehat{\chi}}}_1 \leq \tau} }\\
&\stackrel{(b)}{\geq}  \P{{\sum_{j,k}|\chi_{j,k} -\widehat{\chi}_{j,k}| \leq \tau} },
\end{align}
where $(a)$ follows from \cite[Lemma 11.1]{Petz2008}, and $(b)$ follows from the triangle inequality. By \eqref{eq:calE-chi}, we also have
\begin{multline}
\P{ |\lambda_{\min}(\widehat{\calE}(\rho_0^{\widetilde{A}})) - \lambda_{\min}({\calE}(\rho_0^{\widetilde{A}}))| \leq \tau}\\
 \leq  \P{{\sum_{j,k}|\chi_{j,k} -\widehat{\chi}_{j,k}| \leq \tau} }.
\end{multline}
 Using \eqref{eq:chi-bound}, we thus obtain
 \begin{multline}
 \mathbb{P}\left(|\lambda_{\min}({\chi}) - \lambda_{\min}(\widehat{\chi})| \leq \tau\right.\\
 \left. \textnormal{ and } |\lambda_{\min}(\widehat{\calE}(\rho_0^{\widetilde{A}})) - \lambda_{\min}({\calE}(\rho_0^{\widetilde{A}}))| \leq \tau\right)\\
  \geq 1 - 16d^4e^{-\frac{1}{256 d^12\kappa_{\max}^2}\widetilde{\ell}\tau^2}.
 \end{multline}
 We now establish bounds on the accuracy of the estimates $\widehat{D}^B$ and $\widehat{D}^E$ when $\lambda_{\min}(\chi)\geq\widetilde{\lambda}^{\chi} -2\tau,\lambda_{\min}(\calE(\rho^{\widetilde{A}}_0)) \geq \widetilde{\lambda}^B-2\tau$.  We choose $\epsilon>0$ small enough such that
\begin{multline}
 \epsilon \left(\frac{ \log(d-1)}{2}  + d\log \frac{1}{\min(\widetilde{\lambda}^B-2\tau, \widetilde{\lambda}^E)} \right.\\
 \left.+ \frac{d^2}{\min(\widetilde{\lambda}^B-2\tau, \widetilde{\lambda}^E) - \epsilon}\right) +\Hb{\frac{\epsilon}{2}} \leq \tau.
\end{multline}
By Lemma~\ref{lm:chi-lambda-max}, we can choose $\gamma>0$ independent of $\lambda_{\max}(\chi)$ such that
\begin{align}
d^6\kappa_{\max} \gamma &\leq \epsilon,\\
\frac{d^{18} \kappa_{\max}^2 \lambda_{\max}(\widetilde{\rho})\sqrt{\lambda_{\max}(\chi)} \gamma^2 }{2\pr{\lambda_{\min}(\chi) - 2\tau  - d^{5}\kappa_{\max} \gamma}}  &\leq \epsilon.
\end{align}
By Lemma~\ref{lm:approx-states}  and Lemma~\ref{lm:div-l1-bound2}, we have
\begin{multline}
\mathbb{P}\left(D^B(\calE)- 2\tau \leq \widehat{D}^B \leq D^B(\calE)\right.\\
\left.D^E(\calE) \leq \widehat{D}^E \leq D^E(\calE) + 2\tau\right) \leq   32 d^4 e^{-\frac{1}{256} \widetilde{\ell}\gamma^2}.
\end{multline}
Since $\widetilde{\ell} \geq \frac{\frac{\ell}{2d^2} - 1}{2d^2}-1$, we can choose $\xi>0$ small enough such that the above upper-bound is less than $2^{-\xi \ell}$.
\end{proof}
\begin{proof}[Proof of Lemma~\ref{lm:cascade-analysis}]
We only prove the second part of the lemma and the proof of the second part can be obtained by the exact same approach. Let $P_e(D^B, D^E)$, $S(D^B, D^E)$, $C(D^B, D^E)$ indicate the probability of error, secrecy, and covertness of the protocol discussed in the proof of Theorem~\ref{th:universal-covert}, respectively, when we use the parameters $D^B$ and $D^E$. By  the law of total probability, the probability of error of the overall protocol is
\begin{multline}
\E[\widehat{D}^B\widehat{D}^E]{P_e(\widehat{D}^B, \widehat{D}^E)} \\
\begin{split}
&= \E{P_e(\widehat{D}^B, \widehat{D}^E)|\calA}\P{\calA}\\
&\phantom{=} + \E{P_e(\widehat{D}^B, \widehat{D}^E)| \calA^c}\P{\calA^c}\\
&\stackrel{(a)}{\leq} 2T^{-5} + \E{P_e(\widehat{D}^B, \widehat{D}^E)| \calA^c}\P{\calA^c}\\
&{\leq} 2T^{-5} + \epsilon,
\end{split}
\end{multline}
where $\calA \eqdef \{ \widehat{D}^B \leq D^B(\calE), \widehat{D}^E \geq D^E(\calE)\} \cup \{H=0\}$, $(a)$ follows from Theorem~\ref{th:universal-covert}. For the secrecy, first note that the estimation phase does not leak any information about the key. Furthermore, by convexity of the quantum relative entropy, we have
\begin{align}
S
&\leq \E[\widehat{D}^B\widehat{D}^E]{S(\widehat{D}^B, \widehat{D}^E)} \\
&= \E{S(\widehat{D}^B, \widehat{D}^E)|\calA}\P{\calA} + \E{S(\widehat{D}^B, \widehat{D}^E)| \calA^c}\P{\calA^c}\\
&\stackrel{(a)}{\leq} L_1T^{-4} + \E{S(\widehat{D}^B, \widehat{D}^E)| \calA^c}\P{\calA^c}\\
&\stackrel{(b)}{\leq} L_1T^{-4} +\pr{T\log \frac{1}{\widetilde{\lambda}^E} + \ell^{\max}} \epsilon,
\end{align}
where $(a)$ follows from Theorem~\ref{th:universal-covert}, and $(b)$ follows from the upper-bound $S\leq T\log \frac{1}{\widetilde{\lambda}^E} + \ell^{\max}$. Finally, for covertness, since the estimation and transmission phases are independent, we have
\begin{align}
C \leq \delta + \E[\widehat{D}^B\widehat{D}^E]{C(\widehat{D}^B, \widehat{D}^E)}.
\end{align}
Similar to secrecy, we also have
\begin{multline}
\E[\widehat{D}^B\widehat{D}^E]{C(\widehat{D}^B, \widehat{D}^E)} \leq   \frac{\alpha_T^2\chi_2(\rho_1^E(\theta))\|\rho_0^E(\theta)}{2}T \\+ L_2\alpha_T^3 T+ L_1T^{-4} 
+ 2\sqrt{L_1}\log \frac{2}{\widetilde{\lambda}^E}T^{-1} + \epsilon T \log \frac{1}{\widetilde{\lambda}^E}.
\end{multline}
\end{proof}
\subsection{Proof of Theorem~\ref{th:main-collective}}
We describe a protocol running over $\widetilde{T} > 0$ channel uses. Let $T' = \lfloor  \sqrt{\widetilde{T}}\rfloor$ and $T = \widetilde{T} - T' - O(\log T')$. Alice and Bob use the first $T' + O(\log T')$ channel uses for the estimation protocol described in Section~\ref{sec:link-covert-estim} for parameters $q$ and $\ell$ to obtain $H$ as well as estimates $D^B(\calE)$ and $D^E(\calE)$. If $H=0$ the protocol is aborted and if $H=1$, the rest of $T$ channel uses will be used for transmission using the universal protocol as described before for $\widehat{D}^B$, $\widehat{D}^E$, $\widetilde{\lambda}^B - 2\tau$ and $\widetilde{\lambda}^W$.  For a channel satisfying $\lambda_{\min}(\chi)\geq\widetilde{\lambda}^{\chi} -2\tau,\lambda_{\min}(\calE(\rho^{\widetilde{A}}_0)) \geq \widetilde{\lambda}^B-2\tau$, by applying the second part of  Lemma~\ref{lm:cascade-analysis} and Theorem~\ref{th:estimation-analysis}, for some $\xi>0$, we have
\begin{align}
 P_e& \leq 2T^{-5} + 2^{-\xi \ell},\\
 S &\leq L_1T^{-4} + 2^{-\xi\ell} \pr{T \log \frac{1}{\widetilde{\lambda}^E} + \ell^{\max}},\\
 C &\leq  \frac{\alpha_T^2\chi_2(\rho_1^E(\theta))\|\rho_0^E(\theta)}{2}T + L_2\alpha_T^3T + L_1T^{-4}\nonumber\\
 &\phantom{=} + 2\sqrt{L_1}\log \frac{2}{\widetilde{\lambda}^E}T^{-1} + 2^{-\xi\ell} T \log \frac{1}{\widetilde{\lambda}^E}\nonumber\\
 &\phantom{=} + \frac{\ell}{q}(\frac{\dim \calH^E}{\widetilde{\lambda}^E} - 1).
\end{align}
One can check that if $\ell \in \omega(\log T) \cap o\pr{\alpha_T T^{-\frac{3}{4}}}$, which is non-empty by definition of $\alpha_T$, we can always find the sequence $\epsilon_{\widetilde{T}}$ satisfying the conditions in Theorem~\ref{th:main-collective}. If the channel satisfies  $\lambda_{\min}(\chi)\geq\widetilde{\lambda}^{\chi} ,\lambda_{\min}(\calE(\rho^{\widetilde{A}}_0)) \geq \widetilde{\lambda}^B$, by \eqref{eq:H0-prob}, with probability $2^{-\xi \ell}$, the number of  transmitted bits is lower-bounded by
\begin{align}
(1-2\zeta) (D^B(\calE) - D^E(\calE) - 2\tau) \alpha_T T.
\end{align}
If the channel does not satisfy $\lambda_{\min}(\chi)\geq\widetilde{\lambda}^{\chi} ,\lambda_{\min}(\calE(\rho^{\widetilde{A}}_0)) \geq \widetilde{\lambda}^B$, by \eqref{eq:H1-prob} and the first part of Lemma~\ref{lm:cascade-analysis}, we have
\begin{align}
P_e& \leq  2^{-\xi \ell},\\
 S &\leq 2^{-\xi\ell} \pr{T \log \frac{1}{\widetilde{\lambda}^E} + \ell^{\max}},\\
 C &\leq  2^{-\xi\ell} T \log \frac{1}{\widetilde{\lambda}^E} + \frac{\ell}{q}(\frac{\dim \calH^E}{\widetilde{\lambda}^E} - 1),
\end{align}
but no key is generated.

\section{Error exponent calculations} 
\label{sec:error-exponent}
\begin{proof}[Proof of Lemma~\ref{lm:exponent-res}]
For a fix $T$, applying Taylor's theorem on $\phi$ defined in~\eqref{eq:phi-res-def}, we have
\begin{align}
\label{eq:taylor_phi}
\phi(s) = \phi(0) + \phi'(0)s + \frac{\phi''(0)}{2}s^2 + \frac{\phi'''(\eta)}{6}s^3,
\end{align}
for some $s \leq \eta \leq 0$. To compute derivatives of $\phi$, let us define
\begin{align}
A_y(s) &\eqdef \pr{{\widetilde{\rho}_y}^E}^{1-s} \pr{\widetilde{\rho}^E}^s, \\
g(s) &\eqdef \sum_y Q_Y(y) \tr{A_y(s)}.
\end{align}
One can check that $\phi(s) = \log g(s)$. Hence, we obtain
\begin{align}
\phi'(s) &= \frac{g'(s)}{g(s)},\\
\phi''(s) &= \frac{g''(s)}{g(s)} -\pr{\frac{g'(s)}{g(s)}}^2,\\
\phi'''(s) &= \frac{g'''(s)}{g(s)} -3\frac{g'(s)g''(s)}{g^2(s)} +2\pr{\frac{g'(s)}{g(s)}}^3.
\end{align}
Moreover, since $A_y'(s) = -\ln\pr{\widetilde{\rho}_y^E} A_y(s) + A_y(s) \ln\pr{\widetilde{\rho}^E}$, we have
\begin{multline}
g'(s) =  \sum_y Q_Y(y) \tr{ -\ln\pr{\widetilde{\rho}_y^E} A_y(s) + A_y(s) \ln\pr{\widetilde{\rho}^E}},
\end{multline}
\begin{multline}
g''(s) =  \sum_y Q_Y(y) \textnormal{tr}\left(\pr{\ln\pr{\widetilde{\rho}_y^E}}^2 A_y(s)\right. \\
\left.  -2\ln\pr{\widetilde{\rho}_y^E} A_y(s) \ln\pr{\widetilde{\rho}^E}+ A_y(s) \pr{\ln\pr{\widetilde{\rho}^E}}^2\right),
\end{multline}
and
\begin{multline}
\nonumber g'''(s) =  \sum_y Q_Y(y)\textnormal{tr}\left( -\pr{\ln\pr{\widetilde{\rho}_y^E}}^3 A_y(s)\right.\\
\left. +  3\pr{\ln\pr{\widetilde{\rho}_y^Z} }^2A_y(s) \ln\pr{\widetilde{\rho}^E} \right.\\
\left.- 3\ln\pr{\widetilde{\rho}_y^E} A_y(s) \pr{\ln\pr{\widetilde{\rho}^E}}^2+A_y(s) \pr{\ln\pr{\widetilde{\rho}^E}}^3\right).
\end{multline}

Using $A_y(0) = \widetilde{\rho}^E_y$ combined with the above expressions, we obtain
\begin{align}
g(0) &= \sum_y Q_Y(y) \tr{ \widetilde{\rho}^E_y} = 1,\\
g'(0) &= \sum_y Q_Y(y) \tr{ -\ln\pr{\widetilde{\rho}_y^E} \widetilde{\rho}^E_y+\widetilde{\rho}^E_y \ln\pr{\widetilde{\rho}^E}}\\
& =- I(Q_Y, \widetilde{\rho}^E_y ),\\
g''(0) &=  \sum_y Q_Y(y) \textnormal{tr}\left( \pr{\ln\pr{\widetilde{\rho}_y^E}}^2 \widetilde{\rho}^E_y\nonumber\right.\\
&\left.\phantom{=}  -2\ln\pr{\widetilde{\rho}_y^E} \widetilde{\rho}^E_y \ln\pr{\widetilde{\rho}^E}+ \widetilde{\rho}^E_y\pr{\ln\pr{\widetilde{\rho}^E}}^2\right).
\end{align}

Hence, we have
\begin{align}
\phi(0) &= \ln(g(0)) = 0,\\
\phi'(0) &= \frac{g'(0)}{g(0)} =-  I(Q_Y, \widetilde{\rho}^E_y ),\\
\phi''(0) &= \frac{g''(0)}{g(0)} -\pr{\frac{g'(0)}{g(0)}}^2,\\
&=   \sum_y Q_Y(y) \textnormal{tr}\left( \pr{\ln\pr{\widetilde{\rho}_y^E}}^2 \widetilde{\rho}^E_y\right.\nonumber \\
&\phantom{===} \left.-2\ln\pr{\widetilde{\rho}_y^E} \widetilde{\rho}^E_y \ln\pr{\widetilde{\rho}^E}+ \widetilde{\rho}^E_y\pr{\ln\pr{\widetilde{\rho}^E}}^2\right)\nonumber\\
& \phantom{=}- I(Q_Y, \widetilde{\rho}^E_y )^2.
\end{align}

Note that $\phi''(0)$ implicitly depends on $\alpha_T$ the probability that the input is one. Let us define 
\begin{multline}
 h(\alpha) \eqdef \sum_y Q_Y(y) \textnormal{tr}\left( \pr{\ln\pr{\widetilde{\rho}_y^E}}^2 \widetilde{\rho}^E_y \right.\\
 \left.- 2\ln\pr{\widetilde{\rho}_y^E} \widetilde{\rho}^E_y \ln\pr{\widetilde{\rho}^E}+ \widetilde{\rho}^E_y\pr{\ln\pr{\widetilde{\rho}^E}}^2\right)
\end{multline}
 when the input distribution is Bernoulli($\alpha$). One can check that $Q_Y(y)$, $\widetilde{\rho}_y^E$, $\ln(\widetilde{\rho}_y^E)$, and $\ln(\widetilde{\rho}^E)$ are continuously differentiable with respect to $\alpha$, and so is $h$. Moreover, we have

\begin{align}
h(0) 
&=  \sum_y Q_{Y|X}(y|0) \textnormal{tr}\left( \pr{\ln\pr{\widetilde{\rho}_{0,y}^E}}^2 \widetilde{\rho}^E_{0,y}\right.\nonumber \\
 &\phantom{===}\left.-2\ln\pr{\widetilde{\rho}_{0,y}^E} \widetilde{\rho}^E_{0,y} \ln\pr{\widetilde{\rho}^E}+ \widetilde{\rho}^E_{0,y}\pr{\ln\pr{\widetilde{\rho}^E}}^2\right)  \\
&\stackrel{(a)}{=}   \sum_y Q_{Y|X}(y|0) \textnormal{tr}\left( \pr{\ln\pr{\widetilde{\rho}^E}}^2 \widetilde{\rho}^E \right.\nonumber\\
&\phantom{===}\left.-2\ln\pr{\widetilde{\rho}^E} \widetilde{\rho}^E\ln\pr{\widetilde{\rho}^E}+\widetilde{\rho}^E\pr{\ln\pr{\widetilde{\rho}^E}}^2\right) = 0,
\end{align}
where $(a)$ follows from Lemma~\ref{lm:rho_0}. By the mean value theorem, we know that $|h(\alpha) - h(0)| =  |h(\alpha)| = h'(\beta) \alpha$ for some $0 < \beta < \alpha$. Since $h'$ is continuous for a small neighborhood around zero, it is bounded and therefore, we have $|h(\alpha_T)| = O(\alpha_T)$. Furthermore, Lemma~\ref{lm:rho_0} implies that $ I(Q_Y, \widetilde{\rho}^E_y )^2 = O(\alpha_T^2)$. Thus, there exists $B>0$ such that $|\phi''(0)| \leq B\alpha_T$ for $T$ large enough. Notice next that  $g$, $g'$, $g''$, and $g'''$ are jointly continuous functions of both variables $s$ and $\alpha_T$ in a neighborhood around (0, 0). Additionally, since $g(0) = 1$ when $\alpha = 0$, we conclude that $\phi'''$ is also continuous in both $s$ and $\alpha_T$ in a neighborhood around (0, 0). Therefore, for $B$ large enough, $|s|$ small enough and $T$ large enough, we have $|\phi'''(s)| \leq B$. Combing $\phi(0) = 0$, $\phi'(0) = -I(Q_Y, \widetilde{\rho}_y^E)$, $|\phi''(0)| \leq B \alpha_T$, and $|\phi'''(\eta)|\leq B$ with \eqref{eq:taylor_phi}, we obtain the desired result.
\end{proof}
\begin{proof}[Proof of Lemma~\ref{lm:exponent-reliability}]
Consider any cq-channel $x\mapsto \rho^B_x$ with $\lambda_{\min}(\rho^B_0) = \lambda_{\min} > 0$. We first show that the corresponding function $\phi$ is smooth enough to use Taylor theorem. Let us define
\begin{align}
A(s, p) &\eqdef ((1-p)\pr{\rho_0^B}^{1-s} + p \pr{\rho_0^B}^{1-s}, s)\\
g(M, s) &\eqdef (\tr{M^{\frac{1}{1-s}}}, s)\\
\psi(x, s) &\eqdef -(1-s)\log(x).
\end{align}
By definition, we have $\phi(s, p) = (\psi \circ g \circ A)(s, p)$. Additionally, all these three functions are from a subset of a Banach space to a Banach space, which means that we can consider their Fr\'echet derivative. In the following lemma, we show that they are infinitely many times differentiable.
\begin{lemma}
The functions $A$, $g$, and $\psi$ are infinitely many times differentiable on 
\begin{align}
&[0, 1[\times [0, 1[,\\
&\{M\in \calL(\calH): M \text{ is Hermitian}, M\succ 0\} \times [0, 1[,\\
&[0, 1[\times [0, \infty[,
\end{align}
respectively \footnote{For the boundary points we consider the one-sided derivative.}.
\end{lemma}
\begin{proof}
We investigate each function separately.
\begin{itemize}
\item \textbf{Differentiability of $A$}: It is enough to check the differentiability of $A_1(s, p)\eqdef (1-p)\pr{\rho_0^B}^{1-s} + p \pr{\rho_1^B}^{1-s}$. We shall provide explicit expressions for all partial derivatives of $A_1$ to any order. For any Hermitian operator $\rho\in\calL(\calH)$ with $\rho\succeq 0$ and $\rho \neq 0$, let $\rho = \sum_{e } \lambda_e \ket{e}\bra{e}$ be an eigen-decomposition for $\rho$. We define $\log \rho\eqdef \sum_{e : \lambda_e \neq 0} \log(\lambda_e)\ket{e}\bra{e}$ which is different from the usual definition since we disregard the zero eigenvalues. With this definition, one can check that for any $i\geq 1$, we have
\begin{align}
\frac{d^i}{ds^i} (\rho^{1-s}) =\rho^{1-s} \pr{-\log \rho}^i.
\end{align}
Hence, using the linearity of Fr\'echet derivative, if we take $i$ partial derivatives with respect to $s$ and $j$ partial derivatives with respect to $p$ at any order, the result is
\begin{align}
\begin{cases}
(1-p)\pr{\rho_0^B}^{1-s} \pr{-\log \rho_0^B}^i\quad& j = 0,\\
+ p \pr{\rho_1^B}^{1-s}\pr{-\log \rho_1^B}^i\\
-\pr{\rho_0^B}^{1-s} \pr{-\log \rho_0^B}^i+ \pr{\rho_1^B}^{1-s}\pr{-\log \rho_1^B}^i\quad& j = 1,\\
0\quad &j \geq 2.
\end{cases}
\end{align}
This also means that all partial derivative are differentiable and therefore continuous. Accordingly, $A_1$ is infinitely many times Fr\'echet differentiable.
\item \textbf{Differentiability of $g$}: Again we only check the differentiability of $g_1(M, s) \eqdef \tr{M^{\frac{1}{1-s}}}$. In this case, it is more challenging to obtain a closed-form expression for partial derivatives. However, we will prove that any partial derivative is a multilinear form mapping $(K_1, \cdots, K_m)\in \calL(\calH)^m$ to $\mathbb{R}$ and is a summation of terms of the form
\begin{align}
\label{eq:derivative-g-term}
\frac{p(s)}{(1-s)^i} \tr{K_1\cdots K_m M^{\frac{q(s)}{(1-s)^j}}(\log M)^k},
\end{align}
where $q$ and $p$ are polynomial in $s$, and $i$, $j$, and $k$ are non-negative integers. Using induction on the total number of partial derivative taken and linearity of the derivative, it is enough to show that if we take the derivative of \eqref{eq:derivative-g-term} with respect to $s$ or $M$, we would have an expression that is a summation of term of the same form. Applying the rules of differentiation, one can check that
\begin{multline}
\frac{\partial }{\partial s} \pr{\frac{p(s)}{(1-s)^i} \tr{K_1\cdots K_m M^{\frac{q(s)}{(1-s)^j}}(\log M)^k}} 
\\ =\frac{p(s)\pr{jq(s) +(1-s) q'(s)}}{(1-s)^{i+j+1}}\\
\times \tr{K_1\cdots K_m M^{\frac{q(s)}{(1-s)^j}}(\log M)^{k+1}}\\ + \frac{ip(s) +(1-s)p'(s)}{(1-s)^{i+1}}\tr{K_1\cdots K_m M^{\frac{q(s)}{(1-s)^j}}(\log M)^k},
\end{multline}
and 
\begin{multline}
\frac{\partial }{\partial M} \pr{\frac{p(s)}{(1-s)^i} \tr{K_1\cdots K_m M^{\frac{q(s)}{(1-s)^j}}(\log M)^k}}\\
 =K\mapsto \frac{p(s)}{(1-s)^i}\frac{q(s)}{(1-s)^j}\\
  \textnormal{tr}\left(K K_1\cdots K_m \left(\frac{q(s)}{(1-s)^j}M^{\frac{q(s)}{(1-s)^j} -1}(\log M)^k\right.\right.\\
  \left.\left. + k M^{\frac{q(s)}{(1-s)^j} -1} \pr{\log M} ^{k-1} \right)\right).
\end{multline}
Therefore, $g_1$ has partial derivatives of any order. Using the same argument that we used for $A_1$, we conclude that $g_1$ is infinitely many Fr\'echet differentiable.
\item  \textbf{Differentiability of $\psi$}: $\psi$ is product of two smooth functions $(x, s) \mapsto -(1-s)$ and $(x, s) \mapsto \log x$, and therefore, it is smooth on its domain.
\end{itemize}
\end{proof}
We next check that $A(s, p)$ lies in the $\{M\in \calL(\calH): M \text{ is Hermitian}, M\succ 0\} $ where $g$ is differentiable. By our assumption that $\lambda_{\min}> 0$, $\rho_0^B$ is positive semi-definite, and so is $\pr{\rho_0^B}^{1-s}$ for $s\in[0, 1[$. Furthermore, since $\rho_1^B \succeq 0$, we have $A(s, p) \succ 0$ for all $(s, p)\in[0, 1[\times [0, 1[$. Thus, by chain rule, $\phi$ is a smooth function on $[0, 1[\times [0, 1[$. Apply Taylor theorem, we have
\begin{multline}
\phi(s, p) \\= \phi(0, p) + \frac{\partial \phi(0, p)}{\partial s}s + \frac{1}{2} \frac{\partial^2 \phi(0, p)}{\partial^2 s}s^2 + \frac{1}{6} \frac{\partial^3 \phi(\eta, p)}{\partial^3 s}s^3, 
\end{multline}
for some $\eta\in[0, s]$ that can depend on $s$. Similarly, we have
\begin{align}
 \frac{\partial^2 \phi(0, p)}{\partial^2 s} =  \frac{\partial^2 \phi(0, 0)}{\partial^2 s} +  \frac{\partial^3 \phi(0, \tau)}{\partial^2 s\partial p}p,
\end{align}
for some $\tau \in [0, p]$.
Additionally, one can check that $A(s, p)$ and all its derivatives depend continuously on $\rho_0^B$ and $\rho_1^B$. Since any continuous function achieves its maximum on a compact domain,  we have
\begin{align}
\sup_{ \tau\in[0, \widetilde{p}], \rho_0^B\in\calD(\calH), \rho_1^B\in\calD(\calH): \lambda_{\min}(\rho_0^B) \geq \widetilde{\lambda}}\left|\frac{\partial^3 \phi(0, \tau)}{\partial^2 s\partial p} \right| &< \infty,\\
\sup_{\eta\in[0, \widetilde{s}], p\in[0, \widetilde{p}], \rho_0^B\in\calD(\calH), \rho_1^B\in\calD(\calH): \lambda_{\min}(\rho_0^B) \geq \widetilde{\lambda}}\left|\frac{\partial^3 \phi(\eta, p)}{\partial^3 s} \right| &< \infty.
\end{align}
Moreover, from the definition and  some calculations, $\phi(0, p) =0$,  $\frac{\partial^2 \phi(0, 0)}{\partial^2 s} = 0$, and by \cite{hayashi2009universal}, $\frac{\partial \phi(0, p)}{\partial s} = I(p)$. It implies that there exists $B>0$, such that for all cq-channels $x\mapsto \rho_x^B$ with $\lambda_{\min}(\rho_0^B)\geq \widetilde{\lambda}$, we have
\begin{align}
\phi(s, p) \geq  I(p)s -B\pr{ps^2 + s^3}.
\end{align}
Furthermore, using same approach, we can prove $I(p) \geq p\D{\rho_1^B}{\rho_0^B} - Bp^2$.
\end{proof}

\begin{proof}[Proof of Lemma~\ref{lm:exponent-res2}]
If we define
\begin{align}
A(s, p) &\eqdef \pr{(1-p)\pr{\rho_0^E}^{1-s} p \pr{\rho_1^E}^{1-s}}\pr{(1-p)\rho_0^E + p\rho_1^E}^s\\
g(M) &\eqdef \tr{M}\\
\psi(x) &\eqdef \log(x),
\end{align}
similar to the proof of Lemma~\ref{lm:exponent-reliability}, one can check that all these functions are infinitely many times Fr\'echet differentiable. Since, $\phi = \psi \circ g \circ A$, the rest of proof is exactly similar to that of Lemma~\ref{lm:exponent-reliability}.
\end{proof}

\section{Technical lemmas}
\label{sec:tech-lemmas}
\begin{lemma}
\label{lm:lambda-min-max}
Suppose $A$ and $B$ are Hermitian in $\calL(\calH)$. Then, we have
\begin{align}
\lambda_{\min}(A) &\geq \lambda_{\min}(B) - \|A-B\|_2 \geq  \lambda_{\min}(B) - \|A-B\|_1\\
\lambda_{\max}(A) &\leq \lambda_{\max}(B) + \|A-B\|_2 \leq  \lambda_{\max}(B) + \|A-B\|_1
\end{align}
\end{lemma}
\begin{proof}
If $\lambda_{\min}(A) \eqdef \lambda_1 \leq \cdots \leq \lambda_d \eqdef \lambda_{\max}(A)$ and $\lambda_{\min}(B) \eqdef \gamma_1 \leq \cdots \leq \gamma_d \eqdef \lambda_{\max}(B)$ are the eigenvalues of $A$ and $B$, respectively, then by \cite[Corollary 6.3.8]{horn1990matrix}, we have $\|A-B\|_1^2 \geq \|A-B\|_2^2 \geq \sum_{i=1}^d (\lambda_i - \gamma_i)^2$ which results in the desired bounds.
\end{proof}

\begin{lemma}
\label{lm:chi-lambda-max}
For any quantum channel $\calE: \calL(\calH^A)\to\calL(\calH^A)$ with $\chi$-representation matrix $\chi$, we have $\lambda_{\max}(\chi)\leq \sqrt{d}$, where $d\eqdef \dim(\calH^A)$.
\end{lemma}
\begin{proof}
Since $\chi$ is Hermitian, it admits to an eigen-decomposition representation, i.e., for some unitary matrix $U$ and real values $\Lambda_1, \cdots, \Lambda_{d^2}$, we have $\chi_{i,j} = \sum_{k=1}^{d^2}d_iU_{i, k} U_{j, k}^*$. By \cite[Eq. (8.168)]{nielsen2002quantum}, $\calE$ has a Kraus representation $\calE(\rho) = \sum_{i=1}^{d^2}E_i\rho E_i^\dagger$ for $E_i =\sqrt{\Lambda_i} \sum_{j=1}^{d^2}U_{j,i} \widetilde{E}_j$. We hence have
\begin{align}
{\norm{E_i}}_2 
&= {\sqrt{\Lambda_i}} {\norm{\sum_{j=1}^{d^2}U_{j,i} \widetilde{E}_j}}_2\\
&=  {\sqrt{\Lambda_i}}  \sqrt{\tr{\pr{\sum_{j=1}^{d^2}U_{j,i}^* \widetilde{E}_j^\dagger}\pr{\sum_{j=1}^{d^2}U_{j,i} \widetilde{E}_j}}}\\
&= {\sqrt{\Lambda_i}} \sqrt{\sum_{j=1}^{d^2}\sum_{j'=1}^{d^2}U_{j, i}^*U_{j', i} \tr{E_{j}^\dagger E_{j'}}}\\
&=  {\sqrt{\Lambda_i}} \sqrt{\sum_{j=1}^{d^2}U_{j, i}^*U_{j, i}}\\
&\stackrel{(a)}{=}  {\sqrt{\Lambda_i}},\label{eq:norm-kraus-eigen}
\end{align}
where $(a)$ follows since $U$ is unitary. Because $\calE$ is a quantum channel, we have $\sum_{i=1}^{d^2}E_i^\dagger E_i = I$. Taking the  trace from this equality, we obtain that 
\begin{align}
d = \tr{I} = \tr{\sum_{i=1}^{d^2}E_i^\dagger E_i } = \sum_{i=1}^{d^2}{\norm{E_i}}_2^2.\label{eq:sum-norm-kraus}
\end{align}
Using \eqref{eq:norm-kraus-eigen} and \eqref{eq:sum-norm-kraus}, we conclude that
\begin{align}
\lambda_{\max}(\chi) = \max_{i\in\intseq{1}{d^2}} \Lambda_i \leq {\norm{E_i}}_2 \leq \sqrt{d}.
\end{align}
\end{proof}

\begin{lemma}
Consider any quantum channel $\calE:\calL(\calH) \to \calL(\calH)$ with $\dim \calH = d$ and characterized by $\calE(\rho) = \sum_{i, j}\widetilde{E}_i \rho \widetilde{E}_j^\dagger \chi_{ij}$. Define another Hilbert space $\calH^\dagger$ spanned by an orthonormal basis $\{\ket{j}: j\in\intseq{1}{d^2}\}$. Then, up to a unitary transformation, the complementary channel $\calE^\dagger:\calL(\calH) \to \calL(\calH^\dagger)$ would be
\begin{align}
\calE^{\dagger}(\rho) = \sqrt{\chi}^* \widetilde{\rho} \sqrt{\chi}^*,
\end{align}
where
\begin{align}
\chi &\eqdef \sum_{j, k} \ket{j}\bra{k} \chi_{jk}\\
\widetilde{\rho} &\eqdef \sum_{j, k} \ket{j}\bra{k} \tr{\widetilde{E}_j \rho \widetilde{E}_k^\dagger}.
\end{align}
\end{lemma}

\begin{proof}
By \cite{nielsen2002quantum}, without loss of generality we can assume that $\chi$ is Hermitian. Therefore, let $\chi = \sum_{j} d_j \ket{u_j}\bra{u_j}$ be an eigen-decomposition of $\chi$. For $E_j \eqdef \sum_k \sqrt{d_j} \braket{k}{u_j}\widetilde{E}_k$, we have
\begin{multline}
\sum_{j} E_j \rho E_j^\dagger\\
\begin{split}
&= \sum_j \pr{ \sum_{k} \sqrt{d_j} \braket{k}{u_j}\widetilde{E}_k} \rho \pr{\sum_{k'} \sqrt{d_j} \braket{u_j}{k'} \widetilde{E}_{k'}^\dagger}\\
&=\sum_{k}  \sum_{k'} \sum_j \widetilde{E}_k\rho \widetilde{E}_{k'}^\dagger d_j \braket{k}{u_j}  \braket{u_j}{k'}\\
&=\sum_{k}  \sum_{k'}\widetilde{E}_k\rho \widetilde{E}_{k'}^\dagger \bra{k}\pr{ \sum_j d_j\ket{u_j}  \bra{u_j}}\ket{k'}\\
&=\sum_{k}  \sum_{k'}\widetilde{E}_k\rho \widetilde{E}_{k'}^\dagger \bra{k}\chi\ket{k'}\\
&=\sum_{k}  \sum_{k'}\widetilde{E}_k\rho \widetilde{E}_{k'}^\dagger \chi_{kk'}\\
&= \calE(\rho).
\end{split}
\end{multline}
This implies that $\sum_{j} E_j \rho E_j^\dagger$ is a Kraus representation for $\calE$, and therefore, by \cite{wilde2013quantum}, a representation for the complementary channel is
=\begin{align}
\widetilde{\calE}^\dagger(\rho) = \sum_{j,k} \tr{E_j \rho E_k^\dagger}\ket{j}\bra{k}.
\end{align}
Hence, it is enough to show that for some unitary operator $U$ onto $\calH^\dagger$, we have
\begin{align}
 \sqrt{\chi}^* \widetilde{\rho} \sqrt{\chi}^*= U \widetilde{\calE}^\dagger(\rho) U^\dagger.
\end{align}
Let $U\eqdef \sum_j \ket{\widetilde{u}_j}\bra{j}$ where  $\ket{\widetilde{u}_j} \eqdef \sum_i \braket{u_j}{i}\ket{i}$. One can check that it is a unitary operator, and we have
\begin{align}
U \widetilde{\calE}^\dagger(\rho) U^\dagger 
&= \pr{\sum_j \ket{\widetilde{u}_j}\bra{j}} \pr{\sum_{k, k'}  \tr{E_k \rho E_{k'}^\dagger}\ket{k}\bra{k'}}\nonumber\\
&\phantom{=}\times \pr{\sum_{j'} \ket{{j'}}\bra{\widetilde{u}_{j'}}}\\\displaybreak[0]
&= \sum_{jj'kk'} \tr{E_k \rho E_{k'}^\dagger} \ket{\widetilde{u}_j}\bra{j}   \ket{k}\bra{k'}\ket{{j'}}\bra{\widetilde{u}_{j'}}\\
&=\sum_{kk'} \tr{E_k \rho E_{k'}^\dagger} \ket{\widetilde{u}_k}\bra{\widetilde{u}_{k'}}\\\displaybreak[0]
&= \sum_{kk'} \textnormal{tr}\left(\pr{\sum_j \sqrt{d_k} \braket{j}{u_k}\widetilde{E}_j}\rho\nonumber\right.\\
&\phantom{===} \times\left.\pr{\sum_{j'} \sqrt{d_{k'}} \braket{u_{k'}}{j'}\widetilde{E}_{j'}^\dagger}\right) \ket{\widetilde{u}_k}\bra{\widetilde{u}_{k'}}\\\displaybreak[0]
&= \sum_{jj'kk'} \sqrt{d_k}\sqrt{d_{k'}}  \tr{\braket{j}{u_k}\braket{u_{k'}}{j'}\widetilde{E}_j\rho   \widetilde{E}_{j'}^\dagger}\nonumber\\
&\phantom{====}\times\ket{\widetilde{u}_k}\bra{\widetilde{u}_{k'}}\\\displaybreak[0]
&= \sum_{jj'}  \tr{\widetilde{E}_j\rho   \widetilde{E}_{j'}^\dagger}\sum_{kk'}\sqrt{d_k}\sqrt{d_{k'}}  \braket{j}{u_k}\braket{u_{k'}}{j'}\nonumber\\
&\phantom{===}\times\ket{\widetilde{u}_k}\bra{\widetilde{u}_{k'}}\\\displaybreak[0]
& = \sum_{jj'}  \tr{\widetilde{E}_j\rho   \widetilde{E}_{j'}^\dagger} \pr{\sum_{k}\sqrt{d_k} \braket{j}{u_k} \ket{\widetilde{u}_k}}\nonumber\\
&\phantom{===}\times\pr{\sum_{k'}\sqrt{d_{k'}}\braket{u_{k'}}{j'}\bra{\widetilde{u}_{k'}}}\\\displaybreak[0]
& = \sum_{jj'}  \tr{\widetilde{E}_j\rho   \widetilde{E}_{j'}^\dagger} \pr{\sum_{k}\sqrt{d_k} \braket{\widetilde{u}_k}{j} \ket{\widetilde{u}_k}}\nonumber\\
&\phantom{===}\times \pr{\sum_{k'}\sqrt{d_{k'}}\braket{j'}{\widetilde{u}_{k'}}\bra{\widetilde{u}_{k'}}}\\\displaybreak[0]
& =  \pr{\sum_{k}\sqrt{d_k} \ket{\widetilde{u}_k} \bra{\widetilde{u}_k}}\pr{\sum_{jj'}  \tr{\widetilde{E}_j\rho   \widetilde{E}_{j'}^\dagger} \ket{j}\bra{j'}}\nonumber\\
&\phantom{=}\times\pr{\sum_{k'}\sqrt{d_{k'}}\ket{\widetilde{u}_{k'}}\bra{\widetilde{u}_{k'}}}\\\displaybreak[0]
&= \sqrt{\chi}^* \widetilde{\rho} \sqrt{\chi}^*
\end{align}

\end{proof}

\begin{lemma}
\label{lm:product-l1}
Let $A, B\in\calL(\calH)$ and $B$ be Hermitian. Then,
\begin{align}
\|AB\|_1 \leq  \sigma_{\max}(A) \|B\|_1,
\end{align}
where $\sigma_{\max}(A)$ is the maximum singular value of the $A$.
\end{lemma}

\begin{proof}
Consider an eigen-decomposition of $B$, i.e., $B = \sum_bb\ket{b}\bra{b}$. Then,
\begin{align}
\|AB\|_1 
&= \left\| A\pr{\sum_b b \ket{b}\bra{b}}\right\|_1\displaybreak[0]\\
&\leq \sum_b |b| \|A\ket{b}\bra{b}\|_1\displaybreak[0]\\
&\leq \sum_b |b| \tr{\sqrt{\ket{b}\bra{b}A^\dagger A \ket{b}\bra{b}}}\displaybreak[0]\\
&=  \sum_b |b| \sqrt{\bra{b}A^\dagger A \ket{b}}\displaybreak[0]\\
&= \sum_b |b| \|A\ket{b}\|_2\\
&\leq \sigma_{\max}(A) \pr{\sum_{b}|b|}\\
&= \sigma_{\max}(A)\|B\|_1.
\end{align}
\end{proof}

\begin{lemma}
\label{lm:norm-derivative}
Let $\calI\subset \mathbb{R}$ be an interval and $f:\calI \to \mathbb{R}$ and $A(x):\mathbb{R}\to \calL(\calH)$  be differentiable functions such $A(x)$ is Hermitian and its spectrum is included in $\calI$ for all $x$.  For any operator norm $\|\cdot\|$ satisfying $\max(\|PA\|, \norm{AP}) \leq \|A\|$ where $A$ is an arbitrary operator and $P$ is a projection, we have
\begin{multline}
\left\| \frac{d}{dx'} f(A(x'))\bigg|_{x'=x}\right\| \\
\leq d^2 \sup_{\mu\in[\lambda_{\min}(A'(x)), \lambda_{\max}(A'(x))]} |f'(\mu)|  \|A'(x)\|.
\end{multline}
\end{lemma}
\begin{proof}
We use a formula in \cite{wilde2013quantum} for the derivative of an operator-valued function. Let $f:\mathbb{R} \to \mathbb{R}$ and $A(x):\mathbb{R}\to \calL(\calH)$  be a differentiable functions. Then,
\begin{align}
\frac{d}{dx'}f(A(x'))\bigg|_{x'=x} = \sum_{\nu, \eta} f^{[1]}(\nu, \eta) P_{A(x)}(\nu)A'(x) P_{A(x)}(\eta), 
\end{align}
where the summation is taken over all eigenvalues of $A(x)$, $P_{A(x)}(\nu)$ is the projector onto the subspace of all eigenvectors corresponding to $\nu$, and
\begin{align}
 f^{[1]}(\nu, \eta) = \begin{cases} \frac{f(\nu) - f(\eta)}{\nu - \eta}\quad &\nu \neq \eta\\f'(\nu)\quad &\nu = \eta\end{cases}.
\end{align}
We can now upper-bound the norm of $\frac{d}{dx}f(A(x)) $ by
\begin{multline}
\left\|\frac{d}{dx'}f(A(x'))\bigg|_{x'=x}\right\|\\
\begin{split}
&= \left\|\sum_{\nu, \eta} f^{[1]}(\nu, \eta) P_{A(x)}(\nu)A'(x) P_{A(x)}(\eta)\right\|\\
&\leq  \sum_{\nu, \eta} |f^{[1]}(\nu, \eta)| \left\|P_{A(x)}(\nu)A'(x) P_{A(x)}(\eta)\right\|\\
&\stackrel{(a)}{\leq} \sum_{\nu, \eta} |f^{[1]}(\nu, \eta)| \left\|A'(x) \right\|,
\end{split}
\end{multline}
where $(a)$ follows from our assumption that $\max(\|PA\|, \norm{AP}) \leq \|A\|$. By the mean value theorem, we also have that $f^{[1]}(\nu, \eta) = f'(\mu)$ for some $\mu$ between $\nu$ and $\eta$. Thus,
\begin{multline}
\sum_{\nu, \eta} |f^{[1]}(\nu, \eta)| \left\|A'(x) \right\|\\
 \leq d^2 \sup_{\mu\in[\lambda_{\min}(A'(x)), \lambda_{\max}(A'(x))]} |f'(\mu)|  \|A'(x)\|.
\end{multline}
\end{proof}

\begin{lemma}
\label{lm:div-bound-l1}
Suppose $\rho$ and $\sigma$ are two density matrices on Hilbert space $\calH$ with $\dim \calH = d$ such that $\textnormal{supp} \rho \subset \textnormal{supp} \sigma$ and $\|\rho - \sigma\|_1 \leq \epsilon \leq e^{-1}$. Then,
\begin{align}
\D{\rho}{\sigma} \leq \epsilon \log \frac{d}{\lambda_{\min}(\sigma)\epsilon}.
\end{align}
\end{lemma}
\begin{proof}
Since $\textnormal{supp} (\rho) \subset \textnormal{supp} (\sigma)$, we have
\begin{align}
\D{\rho}{\sigma} 
&= \tr{\rho(\log \rho - \log \sigma)}\\
&= -H(\rho) + H(\sigma) - \tr{(\rho - \sigma) \log \sigma}\\
& \stackrel{(a)}{\leq} \epsilon \log \frac{d}{\epsilon} - \tr{(\rho - \sigma) \log \sigma}\\
&\leq \epsilon \log \frac{d}{\epsilon} + \epsilon \log \frac{1}{\lambda_{\min}(\sigma)},
\end{align}
where $(a)$ follows from Fannes inequality.
\end{proof}
\begin{lemma}
\label{lm:div-l1-bound2}
Suppose $\rho, \rho', \sigma, \sigma' \in \calD(\calH)$ with  $\dim \calH = d$, $\textnormal{supp} (\rho) \subset \textnormal{supp} (\sigma)$, and $\textnormal{supp} (\rho') \subset \textnormal{supp} (\sigma')$. Let $\|\rho - \rho'\|_1 \leq \epsilon$, $\|\sigma - \sigma'\|_1 \leq \epsilon$, and $\lambda_{\min}(\sigma)$ be the minimum eigenvalue of $\sigma$ with $\lambda_{\min}(\sigma) \geq \epsilon)$. Then,
\begin{multline}
|\D{\rho}{\sigma} - \D{\rho'}{\sigma'}|\\
 \leq \epsilon \pr{\frac{ \log(d-1)}{2}  + d\log \frac{1}{\lambda_{\min}(\sigma))} + \frac{d^2}{\lambda_{\min}(\sigma) - \epsilon}}\\
  +\Hb{\frac{\epsilon}{2}}.
\end{multline}
\end{lemma}
\begin{proof}
By definition, we have
\begin{multline}
|\D{\rho}{\sigma} - \D{\rho'}{\sigma'}| \\
\begin{split}
&= |-\avgH{\rho} + \avgH{\rho'} - \tr{\rho\log \sigma} + \tr{\rho'\log\sigma'}|\\
&\leq   |-\avgH{\rho} + \avgH{\rho'}| +   |\tr{(\rho-\rho')\log \sigma}|\\
&\phantom{=} + |\tr{\rho'(\log\sigma'-\log\sigma)}.|
\end{split}
\end{multline}
By Fannes inequality, we have
\begin{multline}
 |-\avgH{\rho} + \avgH{\rho'}| \\
  \leq \frac{1}{2}\|\rho - \rho'\|_1 \log(d-1) + \Hb{\frac{1}{2} \|\rho - \rho'\|_1}.
\end{multline}
Furthermore,  Cauchy-Schwartz inequality for Hilbert-Schmidt inner-products implies that
\begin{align}
|\tr{(\rho-\rho')\log \sigma}| 
& \leq \|\rho-\rho'\|_2 \|\log \sigma\|_2\\
&\leq \| \rho - \rho'\|_1 \|\log \sigma\|_2\\
&\leq \| \rho - \rho'\|_1 d \log \frac{1}{\lambda_{\min}(\sigma)}.
\end{align}
Using Cauchy-Schwartz again, we obtain
\begin{align}
|\tr{\rho'(\log\sigma'-\log\sigma)}|
&\leq \|\rho'\|_2 \|\log\sigma'-\log\sigma\|_2\\
&\leq \|\log\sigma'-\log\sigma\|_2.
\end{align}
To upper-bound $\|\log\sigma'-\log\sigma\|_2$, let us define $F(x) \eqdef \log(\sigma + x(\sigma'-\sigma))$ for $t\in[0, 1]$. Then,
\begin{align}
\|\log\sigma'-\log\sigma\|_2\ 
&= \|F(1) - F(0)\|_2\\
&\stackrel{(a)}{\leq} \sup_{x\in[0,1]}\|F'(x)\|_2.
\end{align}
where $(a)$ follows from mean value theorem of multi-variable functions. Applying Lemma~\ref{lm:norm-derivative} for  $f \eqdef \log$ and $A(x) =  \sigma + x(\sigma'-\sigma)$, we obtain 
\begin{align}
\|F'(x)\|_2 &\leq d^2 \sup_{\mu\in[a, b]} |f'(\mu)|  \|A'(x)\|_2\\
&\leq d^2 \frac{1}{\lambda_{\min}(\sigma + x(\sigma'-\sigma))} \|\sigma' - \sigma\|_2\\
&\leq d^2 \frac{1}{\lambda_{\min}(\sigma + x(\sigma'-\sigma))} \|\sigma' - \sigma\|_1.
\end{align}

Finally, for $x\in[0, 1]$, we have
\begin{align}
\lambda_{\min}(\sigma + x(\sigma'-\sigma)) 
&\leq \lambda_{\min}(\sigma) - \|x(\sigma'-\sigma)\|_2\\
&\leq  \lambda_{\min}(\sigma) - \|\sigma'-\sigma\|_1.
\end{align}
\end{proof}

\bibliographystyle{apsrev4-1}
\bibliography{covert}

\end{document}

%% file: arcom_notation.tex
\DeclareMathAlphabet{\eurm}{U}{eur}{m}{n}
\DeclareMathAlphabet{\mathbsf}{OT1}{cmss}{bx}{n}
\DeclareMathAlphabet{\mathssf}{OT1}{cmss}{m}{sl}
\DeclareMathAlphabet{\mathcsf}{OT1}{cmss}{sbc}{n}

\DeclareSymbolFont{bsfletters}{OT1}{cmss}{bx}{n}  
\DeclareSymbolFont{ssfletters}{OT1}{cmss}{m}{n}
\DeclareMathSymbol{\bsfGamma}{0}{bsfletters}{'000}
\DeclareMathSymbol{\ssfGamma}{0}{ssfletters}{'000}
\DeclareMathSymbol{\bsfDelta}{0}{bsfletters}{'001}
\DeclareMathSymbol{\ssfDelta}{0}{ssfletters}{'001}
\DeclareMathSymbol{\bsfTheta}{0}{bsfletters}{'002}
\DeclareMathSymbol{\ssfTheta}{0}{ssfletters}{'002}
\DeclareMathSymbol{\bsfLambda}{0}{bsfletters}{'003}
\DeclareMathSymbol{\ssfLambda}{0}{ssfletters}{'003}
\DeclareMathSymbol{\bsfXi}{0}{bsfletters}{'004}
\DeclareMathSymbol{\ssfXi}{0}{ssfletters}{'004}
\DeclareMathSymbol{\bsfPi}{0}{bsfletters}{'005}
\DeclareMathSymbol{\ssfPi}{0}{ssfletters}{'005}
\DeclareMathSymbol{\bsfSigma}{0}{bsfletters}{'006}
\DeclareMathSymbol{\ssfSigma}{0}{ssfletters}{'006}
\DeclareMathSymbol{\bsfUpsilon}{0}{bsfletters}{'007}
\DeclareMathSymbol{\ssfUpsilon}{0}{ssfletters}{'007}
\DeclareMathSymbol{\bsfPhi}{0}{bsfletters}{'010}
\DeclareMathSymbol{\ssfPhi}{0}{ssfletters}{'010}
\DeclareMathSymbol{\bsfPsi}{0}{bsfletters}{'011}
\DeclareMathSymbol{\ssfPsi}{0}{ssfletters}{'011}
\DeclareMathSymbol{\bsfOmega}{0}{bsfletters}{'012}
\DeclareMathSymbol{\ssfOmega}{0}{ssfletters}{'012}

\newcommand{\calA}{{\mathcal{A}}}

\newcommand{\calC}{{\mathcal{C}}}
\newcommand{\calD}{{\mathcal{D}}}
\newcommand{\calE}{{\mathcal{E}}}

\newcommand{\calH}{{\mathcal{H}}}
\newcommand{\calI}{{\mathcal{I}}}
\newcommand{\calJ}{{\mathcal{J}}}

\newcommand{\calL}{{\mathcal{L}}}

\newcommand{\calN}{{\mathcal{N}}}

\newcommand{\calP}{{\mathcal{P}}}

\newcommand{\calR}{{\mathcal{R}}}

\newcommand{\calS}{{\mathcal{S}}}
\newcommand{\calU}{{\mathcal{U}}}
\newcommand{\calV}{{\mathcal{V}}}
\newcommand{\calX}{{\mathcal{X}}}
\newcommand{\calY}{{\mathcal{Y}}}
\newcommand{\calW}{{\mathcal{W}}}
\newcommand{\calZ}{{\mathcal{Z}}}

%% file: arcom_functions.tex
\newcommand{\E}[2][]{{\mathbb{E}_{#1}}{\left(#2\right)}}       
\renewcommand{\P}[2][]{{\mathbb{P}_{#1}}{\left(#2\right)}}
\newcommand{\Var}[1]{{\text{\textnormal{Var}}{\left(#1\right)}}}       
\newcommand{\D}[2]{{{\mathbb{D}}\!\left({#1\Vert#2}\right)}}
\newcommand{\avgD}[2]{{{\mathbb{D}}\!\left({#1\Vert#2}\right)}}
\newcommand{\V}[1]{{{\mathbb{V}}\!\left(#1\right)}}

\newcommand{\avgI}[1]{{{\mathbb{I}}\!\left(#1\right)}}
\newcommand{\avgH}[1]{{\mathbb{H}}\!\left(#1\right)}

\newcommand{\Hb}[1]{{\mathbb{H}_b}\left(#1\right)}

\newcommand{\wt}[1]{\ensuremath{\mbox{wt}(#1)}}

\newcommand{\card}[1]{\ensuremath{\left|{#1}\right|}}           
\newcommand{\norm}[2][]{\ensuremath{{\left\Vert{#2}\right\Vert}_{#1}}}   
\newcommand{\eqdef}{\ensuremath{\triangleq}}                    
\newcommand{\intseq}[2]{\ensuremath{\llbracket{#1},{#2}\rrbracket}}  
\newcommand{\indic}[1]{\ensuremath{\mathds{1}\!\left\{#1\right\}}}

\renewcommand{\leq}{\leqslant}
\renewcommand{\geq}{\geqslant}

\newcommand{\tr}[1]{\ensuremath{\text{\textnormal{tr}}\left(#1\right)}}  
\renewcommand{\det}[1]{{\left|{#1}\right|}}                              

\newcommand{\proddist}{%
  \mathchoice{\raisebox{1pt}{$\displaystyle\otimes$}}
             {\raisebox{1pt}{$\otimes$}}
             {\raisebox{0.5pt}{\scalebox{0.7}{$\scriptstyle\otimes$}}}
             {\raisebox{0.4pt}{\scalebox{0.6}{$\scriptscriptstyle\otimes$}}}}

\newcommand{\cn}{\textcolor{red}{[\raisebox{-0.2ex}{\tiny\shortstack{citation\\[-1ex]needed}}]}\xspace}